\documentclass[acmsmall,nonacm]{acmart}
\usepackage[utf8]{inputenc}
\pdfoutput=1

\usepackage{amsthm}
\usepackage{mathrsfs}
\usepackage{framed}
\usepackage{hyperref}
\usepackage{subcaption}
\usepackage{algorithm}
\usepackage[noend]{algpseudocode}
\usepackage{xcolor}
\usepackage{mathtools}
\usepackage{thm-restate}
\usepackage{wrapfig}
\usepackage{subcaption}
\usepackage{multicol}
\usepackage{wasysym}
\usepackage{pifont}

\usepackage{tikz}
\usetikzlibrary{positioning, decorations.pathmorphing, matrix}

\newtheorem{theorem}{Theorem}

\newtheorem{lemma}[theorem]{Lemma}

\newtheorem{definition}[theorem]{Definition}

\newtheorem{proposition}[theorem]{Proposition}
\newtheorem{corollary}[theorem]{Corollary}

\definecolor{pinegreen}{HTML}{00A64F}
\definecolor{col1}{HTML}{00A64F}
\definecolor{col2}{HTML}{0071BC}
\definecolor{col3}{HTML}{613F99}
\definecolor{col4}{HTML}{F26035}
\DeclarePairedDelimiter{\parens}{(}{)}

\DeclarePairedDelimiter{\braces}{\{}{\}}
\DeclarePairedDelimiter{\tup}{\langle}{\rangle}
\newcommand{\seqstate}[1]{{\text{state}\parens{#1}}}
\newcommand{\mem}[1]{{\text{mem}\parens{#1}}}

\newcommand{\wildcard}{\_}
\newcommand{\refl}[1]{~\hyperref[#1]{\ref*{#1}L}}
\newcommand{\localind}[1]{\overset{#1}{\sim}}

\newcommand{\ha}[1]{#1}%{{\color{blue}#1}}
\newcommand{\ns}[1]{#1}%{{\color{red}#1}}
\newcommand{\calC}{\mathcal{C}}
\newcommand{\calO}{\mathcal{O}}
\newcommand{\obj}{\mathit{obj}}

\newcommand{\Readop}{o_{read}}
\newcommand{\Changeop}{o_{change}}

\newcommand{\set}[1]{\left\{ #1 \right\} }
\newcommand{\CAS}{\textsc{CAS}}
\newcommand{\LL}{\textsc{LL}}
\newcommand{\SC}{\textsc{SC}}
\newcommand{\RL}{\textsc{RL}}
\newcommand{\VL}{\textsc{VL}}
\newcommand{\Ld}{\textsc{Load}}
\newcommand{\St}{\textsc{Store}}

\newcommand{\enq}{\textsc{Enqueue}}
\newcommand{\deq}{\textsc{Dequeue}}
\newcommand{\peek}{\textsc{Peek}}
\newcommand{\Read}{\textsc{Read}}
\newcommand{\Write}{\textsc{Write}}
\newcommand{\can}[1]{\mathit{can}\parens{#1}}

\newcommand{\flag}{\mathit{flag}}
\newcommand{\lastval}{\mathit{last}\mhyphen{}\mathit{val}}

\mathchardef\mhyphen="2D

\makeatletter
\newcommand{\labeltext}[3][]{%
    \@bsphack%
    \csname phantomsection\endcsname% in case hyperref is used
    \def\tst{#1}%
    \def\labelmarkup{}% How to markup the label itself
    \def\refmarkup{}%
    \ifx\tst\empty\def\@currentlabel{\refmarkup{#2}}{\label{#3}}%
    \else\def\@currentlabel{\refmarkup{#1}}{\label{#3}}\fi%
    \@esphack%
    \labelmarkup{#2}% visible printed text.
}
\makeatother

\newcommand{\remove}[1]{}

\title{History-Independent Concurrent Objects}
\begin{document}

\author{Hagit Attiya}
\affiliation{%
  \institution{Technion-Israel Institute of Technology}
  \country{Israel}
  \city{Haifa}}
\email{hagit@cs.technion.ac.il}

\author{Michael A. Bender}
\affiliation{%
  \institution{Stony Brook University}
  \country{USA}}
\email{bender@cs.stonybrook.edu}

\author{Martin Farach-Colton}
\affiliation{%
  \institution{New York University}
  \country{USA}}
\email{martin@farach-colton.com}

\author{Rotem Oshman}
\affiliation{%
  \institution{Tel-Aviv University}
  \country{Israel}}
\email{roshman@tau.ac.il}

\author{Noa Schiller}
\affiliation{%
  \institution{Tel-Aviv University}
  \country{Israel}}
\email{noaschiller@mail.tau.ac.il}

\begin{abstract}
    A data structure is called \emph{history independent} if its internal memory representation does not reveal the history of operations applied to it, 
%HA: let's save perform for the low level, and use apply for the high level
    only its current state.
    In this paper we study history independence for concurrent data structures, 
    and establish foundational possibility and impossibility results.
    We show that a large class of concurrent objects cannot be implemented from
    smaller base objects in
    a manner that is both wait-free
    and
    history independent;
    but if we settle
    for either lock-freedom instead of wait-freedom
    or for a weak notion of history independence,
    then at least one object in the class,
    multi-valued single-reader single-writer registers, \emph{can}
    be implemented from smaller base objects,
    binary % single-reader single-writer\NNote{not true, they are multi-writer} 
    registers.

    On the other hand, using large base objects,
    we give a strong possibility result in the form of a universal construction:
    an object with $s$ possible states can be implemented
    in a wait-free, history-independent manner
    from compare-and-swap base objects
    that each have $O(s + 2^n)$ possible memory states,
    where $n$ is the number of processes in the system.
    %\NNote{We ignore the dependency in $r$, the number of responses.}
\end{abstract}

\maketitle

\section{Introduction}
\label{sec:intro}

A data structure is said to be \emph{history independent} (HI) if its internal
representation reveals nothing 
about the history of operations that have been applied to it, 
beyond the current state of the data structure.
For example, if a set is history independent, its internal representation may 
(and must, for correctness) reveal the elements that are currently in the set, 
but it must not reveal elements that were previously inserted and then removed.

%\paragraph{Background: a brief history of history independence.}
The notion of history independence was introduced by 
Micciancio~\cite{Micciancio97}, 
who showed how to build a search tree with a history-independent structure.
Naor and Teague~\cite{NaorTe01} formalized two now-classical
notions of history independence:
a data structure is \emph{weakly history independent} (WHI) 
if it leaks no information to an observer
who sees the memory representation once, 
and it is \emph{strongly history independent} (SHI) 
if it leaks no information even to an observer who 
sees the memory representation at multiple points in the execution.
These notions differ significantly: for example, 
a set where each item inserted is stored at a freshly-chosen
random location in memory may be weakly HI
but not strongly HI, because if an item is inserted, removed,
and then inserted again, it may be placed in different locations each time it is inserted;
an observer who sees the memory after each of the two insertions would know that
the item was removed and re-inserted.

History independence has been extensively studied in sequential 
data structures (see below),
and the foundational algorithmic work on history independence has 
found its way into systems like voting machines and storage.
However, history independence was studied only peripherally 
in \emph{concurrent} data structures. 
This paper initiates a thorough study of history-independent 
concurrent data structures,
and establishes fundamental possibility and impossibility results.

%\paragraph{Our contribution: history independence for concurrent data structures.}
%In this paper 
We focus on a concurrent notion of strong history independence
for deterministic data structures,
%\TODO{do we want to say deterministic up here? Any caveats we want to mention?}, 
and characterize the boundaries of what can be achieved.
Defining history independence for concurrent objects is non-trivial, 
because the sequential definition of history independence allows 
the observer to examine the memory only \emph{in-between} 
operations---that is, in a quiescent state---while 
a concurrent implementation might \emph{never be} in a quiescent state.
%Some of our results also hold for randomized implementations. 
%\Hnote{say at some point why we ignore weak history independence}
One of our main questions is whether it should be permissible for 
the observer to open the ``black box'' of a single operation 
and inspect the memory when the system is not quiescent,
and what are the implications of this choice in terms of what can be 
implemented concurrently in a history independent manner.

%We therefore consider three \TODO{two? are we keeping operation-quiescent HI?} definitions:
%the strongest, which we call \emph{perfect history independence}, allows the adversary to inspect the memory at any point in the execution;
%\TODO{operation-quiescent HI? I'd rather not;}
%and finally, \emph{quiescent history independence} allows the adversary to inspect the memory only when the execution is quiescent, that is, no operation is pending, as in the sequential case.
%All of our definitions are with respect to a linearization function, which determines what state the abstract object is in at a given point in the execution.

We begin by asking whether a concurrent object $A$ can be implemented 
out of ``smaller'' base objects of type $B$ 
in a manner that is history independent;
here, ``smaller'' means that $B$ has fewer states than $A$.
Our motivating example is the famous implementation of a 
multi-valued single-writer single-reader wait-free register 
from Boolean registers~\cite{VIDYASANKAR1988287}.
We observe that this implementation is not history independent even 
in the weakest sense (see Section~\ref{section: swsr example}),
and in fact, there is a good reason for this:
we prove that for a fairly general class of objects,
which includes read/write registers,
there is no
wait-free history independent implementation out of smaller base objects, regardless of the type of the smaller objects.
This result holds even if the observer can only inspect the memory 
when there are \emph{no state-changing operations pending} (but read operations may be ongoing).
If the observer can inspect the memory at any point (including while state-changing operations are ongoing),
then even a lock-free implementation is impossible.
On the other hand, for multi-valued single-writer single-reader registers,
there is a wait-free history independent implementation from binary registers,
if we restrict the observer's inspections to points where the system is completely quiescent.
While our results are stated for deterministic algorithms, our impossibility result also applies to randomized implementations of \emph{reversible objects}, which are objects where every state can be reached from every other state (see Section~\ref{sec:prelim}).

Since it is impossible to implement objects out of smaller objects in a wait-free, history independent manner (except possibly in the weakest, quiescent sense), we turn to concurrent implementations where the base object $B$ is large enough to store the full state of the abstract object $A$ that we want to implement.
For this regime we give a strong possibility result, in the form of a \emph{universal implementation}
from compare-and-swap (CAS) objects:
we show that any object $A$ can be implemented in a wait-free, history independent manner from sufficiently-large CAS objects. 
%\Rnote{I think the following sentence is not true, I think it's only true for operations that don't overlap with any currently-running operation, or maybe even only for operations that transitively don't overlap with an operation that overlapped with... a current operation:} 
Our implementation 
reveals nothing about past states of the object or operations that completed
prior to the invocation of any currently-pending operations;
and when no state-changing operation is pending,
the state of the memory reflects only the current abstract state of the object.
Our implementation uses an extended version of an LL/SC object, 
inspired by~\cite{jayanti_et_al:LIPIcs.DISC.2023.25,IsraeliRappoportPODC94}, which we implement from atomic $\CAS$.

\subsubsection*{Additional Related Work.}
%In addition to presenting the definition of strong and weak 
%history independence, 
%Naor and Teague~\cite{NaorTe01} showed that even the bit representation 
%of a data structure (e.g., the memory addresses where the tree nodes are stored) 
%can leak information to an observer. 
%They also showed how to build strongly history-independent dictionaries 
%(trees and hash tables).
%\Rnote{I chopped out the first paragraph, it was out of context here. Naor and Teague get proper credit later on in this section.}
Hartline et al.~\cite{HartlineHoMo05,HartlineHoMo02} 
showed that a data structure with a strongly-connected state graph 
is strong HI if and only if 
each state of the state graph has a unique \emph{canonical} 
representation.
We rely on a similar characterization % extensively 
for both possibility and impossibility results.
% Canonical representation~\cite{HartlineHoMo05,HartlineHoMo02,BuchbinderPe03}. 

There is a large literature on figuring out which data structures 
can be made history independent without an asymptotic slow down.
%~\cite{Micciancio97, NaorTe01, HartlineHoMo05,HartlineHoMo02,AcarBlHa04, BuchbinderPe03, BlellochGo07, NaorSeWi08, Golovin09, Golovin10,BenderBeJo16}.  
These results include 
fast HI constructions for cuckoo hash tables~\cite{NaorSeWi08}, 
linear-probing hash tables~\cite{BlellochGo07,GoodrichKoMi17}, 
other hash tables~\cite{BlellochGo07,NaorTe01}, 
trees~\cite{Micciancio97,AcarBlHa04},
memory allocators~\cite{NaorTe01,GoodrichKoMi17}, 
write-once memories~\cite{MoranNaSe07},
priority queues~\cite{BuchbinderPe03}, 
union-find data structures~\cite{NaorTe01}, 
external-memory dictionaries~\cite{Golovin08,Golovin09,Golovin10,BenderBeJo16}, 
file systems~\cite{BajajSi13b, BajajSi13a,BajajChSi15, RocheAvCh15},
cache-oblivious dictionaries~\cite{BenderBeJo16}, 
order-maintenance data structures~\cite{BlellochGo07},
packed-memory arrays/list-labeling data structures~\cite{BenderBeJo16,BenderCoFa22},
and geometric data structures~\cite{Tzouramanis12}.
Given the strong connection between history independence and unique representability~\cite{HartlineHoMo02,HartlineHoMo05}, 
some earlier data structures % that predate the formalism 
can be made history independent, 
including % ordered linear probing/Robin-hood 
hashing variants~\cite{AmbleKn74,CelisLaMu85}, 
skip lists~\cite{Pugh90}, treaps~\cite{AragonSe89}, 
and other less well-known deterministic data structures 
with canonical representations~\cite{SundarTa90,AnderssonOt91,AnderssonOt95,PughTe89,Snyder77}. 

The foundational algorithmic work on history independence has 
found its way into systems.
There are now voting machines~\cite{BethencourtBoWa07}, 
file systems~\cite{BajajSi13b,BajajChSi16,BajajChSi15}, 
databases~\cite{BajajSi13a,PoddarBoPo16,RocheAvCh16},
and other storage systems~\cite{ChenSi15} that support 
history independence as an essential feature.

To the best of our knowledge, 
the only prior work to consider history independence in concurrent implementations
is by Shun and Blelloch~\cite{ShunBlellochSPAA14}.
They implement a concurrent hash table, 
based on a sequential SHI hash table~\cite{NaorTe01, BlellochGo07}, 
in which only operations of the same type can be executed concurrently. 
The implementation of~\cite{ShunBlellochSPAA14} guarantees 
that if there are no ongoing insert or delete operations,
each state of the hash table has a unique canonical representation in memory.
%That is, using our terminology, the implementation is HI when the observer 
%can only inspect the memory when no state-changing operations are pending.
This work does not provide a formal definition of history independence for concurrent implementations, 
and does not support concurrent operations of different types.

Our universal implementation draws inspiration from prior universal implementations, 
using CAS~\cite{HerlihySIG90}, hardware LL/SC~\cite{HerlihyLLSC93}, 
and consensus objects~\cite{HerlihyWaitFreeSynch}.
These implementations are not history independent:
the implementation in~\cite{HerlihyWaitFreeSynch} explicitly keeps tracks 
of all the operations that have ever been invoked,
while the implementations in~\cite{HerlihySIG90,HerlihyLLSC93} store information 
that depends on the sequence of applied operations.
Moreover, they use dynamic memory, and allocate new memory every time the state of the object is modified, which risks revealing information about the history of operations.
While there is work on sequential history-independent memory allocation (e.g.,~\cite{NaorTe01}), 
to our knowledge, no \emph{concurrent} history-independent memory allocator is known.
%, and this seems to be a challenging and interesting problem.
%These implementations use dynamic memory, where new memory is allocated 
%for every new state of the object.
%Memory can be allocated in an HI manner for sequential implementation~\cite{NaorTe01}.
%However, since more complex memory management is required for concurrent implementations, 
%there is no straightforward definition of history independence 
%for concurrent implementations that use dynamic memory.

Fatourou and Kallimanis~\cite{FatourouKallimanisSPAA11}
give a universal implementation from hardware LL/SC,
where the full state of the object, 
along with additional information, is stored in a single memory cell.
Our history-independent universal implementation bears some similarity to~\cite{FatourouKallimanisSPAA11},
but their implementation keeps information about completed operations, such as their responses,
and is therefore not history independent.
Clearing this type of information from memory so as not to reveal completed operations is non-trivial,
and we address this in our implementation.

\section{Preliminaries}
\label{sec:prelim}

\subsubsection*{Abstract Objects.}
An \emph{abstract object} is defined by a set of states and a set of operations, 
each of which may change the state of the object and return a response.
Formally, an abstract object $\calO$ is a tuple $(Q, q_0, O, R, \Delta)$, 
where $Q$ is the object's set of states, $q_0\in Q$ is a designated initial state, $O$ is the object's set of operations,
$R$ is the object's set of responses and $\Delta : Q \times O \to Q \times R$ is the function specifying the object behavior, known as the \emph{sequential specification} of the object.
We assume that the abstract object is deterministic, i.e., the function $\Delta$ is deterministic.
We also assume that all states in $Q$ are reachable from the initial state $q_0$.
% each operation transitions each state to another state deterministically.
% Possible assumptions on an abstract object: (1) all states are reachable from the initial state, (2) all states are visible -- for any state $q\in Q$ there is a reachable state $q_1$ from $q$, state $q'$ and operation $o\in O$ such that for any reachable state $q_2$ from $q'$, $r \neq r'$, where $(q_1,o,\circ,r),(q_2,o,\circ,r') \in \Delta$.

A \emph{sequential implementation} specifies how each operation should be concretely implemented in memory. The object's abstract state is represented in memory using some \emph{memory representation}, and the implementation specifies how each operation should modify this memory representation when it is applied.
We note that in general, a single state $q \in Q$ may have multiple possible memory representations associated with it;
for example, if we implement a set (i.e., a dictionary) using a balanced search tree,
then the abstract state of the object
consists of the contents of the set,
but
the layout in memory may also depend on the order of insertions, etc.

% \begin{definition}[Weak History Independence (WHI)~\cite{NaorTe01}]
%      A sequential implementation of an object is \emph{weakly history independent} if any pair of operation sequences that take the object from the initial state to the same state, induce the same distribution on the memory representation.
% \end{definition}
% Intuitively, \emph{weak history independence} (WHI) assumes the adversary examines the memory representation only once. If an adversary can gain access to the memory representation multiple times along the sequence of operations, assuming WHI is not enough to conceal the sequence of operations. 
% Hence, a stronger definition is required, which considers an adversary that may access the memory representation at multiple points during the execution, called \emph{strong history independence} (SHI).
% \begin{definition}[Strong History Independence (SHI)~\cite{NaorTe01}]
%      A sequential implementation of an object is \emph{strongly history independent} if for any pair of operation sequences $o_1^1, \dots, o_{l_1}^1$ and $o_1^2, \dots, o_{l_2}^2$, and two lists of points $P_1 = \{ i_1^1, \dots, i_l^1\}$ and $P_2 = \{ i_1^2, \dots, i_l^2\}$, such that for all $b\in \{1,2\}$ and $1\leq j \leq l$ we have that $1\leq i_j^b \leq l_b$ and the two operations' prefixes $o_1^1, \dots, o_{i_j^1}^1$ and $o_1^2, \dots, o_{i_j^2}^2$ take the object from the initial state to the same state, then the distributions of the memory representations at the points of $P_1$ and the corresponding points of $P_2$ are identical.
% \end{definition}

\subsubsection*{Sequential History Independence.}
A \emph{history-independent} implementation 
is one where the memory representation of the object
reveals only its current abstract state,
and not the sequence of operations that have led to that state.

\begin{definition}[Weak History Independence (WHI)~\cite{NaorTe01}]
     A sequential implementation of an object is \emph{weakly history independent} if any pair of operation sequences that take the object from the initial state to the same state, induce the same distribution on the memory representation.
\end{definition}
Intuitively, WHI assumes the adversary examines the memory representation only once. If an adversary can gain access to the memory representation multiple times along the sequence of operations, assuming WHI is not enough to conceal the sequence of operations. 
Hence, a stronger definition is required, which considers an adversary that may access the memory representation at multiple points during the execution.
\begin{definition}[Strong History Independence (SHI)~\cite{NaorTe01}]
     A sequential implementation of an object is \emph{strongly history independent} if for any pair of operation sequences $o_1^1, \dots, o_{l_1}^1$ and $o_1^2, \dots, o_{l_2}^2$, and two lists of points $P_1 = \{ i_1^1, \dots, i_{\ell}^1\}$ and $P_2 = \{ i_1^2, \dots, i_{\ell}^2\}$, such that for all $b\in \{1,2\}$ and $1\leq j \leq \ell$ we have that $1\leq i_j^b \leq l_b$ and the two operations' prefixes $o_1^1, \dots, o_{i^1_j}^1$ and $o_1^2, \dots, o_{i^2_j}^2$ take the object from the initial state to the same state, then the distributions of the memory representations at the points of $P_1$ and the corresponding points of $P_2$ are identical.
\end{definition}

When deterministic implementations are concerned, both versions of history independence---weak and strong history independence---converge,
since in the absence of randomness it does not matter whether the adversary examines the memory only at one point, or whether it may examine the memory at multiple points in the execution.  
% For completeness, Appendix~\ref{app:defs prelim} gives the formal definitions of weak HI (WHI) 
% and strong HI (SHI).

One way to achieve history independence is to ensure that whenever the object is the same abstract state, the memory representation is the same.
Implementations that have this property are called \emph{canonical}:
every abstract state $q \in Q$ corresponds to exactly one memory representation, $\can{q}$,
which we call its \emph{canonical memory representation}, and
every sequence of operations that leads the object to state $q$
must leave the memory in state $\can{q}$.
A canonical sequential implementation is, by definition, strongly history independent,
but when randomization is allowed, not every strongly history-independent implementation must be canonical.
Nevertheless, it was shown in~\cite{HartlineHoMo02,HartlineHoMo05}
that for \emph{reversible} objects,
which are objects where every state is reachable from every other state,%
\footnote{For example, a register is reversible, while an increment-only counter is not.}
the use of randomization is somewhat limited:~\cite{HartlineHoMo02,HartlineHoMo05} show that strong history independence requires that the implementation
must fix a canonical representation when the object is initialized,
meaning the randomness cannot affect the memory representation
except to fix the mapping from states to their canonical representations at the very beginning of the execution.

%is mapped to exactly one canonical memory representation.
%If an implementation is canonical, it is also SHI.
%The converse, where SHI implies a canonical implementation, holds for 
%for the following class of objects.
% It was proven that a sequential implementation of a \emph{reversible} object, i.e., all states are mutually reachable, is SHI if and only if the implementation is canonical.
%An object is \emph{reversible}, if for any pair of states $q, q'\in Q$, $q'$ is reachable from $q$, that is, some sequence of operations takes the object from state $q$ to state $q'$.
%\TODO{Embed the lemma and the proposition in the discussion}

%\begin{lemma}[\cite{HartlineHoMo02,HartlineHoMo05}]
%    In a SHI sequential implementation of a reversible object, a unique canonical memory representation is determined for each state at initialization\footnote{\ns{For example, a common initialization for a hash table is randomly choosing the hash functions.}}.
%\end{lemma}

%Note that all abstract objects have a SHI sequential implementation, which sets a unique canonical memory representation for each state and maps each sequence of operations to the canonical representation of the resulting state. 
%However, this implementation might not be \emph{efficient}.

% Hence, for sequential implementations, we are only interested in the complexity of such implementations~\cite{BUCHBINDER2006291}.

As we said above, for deterministic implementations,
weak and strong history independence coincide;
both require the implementation to be canonical:

\begin{proposition}
    For deterministic sequential implementations, WHI and SHI are equivalent to  
    requiring that a unique canonical memory representation is determined for each state at initialization.
%\Rnote{Could potentially replace this by saying that WHI and SHI are equivalent to saying that the implementation is canonical.}
\end{proposition}

% \begin{proof}
% Let $S_1$ and $S_2$  be two sequences of operations that take the object from 
% the initial state to the same state $q$.
% Since the implementation is deterministic, 
% $S_1$ and $S_2$ map to the same memory representation $\can{q}$, 
% as needed. 
% \end{proof}

% This implies that the implementation is canonical, and thus, 
% also SHI and WHI this concludes the proposition.
Since we consider only deterministic implementations in this paper, we do not distinguish between SHI and WHI, and refer to them simply as \emph{HI}.

%Note this still allows for random initialization that gives a distribution of deterministic mappings.

\subsubsection*{The Asynchronous Shared-Memory Model.}
We use the standard % asynchronous shared memory model. There are 
model in which $n$ processes, $p_1, \dots, p_n$, communicate through shared \emph{base objects}.
An implementation of an abstract object $\calO$ specifies 
%\ns{the initial state $q_0$} and 
for each process a program for every operation in $O$.
When receiving an \emph{invocation} of an operation $o\in O$, 
process $p_i$ takes \emph{steps} according to this program. 
Each step by process $p_i$ consists of some local computation 
and a single primitive operation on a \emph{base object}. 
The process possibly changes its local state after a step and 
possibly returns a \emph{response} to a higher-level operation.
%HA: took this out --- confusing
%Each operation on an \emph{atomic} object completes in one step. 
%The base objects can be either implemented or atomic.

A \emph{configuration} $C$ specifies the state of every process 
and the state of every base object.
The \emph{initial configuration} is denoted by $C_0$,
and we assume that it is unique.
An \emph{execution} $\alpha$ is an alternating sequence of configurations and steps.
An execution can be finite or infinite.
%A \emph{finite} (resp.\ \emph{infinite}) execution includes a finite (resp.\ infinite) number of steps. 
%A \emph{partial} execution is either a finite or infinite execution.
Given two executions $\alpha_1, \alpha_2$,
where $\alpha_1$ ends at configuration $C$
and $\alpha_2$ begins at configuration $C$,
we denote by $\alpha_1 \alpha_2$ the concatenation of the two executions, and we say that $\alpha_1 \alpha_2$
\emph{extends} execution $\alpha_1$.

% A \emph{schedule} $\sigma$ is a sequence of process identifiers, along with operations they invoke in invocation steps.
% Configuration $C$ and schedule $\sigma$ uniquely define exactly one execution, denoted $\mathrm{exec}(C,\sigma)$. For a finite schedule $\sigma$, the final configuration in execution $\mathrm{exec}(C,\sigma)$ is denoted $\sigma(C)$. A \emph{$p$-solo} schedule $\sigma$ consists only of steps by the process $p$.

The \emph{memory representation} of a configuration $C$, denoted $\mem{C}$,
is a vector specifying the state of each base object; 
note that this does not include local private variables held by each process, only the shared memory.
Formally, if there are $m$ base objects with state spaces $Q_1,\ldots,Q_m$ respectively,
then $\mem{C}$ is a vector in $Q_1 \times \ldots \times Q_m$,
and we denote by $\mem{C}[i]$ the state of the $i$-th base object.
For a finite execution $\alpha$, we let $\mem{\alpha}$ denote the memory representation in the last configuration in the execution $\alpha$.
%\Rnote{Do we ever use this? Not in the impossibility proof, I believe. If not, move it to the appendix.}
% We allow initial randomness embodied by the set of \emph{initial configurations} $C_0^*$

%\Rnote{Consider moving this to the appendix, it's well known:}
In this paper we frequently make use of
two types of base objects: the first is a simple \emph{read/write register},
and the second  is a  \emph{compare-and-swap (CAS)} object.
A $\textsc{CAS}$ object $X$ supports the operation $\textsc{CAS}(X, \mathit{old}, \mathit{new})$,
which checks if the current value of the object is $\mathit{old}$, and if so, replaces it by $\mathit{new}$ and returns \textit{true};
otherwise, the operation leaves the value unchanged,
and returns \textit{false}.
We assume that the $\textsc{CAS}$ object supports 
standard read and write operations.
For both read/write registers and $\textsc{CAS}$ objects,
the \emph{state} of the object is simply the value stored in it. 

%\Rnote{Consider moving everything from here until the end to the appendix, this is also standard stuff. I think we can get away with saying ``we use the standard definitions of linearizability, sequential histories, lock-freedom and wait-freedom; see Appendix bla for the formal definitions.''}
An execution $\alpha$ induces a \emph{history} $H(\alpha)$, 
consisting only of the invocations and responses of higher-level operations.
An invocation \emph{matches} a response if they both belong to the same operation.
An operation \emph{completes} in $H$
if $H$ includes both the invocation and response of the operation;
if $H$ includes the invocation of an operation, but no matching response, 
then the operation is \emph{pending}.
If $\alpha$ ends with a configuration $C$, 
and there is no pending operation in $\alpha$, 
then $C$ is \emph{quiescent}.
%\TODO{We switched from talking about executions/histories to configurations. Do we care?}

A history $H$ is \emph{sequential} if every invocation is immediately 
followed by a matching response.
%This naturally extends to define sequential executions.
For a sequential history $H$, 
let $\seqstate{H}\in Q$ be the state reached by applying 
the sequence of operations implied by $H$ from the initial 
state according to $\Delta$. 
This state is well-defined, since the sequence of operations 
is well-defined in a sequential history.

\subsubsection*{Linearizability.}
A \emph{completion} of history $H$ is a history $H'$ whose prefix 
is identical to $H$, and whose suffix includes zero or more responses 
of pending operations in $H$. 
Let $\mathrm{comp}(H)$ be the set of all of $H$'s completions. 
A sequential history $H'$ is a \emph{linearization} of an execution $\alpha$ that arises from an implementation of an abstract object $\calO$ if:
\begin{itemize}
    \item  $H'$ is a permutation of a history in $\mathrm{comp}(H(\alpha))$,
    \item $H'$ matches the sequential specification of $\calO$, and 
    \item $H'$ respects the real-time order of non-overlapping operations in $H(\alpha)$.
\end{itemize}
% (2) $H'$ contains all complete operations and completions of a subset of the pending operations in $H(\alpha)$,
An execution $\alpha$ is \emph{linearizable}~\cite{HerlihyWing90} if it has a linearization, and an implementation of an abstract object is linearizable 
if all of its executions are linearizable.
A \emph{linearization function} $h$ maps an execution $\alpha$
to a sequential history $h(\alpha)$ that is a linearization of $\alpha$.

%We consider two progress conditions. 
\subsubsection*{Progress Conditions.}
An implementation is \emph{lock-free} if there is a pending operation, 
then some operation returns in a finite number of steps.
An implementation is \emph{wait-free} if there is a pending operation by process $p_i$, 
then this operation returns in a finite number of steps by process $p_i$.

\section{History Independence for Concurrent Objects}

As we noted in Section~\ref{sec:intro}, when defining history independence for concurrent objects, we must grapple with the fact that a concurrent system might never be in a quiescent state. In Section~\ref{sec:gen-lb} we prove that if we allow the internal memory to be observed at \emph{any} point in the execution,
then there is a strong impossibility result ruling out even lock-free implementations of a wide class of objects.
This motivates us to consider weaker but more feasible definitions, where the observer may only examine the internal memory at certain points in the execution.

The following definition provides a general framework for defining a notion of history independence that is parameterized by the points where the observer is allowed to access the internal memory;
these points are specified through a set of finite executions, and the observer may access the memory representation only at the end of each such finite execution.
%in the implementation's execution where the adversary can access the memory representation. The points are represented as finite executions, where the adversary can access the memory representation of the last configuration in the execution.
Informally, the definition requires that at any two points where the observer is allowed to examine the internal memory, if the object is in the same state, then the memory representation must be the same;
we use a linearization function to determine what the ``state'' of the object is at a given point.
%The definition uses a linearization function $h$ to determine the state of the object at a given point. For any two points the adversary can access that lead to the same state, the memory representation must be identical.

\begin{definition}
\label{def:genperfectSHI}
    Consider an implementation of an abstract object with a linearization function $h$, and let
    $E$ be a set of finite executions that arise from the implementation.
    The implementation is \emph{HI with respect to $E$} if for any pair of executions $\alpha, \alpha'\in E$ such that $\seqstate{h(\alpha)} = \seqstate{h(\alpha')}$, then $\mem{\alpha} = \mem{\alpha'}$.
    An implementation of an abstract object is \emph{HI with respect} to $E$, if it is HI with respect to $E$ for some linearization function $h$.
\end{definition}

% For example, sequential HI is captured by requiring the implementation to be HI 
% w.r.t.\ to all the finite sequential executions that end with a quiescent configuration.

To prove that an implementation satisfies Definition~\ref{def:genperfectSHI},
it suffices to prove that 
%there is a canonical memory representation $\mathit{can}$ such that
for every finite execution $\alpha$ in the 
set $E$,
if $E$ ends with the object in state $q$ according to the linearization function $h$,
then the memory is in the canonical memory representation $\can{q}$ for state $q$.

%\begin{proposition}
%\label{prop:HI-closed}
%    If an implementation of an abstract object with linearization function $h$ is HI with respect to an execution set $E$, then it is also HI with respect to any execution set $E'\subseteq E$.
%\end{proposition}

%\begin{proof}
%    Let $\alpha, \alpha'\in E'$ such that $\seqstate{h(\alpha)} = \seqstate{h(\alpha')}$. As $E'\subseteq E$, $\alpha, \alpha'\in E$. Thus, $\mem{\alpha} = \mem{\alpha'}$.
%\end{proof}

%Later in the paper, we use the following requirements on executions when discussing HI implementations.

The strongest form of history independence that one might ask for is one that allows the observer to examine the memory at any point in the execution:
\begin{definition}
\label{def:perfectSHI}
    An implementation of an abstract object is \emph{perfect HI}, if the implementation is HI with respect to the set containing all finite executions of the implementation.
\end{definition}

Perfect HI imposes a very strong requirement on the implementation: intuitively, any
two adjacent high-level states
must have adjacent memory representations.
%\ns{Perfect HI imposes a very strong requirement on the canonical representations of the states, as stated in the next proposition.}
%Let $mem_1, mem_2$ % $\in Q_1 \times \dots \times Q_m$ 
Formally, 
we say that the \emph{distance} between
two memory representations $\mathit{mem}_1, \mathit{mem}_2$
is $d$ if 
there are exactly $d$ indices $i \in [m]$ where $\mathit{mem}_1[i] \neq \mathit{mem}_2[i]$ (recall that $m$ is the number of base objects used in the implementation).
The following proposition
%, proved in Appendix~\ref{app:proofs prelim},
holds for \emph{obstruction-freedom},\footnote{
    An implementation is \emph{obstruction-free} if an operation by process $p_i$ returns 
    in a finite number of steps by $p_i$, if $p_i$ runs solo.} 
a progress guarantee even weaker than lock-freedom and wait-freedom.

\begin{restatable}{proposition}{pHI}
\label{prop:p-hi}
In any obstruction-free
perfect HI implementation of an abstract object with state space $Q$,
for any $q_1 \neq q_2 \in Q$
such that $q_2$ is reachable from $q_1$ in a single operation,
the distance between $\can{q_1}$
and $\can{q_2}$ is at most 1.
\end{restatable}

\begin{proof}
    Assume that $\can{q_1}$ and $\can{q_2}$ are at distance at least 2.
    Consider a sequential execution $\alpha$ 
    that ends in a quiescent configuration such that only one process $p_i$ takes steps in $\alpha$ and 
    $\seqstate{H(\alpha)} = q_1$. Since the implementation is perfect HI, $\mem{\alpha} = \can{q_1}$.
    Let $\beta = \alpha \alpha'$ be an extension of $\alpha$, 
    such that $\alpha'$ consists of a single operation by process $p_i$ and $\seqstate{H(\alpha)} = q_2$. 
    This extension exists since $q_2$ is reachable from $q_1$ by a single operation.
    Since the implementation is perfect HI, $\mem{\beta} = \can{q_2}$. 
    This means that $p_i$ changes the state of two base objects in two separate steps in $\alpha'$, 
    since $\can{q_1}$ and $\can{q_2}$ are at distance at least 2. 
    Let $\beta'$ be the shortest prefix of $\beta$ that includes the first step 
    that changes the state of a base object during the operation executed in $\alpha'$. 
    Then, $\mem{\beta'} \neq \can{q_1}$ and also  $\mem{\beta'} \neq \can{q_2}$.
    However, 
    for every linearization function $h$,
    $\seqstate{h(\beta')} = q_1$ or $\seqstate{h(\beta')} = q_2$, 
    in contradiction.
\end{proof}

Proposition~\ref{prop:p-hi-lb} (Section~\ref{sec:gen-lb}) shows that 
a large class of objects cannot meet the requirement above,
and therefore, no obstruction-free implementation can be perfect HI. 
This motivates us to consider weaker definitions, 
where the observer may only observe the memory at points that are ``somewhat quiescent'',
or fully quiescent.

We say that an operation $o\in O$ is \emph{state-changing} if there 
exist states $q \neq q'$ such that $o$ causes the object to transition from state $q$ to $q'$.
An operation is \emph{read-only} if it is not state-changing.
A configuration $C$ is \emph{state-quiescent} if there are no pending state-changing operations in $C$. 
Note that a quiescent configuration is also state-quiescent.

\begin{definition}
\label{def:op-quiescentSHI}
    An implementation of an abstract object is \emph{state-quiescent HI} if the implementation is HI with respect to all finite executions ending with a state-quiescent configuration.
    %with a configuration without a pending state-changing operation.
\end{definition}

% \begin{definition}
%     A deterministic implementation is \emph{solo HI (S-HI)} if the implementation is HI with respect to the set containing all partial executions ending with an operation return, and this operation runs alone starting from any configuration.
% \end{definition}
The final definition is the weakest one that we consider, and it allows the observer to examine the memory only when the configuration is fully quiescent:
\begin{definition}
\label{def:quiescentSHI}
    An implementation of an abstract object is \emph{quiescent HI}, if the implementation is HI with respect to all finite executions ending with a quiescent configuration.
\end{definition}

% \begin{definition}
% \label{def:solo-quiescentSHI}
%     An implementation is \emph{solo-quiescent HI}, if the implementation is HI with respect to the set containing all partial executions ending with a state changing operation return, and this operation runs alone starting from a quiescent configuration.
% \end{definition}

Figure~\ref{fig:HI-defs-example} depicts the definitions, 
on an execution of a register implementation,
highlighting several different points where an observer is or is not 
allowed to examine the memory according to each definition.
%(Clearly, perfect HI implies state-quiescent HI, which in turn implies quiescent HI.)
%\Hnote{Exapnd?}
%using an execution of a register.
Clearly, if an implementation of an abstract object with linearization function $h$ 
is HI with respect to an execution set $E$, 
then it is also HI with respect to any execution set $E'\subseteq E$.
Since the respective sets of executions are contained in each other, 
this means that 
perfect HI implies state-quiescent HI, which in turn implies quiescent HI.
%By Proposition~\ref{prop:HI-closed}, with abuse of notation, 
%Definition~\ref{def:quiescentSHI} $\subseteq$ Definition~\ref{def:op-quiescentSHI} $\subseteq$  Definition~\ref{def:perfectSHI}.

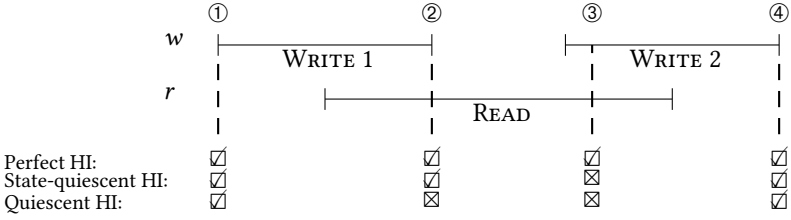
\begin{figure}[tb!] 
 \setlength{\unitlength}{0.14in}
 \centering
 \begin{picture}(30,7)
 \put(2,5.9) {$w$}
 % first write
 \put(4,5.9){\line(1,0){8}}
 \put(4,5.9){\line(0,0){0.4}}
 \put(4,5.9){\line(0,-1){0.4}}
 \put(12,5.9){\line(0,0){0.4}}
 \put(12,5.9){\line(0,-1){0.4}}
 \put(8,5.4) {\makebox(0,0){\ \textsc{Write} 1}}
 {\thicklines
 \multiput(4,4.6)(0,-1){3}{\line(0,1){0.6}}
 \multiput(12,4.6)(0,-1){3}{\line(0,1){0.6}}
 %\put(4,9.2){\oval(0.8,0.8)}
 %\put(12,9.2){\oval(0.8,0.8)}
 }
 \put(4,7.1) {\makebox(0,0){\ding{192}}}
 \put(12,7.1) {\makebox(0,0){\ding{193}}}
 \put(4,1.8) {\makebox(0,0){\CheckedBox}}
 \put(4,1) {\makebox(0,0){\CheckedBox}}
 \put(4,0.2) {\makebox(0,0){\CheckedBox}}
 \put(12,1.8) {\makebox(0,0){\CheckedBox}}
 \put(12,1) {\makebox(0,0){\CheckedBox}}
\put(12,0.2) {\makebox(0,0){\XBox}}
 % second write
 \put(17,5.9){\line(1,0){8}}
 \put(17,5.9){\line(0,0){0.4}}
 \put(17,5.9){\line(0,-1){0.4}}
 \put(25,5.9){\line(0,0){0.4}}
 \put(25,5.9){\line(0,-1){0.4}}
 \put(21,5.4) {\makebox(0,0){\ \textsc{Write} 2}}
 {\thicklines
 \multiput(25,4.6)(0,-1){3}{\line(0,1){0.6}}
 %\put(25,9.2){\oval(0.8,0.8)}
 \multiput(18,4.6)(0,-1){3}{\line(0,1){0.6}}
 \put(18,5.9){\line(0,-1){0.3}}
 %\put(18,9.2){\oval(0.8,0.8)}
 }
 \put(18,7.1) {\makebox(0,0){\ding{194}}}
 \put(18,1.8) {\makebox(0,0){\CheckedBox}}
 \put(18,1) {\makebox(0,0){\XBox}}
 \put(18,0.2) {\makebox(0,0){\XBox}}
 \put(25,7.1) {\makebox(0,0){\ding{195}}}
 \put(25,1.8) {\makebox(0,0){\CheckedBox}}
 \put(25,1) {\makebox(0,0){\CheckedBox}}
 \put(25,0.2) {\makebox(0,0){\CheckedBox}}
 % reader
 \put(2,3.9) {$r$}
 \put(8,3.9){\line(3,0){13}}
 \put(8,3.9){\line(0,0){0.4}}
 \put(8,3.9){\line(0,-1){0.4}}
 \put(21,3.9){\line(0,0){0.4}}
 \put(21,3.9){\line(0,-1){0.4}}
 \put(14.5,3.4) {\makebox(0,0){\ \textsc{Read}}}

 \put(-4,1.6) {\makebox(0,0)[l]{\footnotesize Perfect HI:}}
 \put(-4,0.8) {\makebox(0,0)[l]{\footnotesize State-quiescent HI:}}
 \put(-4,0) {\makebox(0,0)[l]{\footnotesize Quiescent HI:}}
 \end{picture}
 \caption{\small 
Illustration of the three HI definitions.
Perfect HI allows the observer to examine the memory at any point; 
state-quiescent HI allows examination only when there is no state-changing operation pending (points 1, 2 and 4);
while quiescent HI allows examination only when the configuration is quiescent (points 1 and 4).}
 %At points 1 and 4, both state-quiescent and quiescent HI require the memory to be in a canonical representation, as there are no pending operations. At point 2, only state-quiescent HI requires the memory to be in a canonical representation, as there is no pending write operation, however, there is a pending read operation. Perfect HI requires the memory to be in a canonical representation at any point during the execution. This includes point 3, where both the read and write operations are pending.}
 \label{fig:HI-defs-example}
 %\vspace{-15pt}
 \end{figure}

\section{Motivating Example: Multi-Valued Register from Binary Registers}
\label{section: swsr example}
\label{sec:swsrexample}
%\Rnote{I added ``motivating example'' to the section title, xcolor throws an error if I try to highlight it.}

To better understand the notion of history independence, 
and the challenges in achieving it, 
consider the problem of implementing a \emph{single-writer single-reader 
(SWSR) $K$-valued register},
for $K \geq 3$, from \emph{binary registers}, i.e., registers whose value is either 0 or 1. 
%
%
%Since there is only one writer, and read operations do not change
%the state of the register, there can only be one state-changing operation at any point in the execution.
%Note that read operations do not change the state of the register and since there is a single writer, there are no concurrent writes in an execution, i.e., there is only one state-changing operation at any point of the execution. 
%Assume the multi-valued register has $K$ values, $0,\dots K-1$.
%

% In a multi-valued register,
% any pair of states are mutually reachable from each other by a single operation. 
% Thus, the memory representations that arise from an implementation of a $K$-valued 
% ($K \geq 3$) 
% register from binary registers cannot all be at distance 1 from each other. 
% \Hnote{We say later the more general argument in Corollary~\ref{cor:p-hi-lb}; should we still have this here?}
% Proposition~\ref{prop:p-hi} implies that there is no obstruction-free perfect HI implementation 
% of such a register from binary registers.

Algorithm~\ref{alg:multi-valued-baseline} presents 
Vidyasankar's wait-free implementation of a SWSR $K$-valued register 
from binary registers~\cite{VIDYASANKAR1988287}.
The algorithm is for multiple readers, 
but we assume a single reader here.
%Note that this algorithm works for multiple readers, however, we restrict ourselves to only one reader, which still yields non-trivial results.

The value of the register is represented by a binary array $A$ of size $K$,
and intuitively, the register's value at any given moment 
is the smallest index $i \in [K]$
such that $A[i] = 1$.%
\footnote{This is true in the sense that a $\Read$ operation
that does not overlap with a $\Write$ will return the smallest index $i \in [K]$ such that $A[i] = 1$;
for a $\Read$ that does overlap with a $\Write$, the picture is much more complicated~\cite{VIDYASANKAR1988287}.}
In a $\Read$ operation, the reader $r$ \emph{scans up} to find 
the smallest index $i \in [K]$ such that $A[i] = 1$, and then \emph{scans down} from $i$ and returns the smallest index $j \leq i$ 
such that when $A[j]$ was read, its value was 1. (If the reader
executes solo, we will have $i = j$, but if there is a concurrent $\Write$, we may have $j < i$.)
In a $\Write(v)$ operation, the writer $w$ writes 1 to $A[v]$, then writes 0 to all indices $i < v$, starting from index $v - 1$
and proceeding down to index 1.

Since a $\Write(v)$ operation does not clear 
values larger than $v$,
the state of the array $A$ leaks information about past values written to the register:
e.g., if $K = 3$ and there is a $\Write(2)$ operation followed by $\Write(1)$, we will have $A = [1,1,0]$, whereas if we have only a $\Write(1)$,
the state will be $A = [1,0,0]$.
This will happen even in sequential executions, 
so this implementation is not history independent 
even in the minimal sense: it does not satisfy the sequential definition of history independence, even if we consider only sequential executions.

One might hope that this can be fixed by having a $\Write(v)$ operation clear the entire array $A$,
except for $A[v] = 1$, but this would
break the wait-freedom of the implementation:
if the writer zeroes out all positions in the array except one, the reader might never be able to find an array position $i$ where $A[i] = 1$, and it might not find a value that it can return.
In fact, the impossibility result that we prove in the next section 
(Theorem~\ref{thm:oq-hi-gen-lb}) rules out any wait-free implementation that is history-independent,
even if we examine the memory only when no $\Write$
operation is pending.

\begin{wrapfigure}{r}{0.45\textwidth}
%\vspace{-20pt}
\begin{minipage}{0.45\textwidth}\small
\begin{algorithm}[H]
\caption{\small Wait-free SWSR multi-valued register from binary registers~\cite{VIDYASANKAR1988287}}
\label{alg:multi-valued-baseline}
\begin{algorithmic}[1]
    \Statex $A[1 \dots K]$, all entries are initially 0, except
    \Statex $A[v_0]$, where $v_0$ is the initial value 
    \Statex \hrulefill
    %\setlength{\multicolsep}{0.0pt}
    %\begin{multicols}{2}
    \Statex \textsc{Read():} code for the reader $r$
    \State $j \gets 1$
    \While{$A[j] = 0$} $j\gets j + 1$
    \EndWhile
    \State $val \gets j$
    \For{$j = val-1 \dots 1$}
        \If{$A[j] = 1$} $val \gets j$
        \EndIf
    \EndFor
    \State \algorithmicreturn{} $val$
    \Statex
    %\columnbreak
    \Statex \textsc{Write($v$):} code for the write $w$
    \State $A[v] \gets 1$
    \For{$j = v-1 \dots 1$} $A[j] \gets 0$
    \EndFor
    %\end{multicols}
\end{algorithmic}
\end{algorithm}
%\end{minipage}
%\end{wrapfigure}
%\begin{wrapfigure}{r}{0.45\textwidth}
%\vspace{-20pt}
%\begin{minipage}{0.45\textwidth}
\begin{algorithm}[H]
\caption{\small Lock-free state-quiescent HI SWSR multi-valued register from binary registers}
\label{alg:lock-free-op-quiescent}
\begin{algorithmic}[1]\small
    \Statex $A[1\dots K]$, all entries are initially 0, except $A[v_0]$, where $v_0$ is the initial value 
    \Statex \hrulefill
    %\setlength{\multicolsep}{0.0pt}
    %\begin{multicols}{2}
    \Statex \textsc{Read():} code for the reader $r$
    \State $val \gets \bot$
    \While{$val = \bot$}
        \State $val \gets \textsc{TryRead()}$
    \EndWhile
    \State \algorithmicreturn{} $val$
    %\columnbreak
    \Statex
    \Statex \textsc{Write($v$):} code for the write $w$
    \State $A[v] \gets 1$
    \For{$j = v-1,\dots, 1$} $A[j] \gets 0$
    \EndFor
    \For{$j = v + 1 \dots K$} $A[j] \gets 0$
    \EndFor
    %\end{multicols}
\end{algorithmic}
\end{algorithm}
%\end{minipage}
%\end{wrapfigure}
%
%\begin{wrapfigure}{r}{0.43\textwidth}
%\vspace{-13pt}
%\begin{minipage}{0.43\textwidth}
\begin{algorithm}[H]
\begin{algorithmic}[1]\small
    \Statex \textsc{TryRead}(): code for the reader $r$
    \For{$j = 1 \dots K$}
        \If{$A[j] = 1$}
            \State $val \gets j$
            \For{$j' = val-1 \dots 1$}
                \If{$A[j'] = 1$} $val \gets j'$
                \EndIf
            \EndFor
            \State \algorithmicreturn{} $val$
        \EndIf
    \EndFor
    \State \algorithmicreturn{} $\bot$
\end{algorithmic}
\caption{\small Read with failure indication}
\label{alg:try-read}
\end{algorithm}
\end{minipage}
\end{wrapfigure}

Nevertheless, if we are willing to settle for lock-freedom instead of wait-freedom, 
then the approach of clearing out the array after a $\Write$ can be made to work.
We can modify Algorithm~\ref{alg:multi-valued-baseline} 
by having a $\Write(v)$ operation
first clear values down from $v-1$ to $1$
(as in the original implementation),
and then clear values up from $v+1$ to $K$
(which it would not do in the original implementation).
As a result, when there is no $\Write$ operation pending,
the register has a unique representation:
if its value is $v$, then 
the array $A$ is 0 everywhere except at $v$,
where we have $A[v] = 1$. 
Since the reader does not write to memory, 
this implies that the algorithm is state-quiescent HI.
The $\Read$ operation is nearly identical to 
Algorithm~\ref{alg:multi-valued-baseline}, 
except that, as we said above, if a $\Read$ operation
overlaps with multiple $\Write$s, 
it may never find a 1 in the array, and thus it may never find a value to return.
Thus,
the reader repeatedly tries to execute the $\Read$ operation, 
until it finds a value to return. 
We call each such attempt a \textsc{TryRead} (Algorithm~\ref{alg:try-read}), 
and it returns $\bot$ to indicate a failed attempt to find a 1 in $A$.
The code of the modified algorithm appears in Algorithm~\ref{alg:lock-free-op-quiescent}.
While the $\Write$ operation remains wait-free,
the $\Read$ operation is only \emph{lock-free}, and it is only
guaranteed to terminate if it eventually runs by itself.

\subsubsection*{Proof of Algorithm~\ref{alg:lock-free-op-quiescent}}
We define a linearization function for Algorithm~\ref{alg:lock-free-op-quiescent}, 
the same way as it is done in the proof of Algorithm~\ref{alg:multi-valued-baseline} (see~\cite[Section~10.2.1]{HagitBook}).

Let $\alpha$ be an execution of Algorithm~\ref{alg:lock-free-op-quiescent}. 
We say that a low-level read of $A[v]$ in $\alpha$ \emph{reads from} a low-level write to $A[v]$ if this is the latest write to $A[v]$ that precedes this read.
We say that a high-level \textsc{Read} operation $R$ in $H(\alpha)$ \emph{reads from} \textsc{Write} operation $W$, if $R$ returns
the value $v$,
and $W$ is the $\Write$
operation containing the low-level write to $A[v]$ that $R$'s last low-level read of $A[v]$ reads from.
Assume there is an initial logical $\textsc{Write}(v_0)$ operation $W_0$ of the initial value, and any \textsc{Read} that returns the initial value reads from $W_0$.

We linearize the operations in execution $\alpha$ according to the \emph{reads from}
order.
Since there is a single writer $w$, the order of \textsc{Write} operations is well-defined.
The \textsc{Read} operations are also ordered by the order they occur in $\alpha$, which again is well-defined since there is a single reader $r$.
As for the ordering between $\Read$s and $\Write$s, 
immediately after each $\Write$ operation we place
all the $\Read$ operations that read from it,
ordered by the order in which they were invoked.

The next theorem is proved in Appendix~\ref{app:swsr lock-free register proofs}.

\begin{restatable}{theorem}{lockfreelinearization}
\label{lem:lock-free-linearization}
    Algorithm~\ref{alg:lock-free-op-quiescent} is a linearizable lock-free 
    state-quiescent HI SWSR multi-valued register from binary registers.
    %$H$ is a linearization of $\alpha$.
\end{restatable}

% for sufficiently long.
%HA: just to shorten
%The linearization of the algorithm is constructed
%the same way as it is done in the proof of Algorithm~\ref{alg:multi-valued-baseline} 
%(see~\cite[Section~10.2.1]{HagitBook}, and Appendix~\ref{app:swsr lock-free register proofs}).

There is another way to relax our requirements and circumvent the impossibility result: 
we could settle for \emph{quiescent} history independence,
where the adversary is not allowed to inspect the internal
representation except when the system is fully quiescent.
%We show that there is a wait-free implementation of a multi-valued
%register from binary registers that is quiescent HI.
Algorithm~\ref{alg:wait-free-quiescent} is a SWSR multi-valued register from binary registers, that is both wait-free and quiescent HI.
Our algorithm employs the following simple principle: 
the reader announces its presence to the writer whenever it begins a read operation.
The writer, if it sees that the reader might not
find a value to return, \emph{helps} it by
writing a value that the reader is allowed to return, in an area of shared memory
dedicated to this purpose.
We must carefully manage the footprints left in memory by both the reader and the writer,
and ensure that when all operations complete, the memory is left in a canonical representation,
but at the same time, that the reader is never left hanging without a value that it may return.
%
%it uses a quiescent HI algorithm with lock-free reads and wait-free writes and to ensure the reads are in fact wait-free,  
%the reader announces itself to the writer and when the writer identifies a concurrent $\Read$ operation, the writer \emph{helps} the reader by writing a value it can return to a separate location. The helping must be done in a history independent manner that leaves no traces when the system is quiescent, and quiescent HI is achieved because of the lock-free quiescent HI baseline implementation.

In more details,
in addition to the array $A$, which has the same functionality as in 
Algorithm~\ref{alg:lock-free-op-quiescent}, there is an additional array $B$ of size $K$. 
In a $\Read$ operation, first, the reader sets $\flag[1]$ to 1, to announce itself to the writer. Then, the reader tries to read $A$ twice, and if it succeeds in finding an index equal to $1$ in $A$, it returns a value as done in Algorithm~\ref{alg:lock-free-op-quiescent}. 
Otherwise, the reader reads $B$ and returns an index in $B$ that is equal to $1$ (we claim that it will always find such an index).
Before returning, the reader sets $\flag[2]$ to 1 and then clears $B$ by writing $0$ to all the indices in $B$. In addition, it writes $0$ to $\flag[1]$ and then to $\flag[2]$, and returns the value read either from $A$ or from $B$ (whichever one succeeded).
%We prove later (Lemma~\ref{lem:reader-must-write}) that the reader must write to obtain a wait-free implementation.

In a $\Write$ operation, 
the writer reads the entire array $B$ to check if there is any non-zero cell there.
If $B$ contains only zeroes,
and the writer 
identifies a concurrent $\Read$ operation
by observing that $\flag[1] = 1$,
then it writes 1 to $B[\lastval]$, where $\lastval$ is the last value written to the register
(there is only one writer and it retains the last value it wrote before the current one).
After the writer writes to $B$, it reads $\flag[2]$ and then $\flag[1]$, and if it reads $\flag[2] = 1$ or $\flag[1] = 0$, the writer clears the $B$ array by writing 0 to $B[\lastval]$. 
These flag values indicate that either there is no pending $\Read$ operation, or the concurrent $\Read$ operation is done reading $B$.
This ensures that $B$ is cleared by either the writer, 
or by a concurrent $\Read$ operation, and therefore, 
all cells in $B$ are equal to zero when in a quiescent configuration.

Following this interaction with the array $B$,
the writer proceeds in the same manner as in Algorithm~\ref{alg:lock-free-op-quiescent},
writing 1 to location $A[v]$, then clearing $A$ downwards from $v - 1$ to $1$, 
and finally clearing $A$ upwards from $v + 1$ to $K.$

By having the writer first write to $B$ (if it observes a concurrent reader) and then to $A$, we ensure that if the reader fails twice to find a 1 in $A$,
then it is guaranteed that in-between its two scans of $A$,
the writer has written to $A$ in at least two separate $\Write$ operations;
but the writer must have ``seen'' the reader by the time that it wrote to $A$ in its second $\Write$ operation (since the reader immediately sets $\flag[1] = 1$ when it begins),
and at this point is guaranteed to help the reader by setting a cell in $B$ to 1. 

\begin{algorithm}[tb] \small
\caption{\small Wait-free quiescent HI SWSR multi-valued register from binary registers}
\label{alg:wait-free-quiescent}
\raggedright 
     \quad\,\, $A[1\dots K]$, all entries are initially 0, except $A[v_0]$, where $v_0$ is the initial value \\
     \quad\,\, $B[1\dots K]$, all entries are initially 0 \\
     \quad\,\, $\flag[1,2]$, both entries are initially 0 \\
     \quad\,\, local $\lastval$ at the writer $w$, initially $v_0$ \\
     \quad\,\,\, \hrulefill
    \setlength{\multicolsep}{0.0pt}
    \begin{multicols}{2}
    \begin{algorithmic}[1]
    \Statex $\Read()$: code for the reader $r$
    \State $\flag[1] \gets 1$
    \label{lin:raise-flag-1}
    \For{$\mathit{it} = 1,2$}
        \State $val \gets \textsc{TryRead}()$
        \label{lin:try-read}
        \If{$val \neq \bot$} 
            goto Line \ref{lin:raise-flag-2}
        \EndIf
    \EndFor
    \For{$j = 1 \dots K$}
        \If{$B[j] = 1$}  
        \label{lin:reader-read-B}
            $val \gets j$
        \EndIf
    \EndFor
    \State $\flag[2] \gets 1$
    \label{lin:raise-flag-2}
    \For{$j = 1 \dots K$} $B[j] \gets 0$
    \label{lin:reader-clear-B-cond}
    \EndFor
    \State $\flag[1] \gets 0$; 
    \label{lin:clear-flag1}
    % \State 
    $\flag[2] \gets 0$
    \label{lin:clear-flag2}
    \State \algorithmicreturn{} $val$
    \columnbreak
    \Statex $\Write(v)$: code for the write $w$
    \If{$\forall 1\leq j\leq K, B[j] = 0$}
    \label{lin:empty-B-cond}
        \If{$\flag[1] = 1$}
        \label{lin:write-1-B-cond}
            \State $B[\lastval] \gets 1$
            \label{lin:write-1-B}
            \If{$\flag[2] = 1$ or $\flag[1] = 0$}
            \label{lin:writer-clear-B-cond}
                \State $B[\lastval] \gets 0$
                \label{lin:writer-clear-B}
            \EndIf
        \EndIf
    \EndIf
    \State $A[v] \gets 1$
    \label{lin:1toA}
    \For{$j = v-1 \dots 1$} $A[j] \gets 0$
    \label{lin:downward-clear}
    \EndFor
    \For{$j = v + 1 \dots K$} $A[j] \gets 0$
    \label{lin:upward-clear}
    \EndFor
    \State $\lastval \gets v$
\end{algorithmic}
\end{multicols}
\end{algorithm}

We also prove (in Proposition~\ref{lem:reader-must-write}) that it is essential for the reader to write to shared memory, otherwise it is impossible to obtain even a quiescent HI wait-free implementation.

\subsubsection*{Proof of Algorithm~\ref{alg:wait-free-quiescent}}

The next two lemmas (proved in Appendix~\ref{app:swsr wait-free register proofs})
are used to show that Algorithm~\ref{alg:wait-free-quiescent} is linearizable. 
The next lemma proves that a $\Read$ operation returns a valid value.
It is proved by showing that if two \textsc{TryRead} return $\bot$ in a \textsc{Read} operation, 
then there is a $\textsc{Write}$ operation, which overlaps the second \textsc{TryRead}, which, if $B$ has no index equal to 1, writes to $B$ before $R$ starts reading $B$.

\begin{restatable}{lemma}{valnotbot}
\label{lem:val-not-bot}
    If a \textsc{Read} operation $R$ reaches Line~\ref{lin:raise-flag-2}, then $val \neq \bot$.
\end{restatable}

In the next lemma we show that if the reader returns a value read from $B$, it was written by an overlapping \textsc{Write} operation

\begin{restatable}{lemma}{readsboverlaps}
\label{lem:reads-B-overlaps}
    Consider a read of 1 from $B[j]$, $1 \leq j \leq K$, in Line~\ref{lin:reader-read-B} by \textsc{Read} operation R
    and let $W$ be the \textsc{Write} operation that writes this value of 1 to $B[j]$ that $R$ reads,
    then $W$ overlaps $R$.  
\end{restatable}

A \textsc{Read} operation $R$ \emph{reads $A$} if $R$ performs a 
successful \textsc{TryRead} in Line~\ref{lin:try-read} that 
returns a non-$\bot$ value, and $R$ \emph{reads $B$} 
if $R$ reads 1 from $B[j]$ in Line~\ref{lin:reader-read-B}.
If $R$ reads $A$, it \emph{reads $A$ from} \textsc{Write} operation $W$, 
if $R$ returns $v$ and $W$ contains the low-level write to $A[v]$ that $R$'s last low-level read of $A[v]$ reads from.
If $R$ reads $B$, it \emph{reads $B$ from} \textsc{Write} operation $W$, if $R$ returns $v$ and $W$ contains the low-level write to $B[v]$ that $R$'s low-level read of $B[v]$ reads from.

Construct a linearization $H$ of $\alpha$ similarly to the one constructed for Algorithm~\ref{alg:lock-free-op-quiescent}.
Consider the \textsc{Read} operations in the order they occur in $\alpha$, this order is well-defined since there is a single reader $r$.
For every \textsc{Read} operation $R$ in $H(\alpha)$ that reads $A$, let $W$ be the \textsc{Write} operation $R$ reads $A$ from. $R$ is placed immediately after $W$ in $H$.
For every \textsc{Read} operation $R$ in $H(\alpha)$ that reads $B$, let $W_1$ be the \textsc{Write} operation $R$ reads $B$ from and $W_2$ the \textsc{Write} operation that precedes $W_1$ in $\alpha$, i.e., the last \textsc{Write} operation that returns before $W_1$ starts. $R$ is placed before $W_1$ and immediately after $W_2$ in $H$.
Note that the read operations that are placed after a write operation $W$ are ordered by the order in which they were invoked.

\begin{figure}[!tb]
    \centering
    \begin{subfigure}{\textwidth}
    \begin{tikzpicture}
        % reader
        \node at (-.5,0) {$r$};
        \draw[|-|,solid,thick] (0,0) -- (6,0) node [midway, below, text height = 0.4cm] {$R_1$};
        \draw[thick] (4,0.1) -- ++ (0,-0.2);
        \node at (4,0.3) {\footnotesize{writes 0 to $B[j]$}};
        \draw[thick] (2,0.1) -- ++ (0,-0.2);
        \node at (1.8,-0.4) {\footnotesize{\shortstack{\textsc{TryRead} returns \\ $val \neq \bot$}}};
        
        \draw[|-|,solid,thick] (7,0) -- (11,0) node 
        [midway, below, text height = 0.4cm] {$R_2$};
        \draw[thick] (8,0.1) -- ++ (0,-0.2);
        \node at (9,0.3) {\footnotesize{reads 1 from $B[j]$}};

        \node at (4.5,1.5) {$\dots$};

        % writer
        \node at (-.5,1.5) {$w$};
        \draw[|-|,solid,thick] (0,1.5) -- ++ (4,0) node [midway, below] {$W_1$};
        \draw[thick] (1,1.6) -- ++ (0,-0.2);
        \node at (1,1.8) {\footnotesize{writes 1 to $A$}};
        
        \draw[|-|,solid,thick] (5,1.5) -- ++ (5,0) node [midway, below] {$W_2$};
        \draw[thick] (6,1.6) -- ++ (0,-0.2);
        \node at (6,1.8) {\footnotesize{writes 1 to $B[j]$}};
        \draw[thick] (9,1.6) -- ++ (0,-0.2);
        \node at (9,1.8) {\footnotesize{writes 1 to $A$}};
        
        \draw[dotted, thick] (6,1.5) -- (8,0);
        \draw[dotted, thick] (1,1.5) -- (2,0);
        \draw[dotted, thick, red] (9,1.5) -- (2,0);
    \end{tikzpicture}
    \caption{Scenario for read from $A$ before read from $B$}
    \label{fig:thm:wait-free-q-hi-rega}
    \vspace{2mm}
    \end{subfigure}
    \begin{subfigure}{\textwidth}
    \begin{tikzpicture}
        % reader
        \node at (-.5,0) {$r$};
        \draw[|-|,solid,thick] (0,0) -- (7,0) node [midway, below, text height = 0.4cm] {$R_1$};
        \draw[thick] (6,0.1) -- ++ (0,-0.2);
        \node at (6,-0.3) {\footnotesize{reads 1 to $B[j]$}};
        
        \draw[|-|,solid,thick] (8,0) -- (11,0) node 
        [midway, below, text height = 0.4cm] {$R_2$};
        \draw[thick] (9,0.1) -- ++ (0,-0.2);
        \node at (9,0.4) {\footnotesize{\shortstack{\textsc{TryRead} returns \\ $val \neq \bot$}}};

        % writer
        \node at (-.5,1.5) {$w$};
        \draw[|-|,solid,thick] (0,1.5) -- ++ (3,0) node [midway, below] {$W_2$};
        \draw[thick] (2,1.6) -- ++ (0,-0.2);
        \node at (2,1.8) {\footnotesize{writes 1 to $A$}};
        
        \draw[|-|,solid,thick] (4,1.5) -- ++ (6,0) node [midway, below] {$W_1$};
        \draw[thick] (5,1.6) -- ++ (0,-0.2);
        \node at (5,1.8) {\footnotesize{writes 1 to $B[j]$}};
        
        \draw[dotted, thick] (5,1.5) -- (6,0);
        \draw[dotted, thick, blue] (2,1.5) -- (9,0);
        \node at (3.5,0.6) [blue]{\footnotesize{\shortstack{earliest \textsc{Write} $R_2$ \\ can read from}}};
    \end{tikzpicture}
    \caption{Scenario for read from $B$ before read from $A$}
    \label{fig:thm:wait-free-q-hi-regb}
    \end{subfigure}
    \caption{Illustrating the proof of Theorem~\ref{thm:wait-free-q-hi-reg}}
    \label{fig:thm:wait-free-q-hi-reg}
\end{figure}
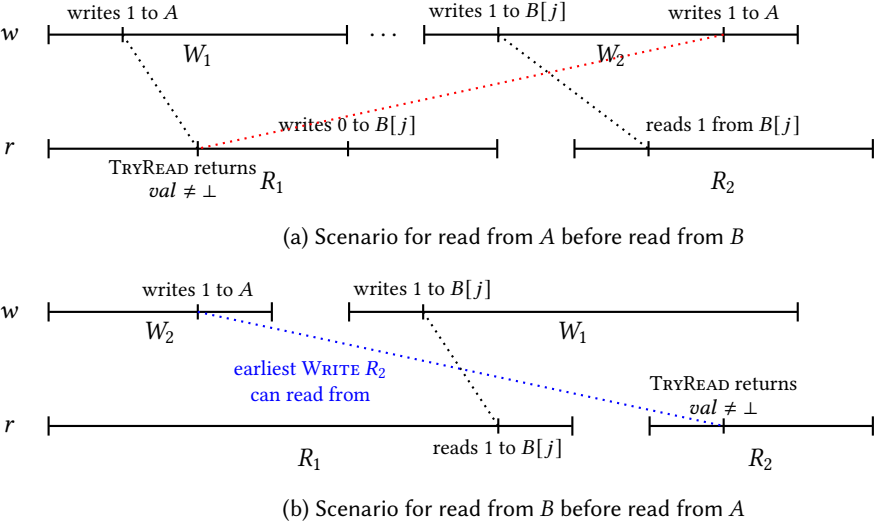

Algorithm~\ref{alg:wait-free-quiescent} is clearly wait-free. 
History independence (proved in Appendix~\ref{app:swsr wait-free register proofs}) 
follows by showing that if the writer writes 1 to an index in $B$ in a \textsc{Write} operation, 
then this write is overwritten either by the writer 
or by an overlapping \textsc{Read} operation. 
This implies that in a quiescent configuration all entries in $B$ are equal to 0.
We next argue linearizability.

\begin{restatable}{theorem}{waitfreeqhireg}
\label{thm:wait-free-q-hi-reg}
Algorithm~\ref{alg:wait-free-quiescent} is a linearizable wait-free 
quiescent HI SWSR multi-valued register from binary registers.
\end{restatable}

\begin{proof}[Proof of linearizabilty]
    By construction, $H$ is in the sequential specification of the register and by Lemma~\ref{lem:val-not-bot},
    it includes all completed operations in $H(\alpha)$.
    It is left to show that the linearization respects the real-time order of non-overlapping operations.
    The order between two \textsc{Write} operations respects the real-time order by the construction of $H$, and a \textsc{Read} operation cannot be placed after a \textsc{Write} operation that follows it, as a reader cannot read from the future.
    
    The primitive read and write operations on array $A$ follow exactly Algorithm~\ref{alg:lock-free-op-quiescent}. Additionally, the projection of $H$ that includes all \textsc{Write} operations and only \textsc{Read} operations that read $A$, is built in the same manner as a linearization of an execution of Algorithm~\ref{alg:lock-free-op-quiescent}. Therefore, 
    %by Lemma~\ref{lem:lock-free-linearization}, 
    the order of the \textsc{Write} operations and \textsc{Read} operations that read $A$ respects the real-time order among themselves. Next, we consider the order between two operations, where one of the operations is a \textsc{Read} operation that reads $B$.

    {\bf Write before read:} Consider a \textsc{Read} operation $R$ that reads $B$.
    By Lemma~\ref{lem:reads-B-overlaps}, $R$ reads $B$ from \textsc{Write} operation $W_1$ that overlaps $R$, 
    and by the construction of $H$, 
    if $R$ is placed between $W_1$ and $W_2$, the \textsc{Write} operation that precedes $W_1$. 
    Thus, $R$ is placed after any \textsc{Write} operation that precedes it in $\alpha$.

    {\bf Read before read:} assume a \textsc{Read} operation $R_1$ returns before a \textsc{Read} operation $R_2$ begins, 
    but $R_2$ is placed before $R_1$ in the linearization.
    There are three cases:
    
    {\bf (1) both $R_1$ and $R_2$ read $B$:} 
    Let $W_1$ be the \textsc{Write} operation $R_1$ reads $B$ from and let $W_2$ be the \textsc{Write} operation $R_2$ reads $B$ from.
    Since $R_2$ is placed before $R_1$ in the linearization, $W_2$ precedes $W_1$ in $\alpha$.
    By Lemma~\ref{lem:reads-B-overlaps}, $W_1$ overlaps $R_1$ and $W_2$ overlaps $R_2$.
    This implies that $W_1$, which begins after $W_2$ returns, begins before $R_1$ returns, but this contradicts that $W_2$ overlaps $R_2$.

    {\bf (2) $R_1$ reads $A$ and $R_2$ reads $B$:} 
    Let $W_2$ be the \textsc{Write} operation $R_2$ reads $B$ from (Figure~\ref{fig:thm:wait-free-q-hi-rega}). 
    $W_2$ writes 1 to $B[j]$ in Line~\ref{lin:write-1-B} after $R_1$ writes 0 to $B[j]$ in Line~\ref{lin:reader-read-B}. 
    Thus, $W_2$ first writes to $A$ after $R_1$ finishes reading $A$, and $R_1$ reads $A$ from \textsc{Write} operation $W_1$ that precedes $W_2$ in $\alpha$. By the construction of $H$, $R_2$ is placed between $W_2$ and the previous \textsc{Write} operation $W_3$ that directly precedes $W_2$. If $W_1 = W_3$, then by the construction of $H$, $R_1$ is placed before $R_2$ in $H$. Otherwise, $W_1$ precedes $W_3$ and this also implies that $R_1$ is placed before $R_2$ in $H$.

    {\bf (3) $R_1$ reads $B$ and $R_2$ reads $A$:} 
    Let $W_1$ be the \textsc{Write} operation $R_1$ reads $B$ from (Figure~\ref{fig:thm:wait-free-q-hi-regb}). By Lemma~\ref{lem:reads-B-overlaps}, $W_1$ overlaps $R_1$. Let $W_2$ be the \textsc{Write} operation that directly precedes $W_1$, then $R_1$ is placed between $W_1$ and $W_2$. Since $W_1$ starts before $R_2$ starts, the earliest \textsc{Write} operation $R_2$ can read $A$ from is $W_2$, that is, $R_2$ reads $A$ from $W_2$ or from \textsc{Write} operation that follows $W_2$ in $\alpha$.
    By the construction of $H$, $R_1$ is placed before $R_2$.
\end{proof}
\section{History Independent Implementations % of Objects 
    from Smaller Base Objects}
\label{sec:small-base-objs}

%In this section we study whether or not it is possible to implement a concurrent object $A$ from smaller base objects $B$, in a history-independent manner. 
In this section, we show that for a large class of objects, 
a reasonably strong notion of history independence---
\emph{state-quiescent history independence} (Definition~\ref{def:op-quiescentSHI})---cannot be achieved from smaller base objects, 
if we require wait-freedom. 

%We show that objects in the class $\calC_t$ cannot be implemented from base objects that have fewer than $t$ states, if we require both wait-freedom and state-quiescent HI. However, if we relax one of these two conditions, then there is at least one object in the class $\calC_t$ --- a $t$-valued single-reader single-writer register --- that can be implemented from smaller base objects (Boolean registers).

%%% ROTEM: Saving this to move to the intro:
%Informally, this class includes objects with a read operation that does not change the state of the object, and the state of the object can change to any other state with a single (different) operation.
%This class includes, for example, simple read-write registers and a CAS object that supports a simple read operation

\subsection{The class $\calC_t$}
Informally, our impossibility result applies to all objects 
with the following properties:
\begin{itemize}
    \item The object has a ``non-trivial'' read operation, which is able to distinguish between $t$ different subsets of the object's possible states; and
    \item The object can be ``moved freely'' from any state to any other state, in a single operation.
\end{itemize}
In fact, the impossibility result applies to other objects, 
including a queue, which do not fall into this class
%the class $\calC_t$
because they cannot be moved from any possible state to any other possible state in a single operation; for example, if a queue currently has two elements, we cannot reach the state where it is empty in one operation. 
% \Rnote{I think the following is true:
% Also, if we restrict to \emph{strongly-linearizable}
% implementations \TODO{cite},
% then the result applies even if the $\Readop$ operation 
% can change the state of the object, provided it must
% return a value that does not depend on the new state of the object. With this generalization,
% strongly-linearizable implementations of
% objects such as swap,
% or compare-and-swap that returns the current value upon failure, are also subject to the impossibility result.
% }
% \Rnote{Reason I think this is true is that it doesn't change anything in the proof. We need to have a well-defined ``state'' that the object is in after a prefix, and that's why I asked for strong linearizability. But other than that it doesn't matter: the read is supposed to return some value that doesn't depend on the new state that it changes the object to,
% so it can never return, because it can never find such a value. But this needs to be formalized and proven.}
For simplicity, we present here the proof for the simpler,
more restricted class described above,
and
we discuss a generalization for a queue in Section~\ref{sec:queue-impossibility}.

The class $\calC_t$ is formally defined as follows:
\begin{definition}[The class $\calC_t$]
An object $\calO$ is in the class $\calC_t$
for $t \geq 2$ if its state space $Q$
can be partitioned into $t$ nonempty subsets $X_1,\ldots,X_t$,
such that
\begin{itemize}
    \item The object has some operation $\Readop$,
    which does not change the state of the object,
    such that for any two states $q_i \in X_i, q_j \in X_j$
    where $i \neq j$,
    the response to $\Readop$ from state $q_i$
    differs from the response to $\Readop$ from state $q_j$.
    \item For any two states $q \neq q' \in Q$
    there is some operation $\Changeop(q, q')$
    that causes the object to transition
    from state $q$ to state $q'$.
\end{itemize}
\end{definition}

An object in the class $\calC_t$ has at least $t$ different states, and any pair of states are mutually reachable from each other by a single operation. 
Thus, the memory representations that arise from an implementation of 
an object in $\calC_t$, $t\geq 3$, from base objects with fewer than $t$ states cannot all be at distance 1 from each other. 
By Proposition~\ref{prop:p-hi} we obtain:

\begin{restatable}{proposition}{philb}
\label{prop:p-hi-lb}
There is no obstruction-free perfect HI implementation of an object in 
$\calC_t$, $t \geq 3$, from base objects with fewer than $t$ states.
\end{restatable}

\begin{proof}
    Consider an object in $\calC_t$, $t \geq 3$, and for each $1\leq i \leq t$, choose a state $q_i \in X_i$. Let $\can{q_1}, \ldots, \can{q_t}$ be a possible configuration of canonical memory representations that arise from a perfect HI implementation. 
    Since $\Changeop$ returns a different response from $q_i$ and from $q_j$, $i\neq j$, $q_1, \ldots, q_t$ are $t$ distinct states that also induce $t$ distinct memory representations, since an $\Changeop$ that runs solo from the same memory representation must return the same response.
    By Proposition~\ref{prop:p-hi}, 
    this implies that for every $2\leq i \leq t$, there is an index $j_i$ such that $\can{q_1}[j_i] \neq \can{q_i}[j_i]$, and for any other index, the memory representations are equal.
    Assume there are $2\leq i,k\leq t$ such that $j_i \neq j_k$, this implies that $\can{q_i}[j_k] = \can{q_1}[j_k] \neq \can{q_k}[j_k]$ and  $\can{q_k}[j_i] = \can{q_1}[j_i] \neq \can{q_i}[j_i]$ and therefore, $\can{q_i}$ and $\can{q_k}$ are at distance at least 2. Thus, for every, $2\leq i,k\leq t$, $j_i = j_k$.
    However, since the base object has less than $t$ state this implies that for some $i\neq k$, $\can{q_i} = \can{q_k}$, 
    which is a contradiction.
\end{proof}

%The ``non-triviality'' condition for read operations is formalized as follows:

%\begin{definition}
%Let $\calO$ be an object with state space $Q$ and response space $R$. \Rnote{Need to define this in the prelim.}
%An operation $o\in \calO$ is called \emph{$t$-separating}, for a positive integer $t\leq |Q|$, if 
%the state space $Q$ can be partitioned into $t$
%non-empty subsets $X_1,\dots, X_t$,
%such that for each $1 \leq i \leq t$,
%there is a distinct response $r_i \in R$ such that for every state $q \in X_i$, $\Delta(o,q) = (\wildcard, r_i)$. 
%In addition, for any $1\leq i < j \leq t$, $a_i \neq a_j$.
%\Rnote{Suspect this can be simplified to require that possible responses for $X_i$ be disjoint from responses in $X_j$ when $i \neq j$. Check later.}
%\end{definition}

%The class $\calC_t$ includes all objects that \TODO{what}.
%\begin{definition}
%An object is in class $\mathcal{C}_t$ if its operations include the next two operations:
%    \begin{itemize}
%        \item A $t$-separating read-only operation \textsc{Read}$()$ with no input. 
%        \Rnote{Why no input? What if some input is needed, like for CAS?}
        
%        \item $\textsc{Change}(i)$ operation, with input $i$ from domain $I$, such that for every pair of states $q_1, q_2\in Q$, there is input $i\in I$ such that $\Delta(q_1, \textsc{Change}(i)) = (q_2, \wildcard)$.
%    \end{itemize}
%\end{definition}

\subsubsection*{Examples of objects in the class $\calC_t$}

A $t$-valued read/write register is in the class $\calC_t$:
it has $t$ different states, each representing the value the register holds, and the $\Read$ operation distinguishes between them;
the $\Write$ operation can move us from any state (i.e., any value) to any other state in a single operation.
We have already seen in Section~\ref{sec:swsrexample} that multi-valued registers \emph{can} be implemented from binary registers, if we weaken either the progress or the history independence requirements.
Our results for SWSR multi-valued registers are summarized in Table~\ref{tab:regImplementationSumm}.
\begin{table}[tb]
\centering\small
\begin{tabular}{|c|c|c|c|}
\hline
Perfect HI (Def.~\ref{def:perfectSHI}) & State-quiescent HI (Def.~\ref{def:op-quiescentSHI}) & Quiescent HI (Def.~\ref{def:quiescentSHI}) & Progress \\ \hline
Impossible (Prop.~\ref{prop:p-hi-lb}) & Impossible (Cor.~\ref{thm:oq-hi-lb}) & Possible (Alg.~\ref{alg:wait-free-quiescent}) & Wait-free \\ \hline
Impossible (Prop.~\ref{prop:p-hi-lb}) & Possible (Alg.~\ref{alg:lock-free-op-quiescent}) & Possible (Alg.~\ref{alg:lock-free-op-quiescent} \& Alg.~\ref{alg:wait-free-quiescent}) &  Lock-free \\ \hline
\end{tabular}
\caption{Summary of results for implementing a SWSR multi-valued register from binary registers}
%\vspace{-20pt}
\label{tab:regImplementationSumm}
\end{table}

Another example
of an object in the class $\calC_t$
is a $t$-valued {\CAS} object that supports a read operation:
the state is again the current value of the $\CAS$,
and the read operation distinguishes between all $t$ possible values;
as for the $\Changeop$ operation,
we can move from any state $q$ to any state $q'$
by invoking $\CAS(X, q, q')$.
%The state is the value the object holds and $\Readop = \textsc{Read}$.
%For {\CAS} object $X$, where the current value of the object is $old$, $
%\CAS(X,old,new)$ operation replaces the value of the object with $new$.

% \Hnote{what is the next comment good for?}
% The parameter $t$ allows the \textsc{Read} operation to return less information than the full value stored in the object, for example, it can return a response indicating the value is between some upper and lower bound. 

To illustrate the importance of the state-connectivity requirement 
in the definition of the class $\calC_t$, 
we argue that a \emph{max register}~\cite{AspnesAC2009}, 
returning the maximum value ever written in it,  is \emph{not} in the class.
The state space of a max register is not well-connected: as soon as we visit state $m$,
the object can never go back to a state smaller than $m$.
A simple modification to Algorithm~\ref{alg:multi-valued-baseline}, 
where the writer only writes to $A$ if the new value is bigger than all the values 
it has written in the past, results in a wait-free state-quiescent HI max register from binary registers.

Another object that is not in the class $\calC_t$ is a \emph{set} over $t$ elements, with insert, remove and lookup operations. 
Even though the set has $2^t$ possible abstract states, its operations return only two responses, ``success'' or ``failure''; thus, we cannot \emph{distinguish} via a single operation between all $2^t$ states, or even between $t$ states (the number of elements that could be in the set).
There is a simple wait-free perfect HI implementation of a set 
over the domain $\set{1,\ldots,t}$, using $t$ binary registers:
we simply represent the set as an array $S$
of length $t$, with $S[i] = 1$ if and only if element $i$ is in the set,
with the obvious implementation of insert, delete and lookup.
%Insert operations write 1 to the location of the element inserted, delete operations write 0, 
%and lookup operations read the array $S$ at the index of the element queried, and return the value they read.

\subsection{Impossibility of Wait-Free, State-Quiescent HI Implementations for the Class $\calC_t$}
\label{sec:gen-lb}

Let $\calO  \in \calC_t$ be a high-level object with state space $Q$. 
Consider a wait-free state-quiescent HI implementation of $\calO$
using $m \geq 1$ base objects $\obj_1,\ldots, \obj_m$.
For each base object $\obj_i$,
let $Q_i$ be the state space of $\obj_i$;
we assume that $|Q_i| \leq t-1$.
This is the only assumption we make about the base objects, 
%Our impossibility result makes no assumptions on the base objects other than requiring that they have fewer than $t$ possible states. 
and our impossibility result applies to arbitrary read-modify-write base objects 
as well as to simple read-write base registers.
Let $h$ be a linearization function for the implementation.
%mapping execution prefixes to a sequence of high-level operations.
%\NNote{I don't think we need $h$ since we say below we know the order of writes in all linearizations (and it is not explicitly used).}

%As there are only $t - 1$ possible states for each base object
%and $t$ subsets $X_1,\ldots,X_t$ of high-level states,
%for every base object $\obj_\ell$
%there exist $i \neq j$ and high-level states $q_i \in X_i, q_j \in X_j$
%such that $\can{q_i}[\ell]=\can{q_j}[\ell]$.
%\NNote{Maybe we can remove this whole paragraph? we defined $\can{q}$ already and the second part is re-stated in the proof below.}

% We consider two processes, the reader $r$ that executes the (single) $read$ operation and the writer $w$ that executes the $write$ operations.

We consider executions with two processes:
\begin{itemize}
    \item A ``reader'' process $r$, which executes a single $\Readop$ operation, and
    \item A ``changer'' process $c$, which repeatedly invokes $\Changeop$ operations.
\end{itemize}

%\Rnote{Low-level detail whose importance cannot be seen at this level of abstraction. Can be moved to the appendix, including the definition of ``active'':}
%\Hnote{active is not used anywhere, I removed it}
For the purpose of the impossibility result, we assume, 
%without loss of generality, 
that the local state of a process $p_i$ contains the complete history of $p_i$'s invocations and responses. 
%In particular, this indicates whether $p_i$ has a pending operation.
%The reader process $r$ is \emph{active} in configuration $C$ if it has yet to return from its single $\Readop$ operation.
Our goal is to show that we can construct an execution where $r$ % is always active, 
does not return from its single $\Readop$ operation,
violating wait-freedom.

The executions that we construct have the following form:
\begin{equation*}
    \alpha_{q_0,\ldots,q_k} = \Changeop(q_0, q_1), r_1,
    \Changeop(q_1, q_2), r_2, \ldots, \Changeop(q_{k-1}, q_k), r_k
\end{equation*}
where $\Changeop(q_i, q_{i+1})$ is an operation
executed by the changer process during which the reader process takes no steps,
and $r_i$ is a single step by the reader process.
The reader
executes a single $\Readop$
operation that is invoked immediately
after the first $\Changeop$
operation completes, and we will argue that the reader never returns.

In any linearization of $\alpha_{q_0,\ldots,q_k}$,
the operations $\Changeop(q_0,q_1),\ldots,\Changeop(q_{k-1},q_k)$
must be linearized in order,
as they do not overlap.
Furthermore, the $\Readop$
operation carried out by the reader is not state-changing.
Thus, the linearization of $\alpha_{q_0,\ldots,q_k}$ ends
with the object in state $q_k$,
and we abuse the terminology by saying that the execution ``ends at state $q_k$''.

We say that execution $\alpha_{q_0,\ldots,q_k}$
\emph{avoids} a subset $X \subseteq Q$
if $\set{q_1,\ldots,q_k} \cap X = \emptyset$.
(Note that we may have $q_0 \in X$ and still 
say that $\alpha_{q_0,\ldots,q_k}$ avoids $X$;
this is fine for our purposes, because the reader 
only starts running after the first $\Changeop$ operation completes.)

\begin{lemma}
    There exists a partition of the possible return values $R$
    for the $\Readop$ into subsets $R_1,\ldots,R_t$,%
    \footnote{We assume there are no unused values in $R$,
    that is, for any $r \in R$,
    there is some state $q \in Q$ such that
    when $\Readop$ is executed from state $q$,
    it returns $r$.}
    such that if an execution $\alpha_{q_0,\ldots,q_k}$ avoids
    $X_i \subseteq Q$,
    then the $\Readop$ operation cannot return any value
    from $R_i$ at any point in $\alpha_{q_0,\ldots,q_k}$.
    \label{lemma:exclusive_returns}
\end{lemma}
\begin{proof}
    For each $1 \leq i \leq t$,
    let $R_i$ be the set of values $r$
    such that for some state $q \in X_i$,
    the $\Readop$ operation returns $r$
    when executed from state $q$.
    By the definition of the class $\calC_t$,
    the sets $R_1,\ldots,R_t$
    are disjoint,
    and since the sets $X_1,\ldots,X_t$
    partition the state space $Q$,
    the sets
    $R_1,\ldots,R_t$
    partition the set of responses $R$.

    Fix an execution $\alpha_{q_0,\ldots,q_k}$
    that avoids $X_i$,
    and recall that in any linearization,
    the operations
    $\Changeop(q_0,q_1),\ldots,\Changeop(q_{k-1},q_k)$
    must be linearized in-order,
    as they are non-overlapping operations by the same process.
    The $\Readop$ operation cannot be linearized before
    the first operation $\Changeop(q_0, q_1)$,
    because it is only invoked after this operation completes.
    Thus, the $\Readop$ operation 
    either does not return in $\alpha_{q_0,\ldots,q_k}$,
    or it is linearized after some operation
    $\Changeop(q_j, q_{j+1})$ where $j \geq 0$.
    In the latter case, let $\ell \neq i$ be the index such that $q_{j+1}\in X_{\ell}$;
    we know that $\ell \neq i$
    as $\alpha_{q_0,\ldots,q_k}$ avoids $X_i$.
    The value returned by the $\Readop$
    is in the set $R_{\ell}$,
    which is disjoint from $R_i$.
    Therefore in $\alpha_{q_0,\ldots,q_k}$ the
    $\Readop$ operation either does not return,
    or returns a value that is not in $R_i$.
\end{proof}

%\TODO{Move this def to the preliminaries, probably. In the prelim it doesn't need to be in the definition environment, either.}
%\Hnote{I kept it here, but downgraded.}

%\begin{definition}

%  The relation $\localind{p}$ is transitive for a fixed $p$.
%\end{definition}

Using the fact that each base object has at most $t -1$
possible states, we can construct 
$t$ arbitrarily long executions
that the reader cannot distinguish from one another,
such that each subset $X_i$ is avoided by one of the $t$ executions.
Two execution prefixes $\alpha_1$ and $\alpha_2$ are \emph{indistinguishable} 
to the reader, denoted $\alpha_1 \localind{r} \alpha_2$, 
if the reader is in the same state in the final configurations of $\alpha_1$ and $\alpha_2$.

The construction is inductive, with each step extending the executions
by one operation and a single step of the reader:
\begin{lemma}
Fix $k \geq 0$,
and suppose we are given $t$
    executions of the form
    $\alpha_i = \alpha_{q_0^i,\ldots,q_k^i}$
    for $i = 1,\ldots,t$,
    such that $\alpha_1 \localind{r} \ldots \localind{r} \alpha_t$,
    and each $\alpha_i$ avoids $X_i$.
    Then we can extend
    each $\alpha_i$
    into an execution 
    $\alpha_i' = \alpha_{q_0^i,\ldots,q_{k+1}^i}$
    that also avoids $X_i$,
    such that $\alpha_1' \localind{r} \ldots \localind{r} \alpha_t'$.
    \label{lemma:exclusive_executions}
\end{lemma}

\begin{proof}
    Let $\alpha_i = \alpha_{q_0^i, \ldots, q_k^i}$
    for $i = 1,\ldots,t$ be executions
    satisfying the conditions of the lemma,
    and let us construct extensions $\alpha_i' = \alpha_{q_0^i, \ldots, q_k^i, q_{k+1}^i}$ for each $i = 1,\ldots,t$.
    By assumption, the reader is in the same local state
    at the end of all executions $\alpha_i$ for $1 \leq i \leq t$,
    and so its next step is the same in all of them.
    Our goal is to 
    choose a next state $q_{k+1}^i$
    for each $i = 1,\ldots,t$,
    and extend each $\alpha_i = \alpha_{q_0^i,\ldots,q_k^i}$
    into $\alpha_i' = \alpha_{q_0^i,\ldots,q_k^i}$
    by appending an operation $\Changeop(q_k^i, q_{k+1}^i)$,
    followed by a single step of the reader.
    We must do so in a way that continues to avoid
    $X_i$, and maintains indistinguishability to the reader.

    Since the implementation of $\calO$ is state-quiescent HI
    and since each execution $\alpha_i$
    ends in a state-quiescent configuration,
    if $\alpha_i$
    ends in state $q$,
    then the memory must be in its canonical representation, $\can{q}$
    (as defined in Section~\ref{sec:prelim}).
     %   , for any high-level state $q \in Q$
%there is a unique memory representation $\can{q}$,
%such that for any execution prefix $\alpha$ \ns{with no pending state-changing operation},
%if the linearization of $\alpha$
%ends at state $q$,
%then at the end of $\alpha$ the memory representation is $\can{q}$.
    
    Let $\obj_{\ell}$
    be the base object
    accessed by the reader in its next step
    in all $t$ executions.
    Because $\obj_{\ell}$
    has only $t - 1$
    possible memory states
    and there are $t$ subsets $X_1,\ldots,X_t$,
    there must exist two distinct subsets $X_j, X_{j'}$
    ($j \neq j'$)
    and two states $q \in X_j, q' \in X_{j'}$
    such that
    $\can{q}[\ell] = \can{q'}[\ell]$.
    For every $1 \leq i \leq t$,
    there is a state $q_{k+1}^i \in \set{q,q'}$
    such that $q_{k+1}^i \notin X_i$:
    if $i \notin \set{j, j'}$
    then we choose between $q$ and $q'$
    arbitrarily, and if $i = j$ or $i = j'$
    then we choose $q_{k+1}^i = q'$
    or $q_{k+1}^i = q$, respectively.
    
    We extend each $\alpha_i = \alpha_{q_0^i,\ldots,q_{k}^i}$
    into $\alpha_i' = \alpha_{q_0^i,\ldots,q_k^i, q_{k+1}^i}$
    by appending
    a complete $\Changeop(q_k^i, q_{k+1}^i)$
    operation, followed by a single
    step of the reader.
    The resulting execution $\alpha_i'$
    still avoids $X_i$,
    as we had $\set{ q_1^i,\ldots, q_k^i } \cap X_i = \emptyset$,
    and the new state also satisfies $q_{k+1}^i \notin X_i$.
    Moreover,
    when the reader takes its step,
    it observes the same state for the base
    object $\obj_{\ell}$ that it accesses in all executions,
    as all of them end in either state $q$
    or state $q'$,
    and $\can{q}[\ell] = \can{q'}[\ell]$.
    Therefore, the reader cannot distinguish
    the new executions from one another.
\end{proof}

By repeatedly applying Lemma~\ref{lemma:exclusive_executions},
we can construct arbitrarily long executions,
with the reader taking more and more steps
(since in an execution $\alpha_{q_0,\ldots,q_k}$
the reader takes $k$ steps) but never returning.
This contradicts the wait-freedom of the implementation,
to yield:
% To prove the impossibility result, 
% we show how to force the $\Readop$ operation to never return.

\begin{restatable}{theorem}{oqhigenlb}
\label{thm:oq-hi-gen-lb}
For any object $\calO$
in the class $\calC_t$,
there is no wait-free implementation 
that is state-quiescent HI using base objects with fewer than $t$
states.
\end{restatable}

\begin{proof}
    We construct $t$
    arbitrarily long executions, in each of which an $\Readop$
    operation takes infinitely many steps but never returns.
    The construction uses Lemma~\ref{lemma:exclusive_executions}
    inductively:
    we begin with empty executions,
    $\alpha_1^0 = \ldots = \alpha_t^0 = \alpha_{q_0}$.
    These executions trivially satisfy the conditions of Lemma~\ref{lemma:exclusive_executions},
    as each $\alpha_i^0$ avoids $X_i$ (technically,
    it avoids \emph{all} subsets $X_j$),
    and furthermore,
    since the reader has yet to take a single step
    in any of them,
    it is in the same local state in all executions.
    We repeatedly apply Lemma~\ref{lemma:exclusive_executions}
    to extend these executions,
    obtaining for each $k \geq 0$
    a collection of
    $t$  executions
    $\alpha_1^k = \alpha_{q_0^1,q_1^1,\ldots},\ldots,\alpha_t^k = \alpha_{q_0^t,q_1^t,\ldots}$,
    such that each $\alpha_i^k$ avoids $X_i$,
    and the reader cannot distinguish the executions from one another.

    Suppose for the sake of contradiction
    that the reader returns a value
    $r$ at some point in $\alpha_1^k$.
    Then it returns the same value $r$
    at some point in each execution $\alpha_i^k$
    for each $i = 1,\ldots,t$,
    as it cannot distinguish these executions,
    and its local state encodes
    all the steps it has taken,
    including whether it has 
    returned a value,
    and if so, what value.
     Let $R_1,\ldots,R_t$
    be the partition from Lemma~\ref{lemma:exclusive_returns}.
    By the lemma,
    for each $i = 1,\ldots,t$,
    since execution $\alpha_i^k$
    avoids $X_i$,
    we must have $r \notin X_i$.
    But this means that there is no value
    that the reader can return, a contradiction.

    Continuing on in this way,
    we can construct arbitrarily long executions,
    with the reader taking more and more steps
    (since in an execution $\alpha_{q_0,\ldots,q_k}$
    the reader takes $k$ steps)
    but never returning.
    This contradicts the wait-freedom of the implementation.    
\end{proof}

\subsection{Multi-Valued Register}

For read/write registers, we obtain the following impossibility result:

\begin{corollary}
\label{thm:oq-hi-lb}
    There is no wait-free state-quiescent HI implementation of a $t$-valued register, $t\geq 3$, from binary registers.
\end{corollary}

This impossibility result is complemented by the register implementations we described in Section~\ref{sec:swsrexample},
where we relaxed either the progress condition
or the history independence condition to circumvent the impossibility.
%The algorithms of the previous section complement this result,
%showing that if we relax one of these two conditions,  
%wait-freedom and state-quiescent HI, 
%then an object in $\calC_t$ can be implemented from smaller base objects 
%(Boolean registers).
%This is due to the state-quiescent history-independent \emph{lock-free} implementation, 
%and the wait-free implementation satisfying a weaker %definition of history independence 
%(quiescent history independence, Definition~\ref{def:quiescentSHI}).
%
%\Rnote{This is really out of place here. Maybe move it to the end of Section~\ref{sec:swsrexample}, where it is already mentioned.}
Corollary~\ref{thm:oq-hi-lb} also implies:
% that the writer must leave the memory in a canonical representation, regardless of the reader. 
%Thus, when the reader does not write in a wait-free quiescent HI multi-valued register implementation from binary registers, 
%the implementation is quiescent HI.

\begin{restatable}{proposition}{readerMustWrite}
\label{lem:reader-must-write}
The reader must write in any wait-free quiescent HI implementation of a SWSR multi-valued register from binary registers.
\end{restatable}

\begin{proof}
    Assume there is a wait-free quiescent HI implementation with linearization function $h$ where the reader never writes to the shared memory. 
    Let $\alpha$ be a finite execution of the implementation, which ends with a configuration without a pending $\textsc{Write}$ operation.
    If there is also no pending $\textsc{Read}$ operation in $\alpha$, the execution ends in a quiescent configuration and the memory of this configuration is equal to the canonical memory representation of the value $\seqstate{h(\alpha)}$.

    Otherwise, since the implementation is wait-free, 
    there is an extension of $\alpha$, $\beta = \alpha \alpha'$, 
    such that only the reader takes steps in $\alpha'$ and completes the pending $\textsc{Read}$ operation, and $\beta$ ends in a quiescent configuration.
    Since the reader does not write to the shared memory, and the writer did not take steps in $\alpha'$, $\mem{\alpha} = \mem{\beta}$.
    In addition, the last high-level value written in $\alpha$ and $\beta$ are identical, i.e., $\seqstate{h(\alpha)} = \seqstate{h(\beta)}$.
    Since $\beta$ ends with a quiescent configuration, the memory of this configuration is equal to the canonical memory representation of the last written value in $\beta$ and the same is true for $\alpha$.
    This shows the implementation is state-quiescent HI, 
    contradicting Corollary~\ref{thm:oq-hi-lb}, 
    which states that there is no wait-free state-quiescent HI implementation.
\end{proof}

\subsection{Extension to a Queue}
\label{sec:queue-impossibility}

The state space of a \emph{queue} is described as a list, that is, an ordered set of elements.
An $\enq(v)$ operation adds the element $v$ to the end of the list and a $\deq()$ operation removes the first element from the list and returns it. We consider a queue that supports a $\peek()$ operation, 
which returns the first element in the list, without changing the object's state.
The elements added to the queue are taken from the finite domain $\set{1,\dots, t}$, $t\geq 2$.
The response space $R = \set{r_0, \ldots, r_t}$ is of size $t+1$, where response $r_i = i$, $1\leq i \leq t$, indicates that $i$ is the first element in the list, and response $r_0 = \emptyset$ indicates that the queue is empty. For simplicity, response $r_0$ is also used as a default response from an {\enq} operation. 

\ns{A queue is not in the class $\calC_t$ since not all states are reachable from each other by a single operation. Items can be added to the queue using $\enq$ operations without changing the value of the item in the front of the queue, which is the response from a $\peek$ operation.
We can use this property to move from state to state, without going through a state that returns a response we wish to avoid. 
To fully control the response values the $\peek$ operation can return in the executions we construct, we pick a \emph{representative} state for each response value, instead of partitioning the whole state space, and construct executions alternating between these states.}

We pick $t+1$ {representative} states from $Q$, $q_0, \dots q_t$, such that the response of {\peek} from state $q_i$, $0\leq i\leq t$, is $r_i$; let $q_i = \set{i}$, $1\leq i \leq t$, that is, a queue that contains exactly one element $i$, and $q_0 = \emptyset$ is the initial state where the queue is empty.

For $i_1\neq i_2$, $0\leq i_1,i_2 \leq t$, define a sequence of operations $S(i_1,i_2)$ that moves the object's state from $q_{i_1}$ to $q_{i_2}$ as follows:
\begin{enumerate}
    \item \underline{If $i_1 = 0$:} Define $S(i_1,i_2) = \enq(i_2)$. 
    \item \underline{If $i_2 = 0$:} Define $S(i_1,i_2) = \deq()$. 
    \item \underline{Otherwise, if $i_1,i_2 \neq 0$:} Define $S(i_1,i_2) = \enq(i_2), \deq()$. 
\end{enumerate}

If $i_1 = 0$, $\Delta(\emptyset, \enq(i_2)) = (\set{i_2},\wildcard)$ and for $i_2 = 0$, $\Delta(\set{i_1}, \deq()) = (\emptyset,\wildcard)$. 
If $i_1,i_2 \neq 0$, we have that $\Delta(\set{i_1}, \enq(i_2)) = (\set{i_1,i_2},\wildcard)$ and then, $\Delta(\set{i_1,i_2}, \deq()) = (\set{i_2},\wildcard)$. 
That is, the sequence of operations goes through a third state $\set{i_1,i_2}$, different than $q_0,\dots, q_t$, such that a $\peek$ operation returns the same response from this state and state $q_{i_1}$. Therefore, even though this sequence has two operations and not one, the response the $\peek$ operation can return goes from $r_{i_1}$ to $r_{i_2}$, without going through a third response value.
This allows to provide a lower bound for a queue with a similar proof to the one for class $\calC_t$, where the changer moves the object's state to a representative state between each step of the reader.
For more details and proof of the next theorem, refer to Appendix~\ref{app:queue-impossibility}.

\begin{restatable}{theorem}{queue}
\label{thm:queue}
There is no wait-free implementation of a queue with a {\peek} operation and elements from domain $\set{1,\ldots,t}$
that is state-quiescent HI using base objects with fewer than $t+1$
states.
\end{restatable}

\section{A History-Independent Universal Implementation}
\label{sec:universal-impl}

In the previous sections, we considered history-independent implementations
of objects from base objects that
are too small to store the state of the abstract object in its entirety, and showed that certain tradeoffs are unavoidable in this setting: for many objects, one must sacrifice either wait-freedom or state-quiescent history independence. We now turn to study large base objects, which can store the entire state of the abstract object, together with auxiliary information; we show that in this regime, a wait-free implementation that is state-quiescent HI \emph{is} possible.
Our implementation actually satisfies a somewhat stronger property than state-quiescent HI: 
%the observer cannot distinguish between a prefix of the execution that does not contain an ongoing operation that is currently pending when replaced with a different prefix, which leads to the same state as the original prefix. In other words, 
at any point in the execution,
%the observer may only learn about pending operations, 
%operations that overlapped with pending operations, 
%and states that the object was in during pending operations. 
the observer cannot gain information about operations that completed before a 
pending operation started, 
except for the state of the object when the earliest pending operation began.
%This property corresponds to weak history independence, since the observer can learn additional information by accessing two arbitrary points in the execution.

\ha{
When the full state of the object can be stored in a single memory cell,  
there is a simple lock-free universal implementation,
using \emph{load-link/store-conditional} (LLSC).\footnote{
    In hardware, \emph{load-linked} reads a memory cell, 
    while \emph{store-conditional} changes this memory cell, 
    provided that it was not written since the process' 
    most recent load-linked.}
The current state of the object is stored in a single cell, 
an operation reads the current value of this cell, using \emph{LL},
and then tries to write the new value of the object (after applying its changes)
in this cell, using \emph{SC}. 

This implementation is clearly perfect HI.
However, it is not wait-free since an operation may repeatedly fail, 
since other operations may modify the memory cell in between its LL and its SC. 
The standard way to make the universal implementation wait-free relies 
on~\emph{helping}~\cite{HerlihyWaitFreeSynch,FatourouKallimanisSPAA11}:
When starting, an operation \emph{announces} its type and arguments in shared memory.
Operations check whether other processes have pending operations 
and help them to complete,
obtaining the necessary information from their announcement;
after helping an operation to complete, 
they store a response to be returned later. 
This breaches history independence,
revealing the type and arguments of prior and pending operations, 
as well as the responses of some completed operations. 

Our wait-free, history-independent universal implementation follows a similar approach,
but ensures that announcements and responses are cleared before operations complete, 
to guarantee that forbidden information are not left in shared memory. 
This is done very carefully in order to erase information only after it is 
no longer needed. 

To use the more standard and commonly-available atomic {\CAS},
we implement an abstraction of a \emph{context-aware} variant of 
LLSC~\cite{jayanti_et_al:LIPIcs.DISC.2023.25}, 
which explicitly manages the set of processes that have load-linked this cell as \emph{context}.
This again breaches history independence, as the context reveals information
about prior accesses. 
To erase this information, the implementation clears the context of a 
memory cell using an additional \emph{release} operation, 
added to the interface of context-aware \emph{releasable LLSC} (R-LLSC).
This allows us to obtain a \emph{wait-free state-quiescent HI 
universal implementation from atomic {\CAS}}.

The next section presents the universal implementation using linearizable R-LLSC objects.
Its basic properties are proved in Section~\ref{sec:universal-impl-correct},
including showing that the algorithm is state-quiescent HI, 
when the R-LLSC implementation is HI. 
We then give a lock-free R-LLSC perfect HI implementation 
from atomic {\CAS} (Section~\ref{sec:R-LLSC}), 
and obtain a \emph{wait-free state-quiescent HI universal} implementation from atomic 
{\CAS}, in Section~\ref{sec:universal-impl-wait-free}.}

\subsection{Universal HI Implementation from Linearizable Releasable LLSC}
\label{sec:universal-impl-llcs}

A context-aware \emph{load-link/store-conditional} (LLSC) object 
over a domain $V$ is defined as follows:
the state of an LLSC object $\mathcal{O}$ is the pair $(\mathcal{O}.val, \mathcal{O}.context)$, 
where $\mathcal{O}.val \in V$ is the value of the object,
and $\mathcal{O}.\mathit{context}$ is
a set of processes. 
The initial state is $(v_0,\emptyset)$, where $v_0 \in V$ is a designated initial value.
\ha{Process $p_i$ can perform the following operations:}
\begin{description}
    \item[${\LL}(\mathcal{O})$:]
    adds $p_i$ to $\mathcal{O}.context$ and returns $\mathcal{O}.val$.
    \item[${\VL}(\mathcal{O})$:] returns $true$ if $p_i \in \mathcal{O}.context$ and \textit{false} otherwise.
    % \item[${\RL}(\mathcal{O})$:] removes $p_i$ from $\mathcal{O}.context$ and returns $true$.
    \item[${\SC}(\mathcal{O},new)$:] if $p_i \in \mathcal{O}.context$, sets $\mathcal{O}.val = new$ and $\mathcal{O}.context = \emptyset$, and returns \textit{true}; otherwise, it returns \textit{false}.
    \item[${\Ld}(\mathcal{O})$:] returns $\mathcal{O}.val$ without changing $\mathcal{O}.context$.
    \item[${\St}(\mathcal{O},new)$:]
    sets $\mathcal{O}.val = new$ and $\mathcal{O}.context = \emptyset$ and returns \textit{true}.
\end{description}
\ha{A {\VL}, {\SC} or {\St} is \emph{successful} if it returns \textit{true};
note that {\St} is always successful.
${\Ld}$ and ${\St}$ operations are added to simplify the code and proof.}

\ha{
The universal implementation appears in Algorithm~\ref{alg:universal-construct-R-LLSC}. 
The right side of Lines~\ref{lin:ll-head1},~\ref{lin:ll-announce} 
and~\ref{lin:ll-head2}, marked in blue, 
as well as Lines~\ref{lin:rl-announce} and~\ref{lin:rl-head}, marked in red,
are \ha{only used to} ensure history independence; 
we ignore them for now, and explain their \ha{usage} later.
% The code uses the function $\textsf{apply}:(Q\times O) \to (Q\times R)$, defined as follows; $\textsf{apply}(q,o) = (q',r)$ such that $(q,o, q',r)\in \Delta$. This function returns the new state and response value that results from applying operation $o$ in state $q$, according to the sequential specification.
We assume that the set of possible responses $R$ of the object 
is disjoint from its set of operations $O$, 
i.e., $R \cap O = \emptyset$, and $\bot \notin R \cup O$.

An array $announce[1..n]$ stores information about pending operations,
while $\mathit{head}$ holds the current state of the object, 
along with some auxiliary information,
like the response to the most recently applied operation $o$, 
and the identifier of the process that invoked $o$.
In-between operations, the value of $head$ is $\tup{q,\bot}$, 
where $q$ is the current state of the object.}

\begin{algorithm}[tb]\small
\caption{\small State-quiescent HI universal implementation from R-LLSC : code for process $p_i$}
\label{alg:universal-construct-R-LLSC}
\begin{algorithmic}[1]
    \Statex $head$, R-LLSC variable initialized to $\langle q_0, \bot \rangle$, where $q_0$ is the initial state
    \Statex $announce[n]$, R-LLSC variable array all cells initialized to $\bot$  
    \Statex local $priority_i$, initialized to $i$
    \Statex %\hrulefill
    \Statex \textsc{ApplyReadOnly}$(op\in O)$:
    \Comment{Read-only operations}
    \State $\tup{q, \wildcard} \gets {\Ld}(head)$
    \label{lin:load-head}
    \State $\wildcard, rsp \gets \Delta(q, op)$
    \State \algorithmicreturn{} $rsp$
    \Statex
    \Statex \textsc{Apply}$(op\in O)$:
    \Comment{State-changing operations}
    \State ${\St}(announce[i], op)$
    \label{lin:announce-inv}
    \While{${\Ld}(announce[i]) \notin R$}
    \label{lin:start-loop}
    
        \algrenewcommand{\alglinenumber}[1]{\footnotesize #1L:}
        \State $\tup{q, r} \gets {\LL}(head)$  
        \label{lin:ll-head1}
        \hfill\makebox[.6\textwidth][l]{\textcolor{blue}{|| \quad \labeltext[\ref*{lin:ll-head1}R.1]{{\footnotesize \ref*{lin:ll-head1}R.1:}}{lin:wait-til1} \textbf{wait until} ${\Ld}(announce[i]) \notin R$}}
        \item[]\hfill \makebox[.6\textwidth][l]{\textcolor{blue}{\phantom{||} \quad {\footnotesize \ref*{lin:ll-head1}R.2:} goto Line~\ref{lin:read-rsp}}}
        
        \algrenewcommand{\alglinenumber}[1]{\footnotesize #1:}
        \If{$r = \bot$}
        \Comment{In-between operations}
        \label{lin:head-state-check}
        \State $help \gets {\Ld}(announce[priority_i])$
        \label{lin:read-help-announce}
        \If{$help \in O$} 
        \label{lin:cond-help}
            %\State 
            $apply\mhyphen{}op \gets help$; $j \gets {priority}_i$
            \Comment{Try to apply another process operation}
        \Else
            \If{${\Ld}(announce[i]) \notin O$}
            \label{lin:read-my-Announce}
                continue
                \Comment{Go to the beginning of the loop}
            \EndIf
            \State $apply\mhyphen{}op \gets op$; $j \gets i$ 
            \Comment{Try to apply \ha{your} own operation}
            \label{lin:read-other-Announce}
        \EndIf
        \State $state, rsp \gets \Delta(q, apply\mhyphen{}op)$
        \label{lin:apply}
        \If{${\SC}(head, \tup{state, \tup{rsp,j}}$)}
        \Comment{End of the first stage}
        \label{lin:sc-headB}
             \State $priority_i \gets
             (priority_i + 1) \bmod n$
             %\parens{priority_i + 1 \,\%\,n} + 1$
            %\Comment{Advance helping priority}
        \EndIf
    \Else
    \State $\tup{rsp,j} \gets r$

            \algrenewcommand{\alglinenumber}[1]{\footnotesize #1L:}
            \State $a \gets {\LL}(announce[j])$  
            \label{lin:ll-announce}
            \hfill\makebox[.6\textwidth][l]{\textcolor{blue}{|| \quad \labeltext[\ref*{lin:ll-announce}R.1]{{\footnotesize \ref*{lin:ll-announce}R.1:}}{lin:wait-til2} \textbf{wait until} ${\Ld}(announce[i]) \notin R$}}
            \item[]\hfill\makebox[.6\textwidth][l]{\textcolor{blue}{\phantom{||} \quad \labeltext[\ref*{lin:ll-announce}R.2]{{\footnotesize \ref*{lin:ll-announce}R.2:}}{lin:rl-announce-r} ${\RL}(announce[j])$}}
            \item[]\hfill\makebox[.6\textwidth][l]{\textcolor{blue}{\phantom{||} \quad {\footnotesize \ref*{lin:ll-announce}R.3:} goto Line~\ref{lin:read-rsp}}}
            
            \algrenewcommand{\alglinenumber}[1]{\footnotesize #1:}

    \If{${\VL}(head) = true$}
    \label{lin:vl-head}
        \If{$a\in O$}
            ${\SC}(announce[j], rsp)$
            \label{lin:replace-inv-rsp}
            \Comment{End of the second stage}
        \EndIf
        \State ${\SC}(head, \tup{q, \bot})$
        \Comment{End of the third stage}
        \label{lin:sc-headA1}
    \EndIf
    \textcolor{red}{\If{$a = \bot$}
        ${\RL}(announce[j])$
        \label{lin:rl-announce}
    \EndIf}
        \State continue
        \Comment{Go to the beginning of the loop}
    \EndIf
    \EndWhile
    \label{lin:end-loop}
    \State $response \gets {\Ld}(announce[i])$
    \label{lin:read-rsp}

    \algrenewcommand{\alglinenumber}[1]{\footnotesize #1L:}
    \State $\tup{q, r} \gets {\LL}(head)$
    \label{lin:ll-head2}
    \hfill \makebox[.6\textwidth][l]{\textcolor{blue}{|| \quad \labeltext[\ref*{lin:ll-head2}R.1]{{\footnotesize \ref*{lin:ll-head2}R.1:}}{lin:wait-til3} \textbf{wait until} ${\Ld}(head) \neq \tup{\wildcard, \tup{\wildcard, i}}$}}
    \item[] \hfill \makebox[.6\textwidth][l]{\textcolor{blue}{\phantom{||} \quad {\footnotesize \ref*{lin:ll-head2}R.2:}  goto Line~\ref{lin:rl-head}}}

    \algrenewcommand{\alglinenumber}[1]{\footnotesize #1:}
    \If{$r = \tup{\wildcard, i}$}
        ${\SC}(head, \tup{q, \bot})$
        \Comment{Clear response from head before returning}
        \label{lin:sc-headA2}
        \textcolor{red}{\Else{ {${\RL}(head)$}}}
        \label{lin:rl-head}
    \EndIf
    \State ${\St}(announce[i], \bot)$
    \Comment{Clear response from $announce[i]$}
    \label{lin:clear-my-announce}
    \State \algorithmicreturn{} $response$
\end{algorithmic}
\end{algorithm}

\ha{A process invoking a read-only operation calls $\textsc{ApplyReadOnly}$, 
which simply reads the object's state from $head$ and returns 
a response according to the sequential specification of the object. 
This does not change the memory representation of the implementation.

A process $p_i$ invoking a state-changing operation calls $\textsc{Apply}$. 
First, $p_i$ announces the operation by writing its description to $announce[i]$,
and then, $p_i$ repeatedly tries to apply operations 
(either its own operation or operations announced by other processes)
until it identifies that its own operation has been applied.
The choice of which operation to apply 
(Lines~\ref{lin:read-help-announce}-\ref{lin:read-other-Announce}) 
is determined by a local variable $priority_i$,
which is not part of the memory representation.
If there is a pending operation by process $p_j$,
$j = priority_i$, then $p_i$ applies $p_j$'s operation;
otherwise, it applies it own operation. 
Each time $p_i$ successfully changes the state of the object,
it increments $priority_i$ (modulo $n$).
This ensures that all pending operations will eventually help the same process.

Applying an operation $o$, with $\Delta(q,o) = (q', r)$,
consists of three stages, each of which can be performed by any process 
(not just the process that invoked $o$, 
and not necessarily the same process for all three stages).

\emph{In the first stage,}
$head$ is changed from $\tup{q,\bot}$ to $\tup{q', \tup{r, j}}$,
where $p_j$ is the process that invoked operation $o$.
The stage starts when some process $p_i$ reads $\tup{q,\bot}$ 
from $head$ with $\LL(head)$ (Line~\refl{lin:ll-head1}),
and decides which operation to try to apply, 
say $o$ by process $p_j$. 
To do so, $p_i$ performs ${\SC}(head, \tup{q',\tup{r,j}})$ (Line~\ref{lin:sc-headB}).
If the {\SC} is successful, the value of $head$ did not change between 
the ${\LL}$ and the $\SC$ of $p_i$.
This ensures that the chosen operation, read from $announce[j]$
in Line~\ref{lin:read-help-announce} or Line~\ref{lin:read-my-Announce}
after the ${\LL}(head)$, is not applied more than once.

\emph{In the second stage,} the response $r$ is written into $announce[j]$, 
overwriting $o$ itself, to notify the invoking process $p_j$ that 
its operation was performed, and what value $p_j$ should return.
A process $p_i$ that writes the response to $announce[j]$
%or clears the response from $head$, 
has read $\tup{q',\tup{r,j}}$ from $head$ with $\LL(head)$
(Line~\refl{lin:ll-head1}) 
and then performs a successful ${\VL}(head)$ in Line~\ref{lin:vl-head}. 
If $p_i$ performs a successful ${\SC}(announce[j],r)$
in Line~\ref{lin:replace-inv-rsp},
then it previously performed a ${\LL}(announce[j])$ 
in Line~\refl{lin:ll-announce}, i.e., 
between ${\LL}(head)$ and ${\VL}(head)$.
This guarantees that the value of $head$ does not change between
${\LL}(head)$ and ${\SC}(announce[j],r)$.

\emph{The third and final stage} 
changes $head$ from $\tup{q',\tup{r,j}}$ to $\tup{q',\bot}$. 
This erases the response $r$ and the process index $j$,
ensuring that forbidden information about the history is not revealed.
The invoking process $p_j$ does not return until its response is cleared
from $head$ (Lines~\refl{lin:ll-head2} and~\ref{lin:sc-headA2}).
%to ensure the implementation is HI.
This ensures that a successful ${\SC}(announce[j],r)$ 
(Line~\ref{lin:replace-inv-rsp})
writes the right response to the applied operation, 
since it can only occur before the response value is cleared from $head$.
If the operation performs a successful ${\SC}(head,q')$ (Line~\ref{lin:sc-headA1}),
then the previous ${\LL}(head)$ (Line~\refl{lin:ll-head1})
guarantees that the replaced value of $head$ was indeed of the form $\tup{q',r}$, 
where $r\neq \bot$.

Finally, before returning, $p_j$ also clears $announce[j]$.}

\subsubsection*{Achieving history independence.}
Algorithm~\ref{alg:universal-construct-R-LLSC}, 
without the lines shown in red, \emph{is not state-quiescent HI}, 
and in fact it is not even quiescent HI:
although we delete past responses from the $head$ and clear $announce[i]$
before returning, their $context$ fields may reveal 
information about the history even when no operation is pending.
For example, suppose process $p_i$ invokes an operation $o$ and writes it to $announce[i]$, and begins the main loop where it tries to perform operations.
Before $p_i$ can even perform $\LL(head)$ in Line~\refl{lin:ll-head1}, faster processes carry out operation $o$ and all other pending operations, and return.
When $p_i$ does reach Line~\refl{lin:ll-head1} and calls $\LL(head)$, it sees that the system is in-between operations ($head = \tup{q,\bot}$), and it finds no other processes requiring help. 
It thus returns straightaway, leaving its link in the $context$ field of $head$.
This might seem innocuous, but it could, for example, reveal that a counter supporting fetch-and-increment and fetch-and-decrement operations,
whose value is currently zero, was non-zero in the past, because the observer can see
that \emph{some} state-changing operation was performed on it.

%since the $context$ part of the LLSC objects' states might be nonempty in a quiescent configuration.
%This happens, for example, due to an ${\LL}(announce[j])$ 
%in Line~\refl{lin:ll-announce} 
%when process $p_j$ has no pending operation, 
%or an ${\LL}(head)$ 
%in Line~\refl{lin:ll-head1} or Line~\refl{lin:ll-head2} 
%when all pending operations are all applied to the object state. 
%In both cases, there might not be a successful $\SC$ that removes the process from $context$ 
%before the system reaches a state-quiescent configuration. 
%This leaks information about the scheduling of processes in the execution leading to this configuration.
%However, in a state-quiescent configuration, 
%the $val$ part of the LLSC objects' is a canonical representation, 
%since $head$'s value is $\tup{q,\bot}$, where $q$ is the state of the object reached by 
%the sequence of operations applied during the execution, 
%and $announce[i]$ is $\bot$, for every $i$. %, $1\leq i \leq n$. 

To address this problem, we add a release ($\RL$) operation to the LLSC object.
$\RL$ removes a process from the context, and
we use it to ensure that the $context$ component
of each LLSC object in the implementation is empty in a state-quiescent configuration.
%In an {\RL} operation by process $p_i$ removes $p_i$ from the object's $context$.
Formally, a \emph{releasable LLSC} (R-LLSC) adds the following operation, 
performed by process $p_i$:
\begin{description}
    \item[${\RL}(\mathcal{O})$:] removes $p_i$ from $\mathcal{O}.context$ and returns \textit{true}.
\end{description}
${\RL}$ operations % , in which a process removes itself from the object's $context$, 
are added in Lines~\ref{lin:rl-announce} and~\ref{lin:rl-head} 
of Algorithm~\ref{alg:universal-construct-R-LLSC}, both marked in red.
We show below \ha{a lock-free implementation of} an R-LLSC object from atomic CAS.
\ha{The} implementation is not wait-free, 
as ${\RL}$ operations may interfere with other ongoing operations (including $\LL$). 
%Nevertheless, it suffices to obtain a wait-free universal HI implementation.
\ha{To handle R-LLSC operations that may block 
and obtain a wait-free universal HI implementation,}
we add the code marked in blue, 
in Lines~\ref{lin:ll-head1},~\ref{lin:ll-announce} and~\ref{lin:ll-head2}.
\ha{These lines interleave} steps in which process $p_i$ checks whether 
some other process $p_j$ has already accomplished what $p_i$ was trying to do 
(e.g., $p_j$ applied $p_i$'s operation for it).
The notation $\parallel$ indicates the interleaving of steps between 
the code appearing to its left and to its right, with some unspecified but finite 
number of steps taken on each side before the process switches 
and starts taking steps of the other side.

In Line~\ref{lin:rl-announce-r}, a {\RL} ensures that if $p_i$'s 
operation is performed by another process while $p_i$ itself is trying 
to help a third process $p_j$, then the $\LL(announce[j])$ in Line~\refl{lin:ll-announce} leaves no trace.
We must do this because we do not know whether the $\LL(announce[j])$ 
on the left side has already ``taken effect'' or not 
at the point where the \textbf{wait until} command on the right-hand side is done.
%in case the waiting stops on the right-hand side of Line~\ref{lin:ll-announce},
%the {\LL} in Line~\refl{lin:ll-announce} leaves no trace.

%They cause the R-LLSC implementation (presented next) to become not wait-free, 
%since some operations, including $\LL$, may block due to other operations. 
% Intuitively, this is because the object state no longer only progresses forward when only {\LL} are applied, it can also progress backward with {\RL} operations. 
%Thus, we assume the {\LL}'s do not provide any progress guarantee and 
%We handle this by adding the code marked with blue, 
%in Lines~\ref{lin:ll-head1},~\ref{lin:ll-announce} and~\ref{lin:ll-head2},
%interleaving a condition that checks whether the process have completed its operation.
%The notation $||$ indicates the interleaving of steps between its right-hand 
%side and its left-hand side. 
%A constant number of steps on each side occur before the 
%process moves to perform steps on the other side, 
%or the process moves on in the algorithm.
%In Line~\ref{lin:rl-announce-r}, a {\RL} is performed to ensure that 
%in case the waiting stops on the right-hand side of Line~\ref{lin:ll-announce},
%the {\LL} in Line~\refl{lin:ll-announce} leaves no trace.

% add a concurrent condition that ensures processes eventually return, marked with blue in the code.

\subsection{Properties of Algorithm~\ref{alg:universal-construct-R-LLSC}}
\label{sec:universal-impl-correct}

\ha{The proof partitions the execution into segments,
with each successful state-change (that is, each successful $\SC(head, \tup{q,\tup{r,i}})$)
beginning a new segment.
We linearize exactly one state-changing operation at the beginning of each such segment, 
and interleave the linearization points of the read-only operations according 
to the segment in which they read the $head$.

Specifically, consider} 
an execution $\alpha$ of Algorithm~\ref{alg:universal-construct-R-LLSC}, 
and fix a linearization of the operations on the R-LLSC objects.
Since we consider the execution $\alpha$ in hindsight, 
this allows us to treat the R-LLSC objects as ``atomic'' and fix a \emph{linearization point}\footnote{
    A linearization point of an operation is a step in the execution, 
    between the operation invocation and response, 
    such that all the linearization points respect the linearization (see~\cite[Chapter 13]{LynchBook}).}
for each operation in the linearization. 
We assume that an operation takes effect exactly at its linearization point.
This allows us to determine the state of the object, according to the linearization points, at any point of the execution. 
For an R-LLSC object $X$, we abuse notation and say $X = v$ to indicate that the value of $X$ is $v$ 
according to the linearization points, and $X.context$ is the $context$ field in $X$'s state according to the linearization points.
A process writes $v$ to an R-LLSC object $X$ if it performs a ${\St}(X,v)$ or a successful ${\SC}(X,v)$, 
and reads value $v$ from R-LLSC object $X$ if it performs a {\Ld} or {\LL} to $X$ that returns $v$. 

Let $op$ be an operation by process $p_i$ in $H(\alpha)$. Denote by $\pi(op) = i$ the invoking process and $I(op)\in O$ the input operation $op$ tries to apply.
We say that an operation $op$ performs some action in $\alpha$ to express that process $\pi(op)$ performs this action during operation $op$. This is well defined since at any point in the execution each process executes at most one pending operation.
Most of the proof considers state-changing operations, and these are simply called operations. 
% Invocation $op$ performs 
% two lines inside the loop between Line~\ref{lin:start-loop} and Line~\ref{lin:end-loop} at the same iteration, if $op$ did not return to Line~\ref{lin:start-loop} between the two lines. We omit to mention that lines happen at the same iteration when it is clear from the context.

We say an operation $op$ is \emph{cleared} if the value $op$ writes to $announce[\pi(op)]$ in Line~\ref{lin:announce-inv} is overwritten. The next invariant states what values can overwrite values written in Line~\ref{lin:announce-inv}. (Missing proofs in this section are deferred to Appendix~\ref{app:universal-impl-correct}.)

\begin{restatable}{invariant}{clearedinv}
\label{lem:cleared-inv}
    If an operation $op$ exits the while loop in Lines~\ref{lin:start-loop}--\ref{lin:end-loop} at some point in the execution, then it is cleared by this point and the value $op$ writes to $announce[\pi(op)]$ in Line~\ref{lin:announce-inv} is overwritten in Line~\ref{lin:replace-inv-rsp} with a value from $R$.
\end{restatable}

\begin{restatable}{invariant}{universalvaluepairs}
\label{lem:universal-value-pairs}
    Let $\tup{q_1,r_1}$ and $\tup{q_2,r_2}$ be two consecutive values written to $head$, then
    \begin{enumerate}
        \item $q_1 = q_2$, $r_1\neq \bot$ and $r_2 = \bot$, or

        \item $r_1 = \bot$ and $r_2 \neq \bot$.
    \end{enumerate}
\end{restatable}

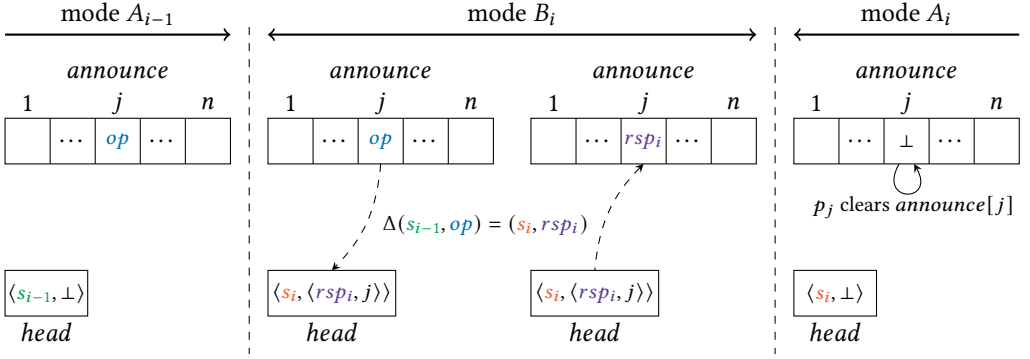
\begin{figure}
    \centering
    \begin{tikzpicture}
        \def \axa {0};
        \def \axb {3.5};
        \def \axc {7};
        \def \axd {10.5};
        \def \hy {-2};

        \draw[thick, ->] (\axa, 1.7) -- ++ (3, 0) node [midway, above] {mode $A_{i-1}$};
        \draw[dashed] (\axb - 0.25 ,1.8) -- ++ (0,-4.3);
        \draw[thick, <->] (\axb, 1.7) -- ++ (6.5, 0) node [midway, above] {mode $B_{i}$};
        \draw[dashed] (\axd - 0.25 ,1.8) -- ++ (0,-4.3);
        \draw[thick, <-] (\axd, 1.7) -- ++ (3, 0) node [midway, above] {mode $A_{i}$};

        % first
        \draw[black] (\axa,0) rectangle ++(3,0.6);
        \node at (\axa + 1.5, 1.2) {$announce$};
        \node at (\axa + 0.3,0.8) {1};
        \node at (\axa + 1.5,0.8) {$j$};
        \node at (\axa + 2.7,0.8) {$n$};
        \draw (\axa + 0.6,0) -- ++ (0,0.6);
        \draw (\axa + 1.2,0) -- ++ (0,0.6);
        \draw (\axa + 1.8,0) -- ++ (0,0.6);
        \draw (\axa + 2.4,0) -- ++ (0,0.6);
        \node at (\axa + 0.9,0.3) {\dots};
        \node[col2] at (\axa + 1.5,0.3) {\footnotesize{$op$}};
        \node at (\axa + 2.1,0.3) {\dots};

        \draw[black] (\axa,\hy) rectangle node{\footnotesize{$\tup{\textcolor{col1}{s_{i-1}}, \bot}$}} ++(1.1,0.6) 
        node [midway, below, text height = 0.5cm] {$head$};

        % second
        \draw[black] (\axb,0) rectangle ++(3,0.6);
        \node at (\axb + 1.5, 1.2) {$announce$};
        \node at (\axb + 0.3,0.8) {1};
        \node at (\axb + 1.5,0.8) {$j$};
        \node at (\axb + 2.7,0.8) {$n$};
        \draw (\axb + 0.6,0) -- ++ (0,0.6);
        \draw (\axb + 1.2,0) -- ++ (0,0.6);
        \draw (\axb + 1.8,0) -- ++ (0,0.6);
        \draw (\axb + 2.4,0) -- ++ (0,0.6);
        \node at (\axb + 0.9,0.3) {\dots};
        \node[col2] at (\axb + 1.5,0.3) {\footnotesize{$op$}};
        \node at (\axb + 2.1,0.3) {\dots};

        \draw[black] (\axb,\hy) rectangle node{\footnotesize{$\tup{\textcolor{col4}{s_{i}}, \tup{\textcolor{col3}{rsp_i},j}}$}} ++(1.7,0.6) 
        node [midway, below, text height = 0.5cm] {$head$};
        \draw[dashed, -stealth ] (\axb + 1.5, 0) to [bend left=20] (\axb + 0.85,\hy + 0.6);
        \node at (\axb + 2.9, \hy + 1.2){\footnotesize{$\Delta(\textcolor{col1}{s_{i-1}}, \textcolor{col2}{op}) = (\textcolor{col4}{s_i}, \textcolor{col3}{rsp_i})$}};
        
        % \node at (\axb + 1.5, \hy - 1) {\footnotesize{\shortstack{some process applies $op$, \\ where $\Delta(s_{i-1}, op) = (s_i, rsp_i)$}}};

        % third
        \draw[black] (\axc,0) rectangle ++ (3,0.6);
        \node at (\axc + 1.5, 1.2) {$announce$};
        \node at (\axc + 0.3,0.8) {1};
        \node at (\axc + 1.5,0.8) {$j$};
        \node at (\axc + 2.7,0.8) {$n$};
        \draw (\axc + 0.6,0) -- ++ (0,0.6);
        \draw (\axc + 1.2,0) -- ++ (0,0.6);
        \draw (\axc + 1.8,0) -- ++ (0,0.6);
        \draw (\axc + 2.4,0) -- ++ (0,0.6);
        \node at (\axc + 0.9,0.3) {\dots};
        \node[col3] at (\axc + 1.5,0.3) {\footnotesize{$rsp_i$}};
        \node at (\axc + 2.1,0.3) {\dots};

        \draw[black] (\axc,\hy) rectangle node{\footnotesize{$\tup{\textcolor{col4}{s_{i}}, \tup{\textcolor{col3}{rsp_i},j}}$}} ++(1.7,0.6) 
        node [midway, below, text height = 0.5cm] {$head$};
        \draw[dashed, stealth -] (\axc + 1.5, 0) to [bend right=20] (\axc + 0.85,\hy + 0.6);
        % \node at (\axc + 1.5, \hy - 1) {\footnotesize{\shortstack{some process writes \\ $rsp_i$ to $announce[j]$}}};

        % fourth
        \draw[black] (\axd,0) rectangle ++(3,0.6);
        \node at (\axd + 1.5, 1.2) {$announce$};
        \node at (\axd + 0.3,0.8) {1};
        \node at (\axd + 1.5,0.8) {$j$};
        \node at (\axd + 2.7,0.8) {$n$};
        \draw (\axd + 0.6,0) -- ++ (0,0.6);
        \draw (\axd + 1.2,0) -- ++ (0,0.6);
        \draw (\axd + 1.8,0) -- ++ (0,0.6);
        \draw (\axd + 2.4,0) -- ++ (0,0.6);
        \node at (\axd + 0.9,0.3) {\dots};
        \node at (\axd + 1.5,0.3) {\footnotesize{$\bot$}};
        \node at (\axd + 2.1,0.3) {\dots};
        \draw[-stealth ] (\axd + 1.4, 0) to[out=-120,in=-60,distance=6mm,swap] (\axd + 1.6, 0) node[below, text height = 0.5cm] {\footnotesize{$p_j$ clears $announce[j]$}};

        \draw[black] (\axd,\hy) rectangle node{\footnotesize{$\tup{\textcolor{col4}{s_{i}}, \bot}$}} ++(1.1,0.6) 
        node [midway, below, text height = 0.5cm] {$head$};
    \end{tikzpicture}
    \caption{Illustrating the transition from mode $A_{i-1}$ to mode $B_i$ and back to mode $A_i$ in Algorithm~\ref{alg:universal-construct-R-LLSC}.}
    \label{fig:mode-transition}
\end{figure}

Invariant~\ref{lem:universal-value-pairs} shows that the algorithm alternates 
between \emph{$A$ modes} and \emph{$B$ modes}.
Specifically, at the beginning of the execution, 
the algorithm is in mode $A_0$. 
The mode changes after each write to $head$ during the execution in the following manner:
If the mode is $A_i$, after the next write to $head$, the algorithm transitions to mode $B_{i+1}$. 
If the mode is $B_i$, after the next write to $head$, the algorithm transitions to mode $A_i$.
We say the algorithm is in mode $A$ if it is in mode $A_i$ for $i\geq 0$, 
and in mode $B$ if it is in mode $B_i$ for $i \geq 1$.
By Invariant~\ref{lem:universal-value-pairs},
since the initial value of $head$ is $\tup{q_0, \bot}$,
if the algorithm is in mode $A_{i-1}$, $i\geq 1$, the value of $head$ is $\tup{q_{i-1},\bot}$ and 
if the algorithm is in mode $B_i$, the value of $head$ is $\tup{q_i,\tup{rsp_i,\wildcard}}$. 
Define $state(0) = q_0$.
For $i\geq 1$, define $state(i) = q_i$, where $head$'s value is equal to $\tup{q_i, \wildcard}$ in modes $B_i$ and $A_i$
and define $response(i) = rsp_i$, where $head$'s value is equal to $\tup{\wildcard, \tup{rsp_i,\wildcard}}$ in mode $B_i$.
(See Figure~\ref{fig:mode-transition}.)

For any $i \geq 1$, consider the successful ${\SC}(head, \tup{s_i,r})$ in  
Line~\ref{lin:sc-headB} by operation $op$ that transitions the algorithm from mode $A_{i-1}$ to mode $B_i$. 
%Since $r\neq \bot$, this happens in Line~\ref{lin:sc-headB} and let $r = \tup{response,\wildcard}$. 
If the condition in Line~\ref{lin:cond-help} holds, let $op^*$ be the operation that writes the value $op$ reads in Line~\ref{lin:read-help-announce}; otherwise, $op^* = op$. We say this transition \emph{applies} $op^*$. 
% For the transition $op^*$ is linearized by,
% the \emph{linearization point} $lin(op^*)$ is defined to be the successful {\SC} operation that transitions the algorithm. Since the $head$ is an atomic R-LLSC object, this operation invokes and responds in a single step in $\alpha$, and $lin(op^*)$ is equal to this step. Note that if an operation was not linearized, then the linearization point $lin$ is undefined.

An ${\LL}(head)$ or $\Ld(head)$ returns in mode $A_i$ or $B_{i}$, $i > 0$, if the last successful ${\SC}(head,\wildcard)$ that precedes the operation transitions the algorithm to this mode. An ${\LL}(head)$ or $\Ld(head)$ returns in mode $A_0$ if no ${\SC}(head,\wildcard)$ precedes the operation.

\begin{restatable}{lemma}{appliedlinearized}
\label{lem:applied=linearized}
\label{lem:lin-inductive-state-rsp}
    Consider a transition from mode $A_{i-1}$ to mode $B_i$, $i\geq 1$, and let $op^*$ be the operation applied by this transition, then
    $op^*$ is applied for the first time and
    $\Delta(state(i-1), I(op^*)) = (state(i), response(i))$.
\end{restatable}

Lemma~\ref{lem:applied=linearized} shows that any 
operation is applied by at most one transition from mode $A$ to mode $B$. 
We say that a transition from mode $A_{i-1}$ to mode $B_i$, $i\geq 1$, \emph{linearizes} the operation it applies.

\begin{restatable}{lemma}{returnhelper}
\label{lem:return-helper}
Let $op^*$ be the operation linearized by the transition from mode $A_{i-1}$ to mode $B_i$, $i\geq 1$, then:

    \begin{enumerate}          
         \item Only operation $op^*$ is cleared in mode $B_i$ and no operation is cleared in mode $A_i$.
        
        \item If operation $op^*$ returns, then it returns in mode $A_{j}$ or mode $B_{j+1}$ only for $j\geq i$.
        \end{enumerate}
\end{restatable}

Define the linearization function $h_{uc}$ according to the transitions from mode $A_{i-1}$ to mode $B_i$, $i\geq 1$, which we call the $i$-th mode transition; first, order linearized operations in $H(\alpha)$ such that $op_1$ precedes $op_2$ in $h_{uc}(\alpha)$ if and only if 
$op_1$ is linearized by the $i$-th mode transition and $op_2$ is linearized by the $j$-th mode transition, such that $i< j$.
A read-only operation ${op}_r$ \emph{reads from} the $i$-th mode transition if the ${\Ld}(head)$ in Line~\ref{lin:load-head} returns in mode $B_i$ or $A_i$.
Consider the read-only operations in $H(\alpha)$ in the order they are invoked in $\alpha$. 
A read-only operation $op_r$ in $H(\alpha)$, which reads from the $i$-th mode transition,
is placed after the operation ${op}_w$ linearized by the $i$-th transition, 
and after all previous read-only operations that also read from the $i$-th mode transition.

% let $h_{uc}(\alpha)$ be the sequential history that consists of the linearized invocations $op$ in $H(\alpha)$, and order the invocations such that $op_1$ precedes $op_2$ in $h_{uc}(\alpha)$ if and only if 
% $op_1$ is linerized by the $i$-th transition and $op_2$ is linerized by the $j$-th such that $i< j$.
% Next, 
 
\begin{lemma}
\label{lem:uc-seq-spec}
    $h_{uc}(\alpha)$ is a linearization of $\alpha$
    and $\seqstate{h_{uc}(\alpha)} = q$ such that $head = \tup{q,\wildcard}$.
\end{lemma}

\begin{proof}
    First, we need to show that $h_{uc}(\alpha)$ includes all completed operations.
    By construction, all completed read-only operations are included in $h_{uc}(\alpha)$ and if a read-only operation reads from the $i$-th mode transition, it is pending in mode $A_i$ or mode $B_i$.

    By Invariant~\ref{lem:cleared-inv}, if an operation returns it must be cleared, and by Lemma~\ref{lem:return-helper}, only linearized operations can be cleared. This shows that only linearized operations return and these operations are included in $h_{uc}(\alpha)$.
    By Lemma~\ref{lem:return-helper}, an operation linearized by the transition from mode $A_{i-1}$ to mode $B_i$, $i\geq 1$, can only return in a mode that follows $B_i$, thus, an operation is linearized while it is pending.
    It is left to show that $h_{uc}(\alpha)$ is in the sequential specification of the abstract object. Since read-only operations return the correct response according to the sequential specification of the abstract object and the state returned from ${\Ld}(head)$, it is enough to show this only for the state-changing operations.
    
    % For any $i\geq 1$, following Lemma~\ref{lem:lin-inductive-state-rsp} and Lemma~\ref{lem:only-lin-return}, the $i$-th operation $op$ in $h_{uc}(\alpha)$ returns response value $response(i)$, where $\mathrm{apply}(state(i-1), op(Inv)) = (state(i), response(i))$. 
    The proof is by induction on the transitions from mode $A_{i-1}$ to mode $B_i$, $i\geq 1$.
    Note that between any two consecutive transitions from mode $A$ to $B$, the first component (i.e., the
    object's state) in the value of $head$ does not change.

    {\bf Base case:} $i=0$, $h_{uc}(\alpha)$ includes only read-only operations that read the value of $head$ in mode $A_0$. Trivially, $h_{uc}(\alpha)$ is in the sequential specification. In addition, $\seqstate{h_{uc}(\alpha)} = q_0$ 
    and in the algorithm initialization we have that $head = \tup{q_0,\bot}$.

    {\bf Induction step:}
    Consider the transition from mode $A_{i-1}$ to mode $B_i$, $i\geq 1$, and assume the induction holds for any $j < i$. 
    Let $\alpha'$ be the prefix of $\alpha$ that contains the first $i-1$ mode transitions from $A$ to $B$.
    Let $op$ be the operation linearized by the transition from mode $A_{i-1}$ to mode $B_i$. 
    The induction hypothesis implies that $h_{uc}(\alpha')$ is in the sequential specification and $\seqstate{h_{uc}(\alpha')} = state(i-1)$.
    by Lemma~\ref{lem:return-helper}, the value $op$ writes in Line~\ref{lin:announce-inv} is replaced with $response(i)$. Thus, $op$ must read this value in Line~\ref{lin:read-rsp} and return it.
    By Lemma~\ref{lem:lin-inductive-state-rsp}, $\Delta(state(i-1), I(op)) = (state(i), response(i))$. This implies that $h_{uc}(\alpha)$ is also in the sequential specification and $\seqstate{h_{uc}(\alpha)} = state(i)$.
\end{proof}

\paragraph{State-quiescent HI}
The next two lemmas show that for every linearization of the operations on the R-LLSC objects, 
the states of the R-LLSC objects in Algorithm~\ref{alg:universal-construct-R-LLSC} 
provide a canonical representation for every higher-level state in a state-quiescent configuration.

\begin{restatable}{lemma}{quiescentannouncebot}
\label{lem:quiescent-announce-bot}
    If process $p_i$ has no pending state-changing operations, then $announce[i] = \bot$.
\end{restatable}

\begin{lemma}
\label{lem:RLLSC-quiescent-state}
    If an execution $\alpha$ ends in a state-quiescent configuration, 
    then the context of every R-LLSC variable $head, announce[1],\dots, announce[n]$,  
    is equal to $\emptyset$.
\end{lemma}

\begin{proof}
    Since read-only operations do not change the state of base objects, we only consider state-changing operations.
    
    Before returning, an operation $op$ performs Line~\ref{lin:ll-head2}.
    If $op$ finishes performing Line~\ref{lin:ll-head2} on the left-hand side, then the ${\LL}(head)$ in Line~\refl{lin:ll-head2} returns.
    If the condition in Line~\ref{lin:sc-headA2} holds,
    $op$ performs an ${\SC}(head,\wildcard)$. Otherwise, $op$ performs an ${\RL}(head)$ in Line~\ref{lin:rl-head}.
    If $op$ finishes performing
    Line~\ref{lin:ll-head2} on the right-hand side, $op$ jumps to Line~\ref{lin:rl-head} and performs an ${\RL}(head)$.
    Thus, if there is no pending state-changing operation by process $p_i$, then $p_i \notin head.context$.

    We show that for any ${\LL}(announce[j])$, $1\leq j \leq n$, by process $p_i$ there is a later $\RL$ or $\SC$ to $announce[j]$ by $p_i$ or a $\St$ to $announce[j]$ by a different process that resets $announce[j].context$ and removes $p_i$ from it.
    This implies that if there is no pending state-changing operation by process $p_i$, then $p_i \notin announce[j].context$.
    Consider operation $op$ that performs Line~\ref{lin:ll-announce}. 
    If $op$ finishes performing Line~\ref{lin:ll-announce} on the right-hand side, $op$ performs an ${\RL}(announce[j])$ in Line~\ref{lin:rl-announce-r}.
    If $op$ finishes performing Line~\ref{lin:ll-announce} on the left-hand side, then the ${\LL}(announce[j])$ in Line~\refl{lin:ll-announce} returns value $a$.
    If the condition in Line~\ref{lin:vl-head} doesn't hold and $a = \bot$, then $op$ performs an ${\RL}(announce[j])$ in Line~\ref{lin:rl-announce}. Otherwise, if $a\neq \bot$, by Lemma~\ref{lem:quiescent-announce-bot} and since the last configuration is state-quiescent, there is a ${\St}(announce[j],\bot)$ before the execution reaches a state-quiescent configuration.
    
    If the condition in Line~\ref{lin:vl-head} holds, the algorithm is in mode $B_i$ for some $i\geq 1$ when $op$ performs the ${\VL}(head)$ in Line~\ref{lin:vl-head}. If $a\in O$, then $op$ performs an ${\SC}(announce[j],\wildcard)$ in Line~\ref{lin:replace-inv-rsp}. Let $op^*$ be the operation linearized by the transition from mode $A_{i-1}$ to mode $B_i$, necessarily $\pi(op^*) = j$. Assume $a = \bot$, by Lemma~\ref{lem:return-helper} $op^*$ is cleared in state $B_i$, therefore, $op^*$ must write $\bot$ to $announce[i]$ in state $B_i$. However, this contradicts that $op^*$ cannot return in state $B_i$. Thus, $a\neq \bot$, and by Lemma~\ref{lem:quiescent-announce-bot} and since the last configuration is state-quiescent, there is a ${\St}(announce[j],\bot)$ before the execution reaches a state-quiescent configuration.
\end{proof}

\subsection{Lock-Free Perfect-HI R-LLSC Object from Atomic {\CAS}}
\label{sec:R-LLSC}

\begin{algorithm}[tb]\small
\caption{\small Lock-free perfect HI R-LLSC object from {\CAS} : code for process $p_i$}
\label{alg:hi-llsc-from-cas}
\raggedright 
     \quad\,\, $X$: {\CAS} variable initialized to $\parens*{v_0,0,\dots,0}$ \\
     \quad
%     \quad\,\,\, \hrulefill
\setlength{\multicolsep}{0.0pt}
\begin{multicols}{2}
\begin{algorithmic}[1]
    \Statex ${\LL}(\mathcal{O})$:
        \State $cur \gets \textsc{Read}(X)$
        \label{lin:ll-read}
        \State $new \gets cur$; $new.context[i] \gets 1$
        \While{!${\CAS}\parens{X, cur, new}$}
        \label{lin:ll-cas}
            \State $cur \gets \textsc{Read}(X)$
            \State $new \gets cur$; $new.context[i] \gets 1$
        \EndWhile
        \State \algorithmicreturn{} $cur.val$
    \Statex
    \Statex ${\SC}(\mathcal{O}, v)$:
    \State $cur \gets \textsc{Read}(X)$
    \label{lin:sc-read1}
    \While{$cur.context[i] = 1$}
    \label{lin:sc-valid-cond}
        \If{${\CAS}(X, cur, \parens{v,0,\dots,0})$}
        \label{lin:sc-cas}
            %\State 
            \algorithmicreturn{} \textit{true}
        \EndIf
        \State $cur \gets \textsc{Read}(X)$
        \label{lin:sc-read2}
    \EndWhile
    \State \algorithmicreturn{} \textit{false}
    \Statex
    \Statex ${\VL}(\mathcal{O})$:
    \State $cur \gets \textsc{Read}(X)$
    \label{lin:vl-read}
    \State \algorithmicreturn{} $cur.context[i]$
    \columnbreak
    \Statex ${\RL}(\mathcal{O})$:
    \State $cur \gets \textsc{Read}(X)$
    \label{lin:rl-read1}
    \State $new \gets cur$; $new.context[i] \gets 0$
    \While{$cur.context[i] = 1$}
    \label{lin:rl-valid-cond}
        \If{${\CAS}(X, cur, new)$} 
        % \State 
        \algorithmicreturn{} \textit{true}
        \label{lin:rl-cas}
        \EndIf
        \State $cur \gets \textsc{Read}(X)$
        \label{lin:rl-read2}
        \State $new \gets cur$; $new.context[i] \gets 0$
    \EndWhile
    \State \algorithmicreturn{} \textit{true}
    \Statex
    \Statex ${\Ld}(\mathcal{O})$:
    \State $cur \gets \textsc{Read}(X)$
    \label{lin:rllscf-read}
    \State \algorithmicreturn{} $cur.val$
    \Statex
    \Statex ${\St}(\mathcal{O}, v)$:
    \State $\textsc{Write}(X,\parens{v,0,\dots,0})$
    \label{lin:rllscf-write}
    \State \algorithmicreturn{} \textit{true}
\end{algorithmic}
\end{multicols}
\end{algorithm}

The implementation of an R-LLSC object using a single atomic
{\CAS} object is based on~\cite{IsraeliRappoportPODC94}, 
and its code appears in Algorithm~\ref{alg:hi-llsc-from-cas}.
The state of the R-LLSC object $\mathcal{O}$ is stored in the {\CAS} object 
in the format $(v, c_1,\ldots,c_n) \in V\times \braces{0,1}^n$,
where $v = \mathcal{O}.val$ is its value, 
and each bit $c_i$ indicates whether or not $p_i \in \mathcal{O}.context$.
Denote $x.val = v$ and $x.context[i] = c_i$.
The implementation is perfect HI, 
because the mapping from abstract state to memory representation is unique, 
and no additional information is stored.

The operations {\Ld} and {\VL} are read-only; to implement them, 
we simply read $X$ and return the appropriate response.
A {\St} operation writes into the {\CAS} a new value with an empty context, 
regardless of the current state of the objects.
Finally, the {\LL}, {\RL} and {\SC} operations are implemented by reading $X$ 
and then trying to update it using a $\CAS$ operation, 
but this is not guaranteed to succeed; 
hence, these operations are only lock-free, not wait-free.

Fix an execution $\alpha$ of Algorithm~\ref{alg:hi-llsc-from-cas}.
We define {linearization points} for operations in $H(\alpha)$ as steps in the execution $\alpha$,
as follows. 
Note that operations on the CAS object $X$ can be linearization points,
since the CAS object is atomic and the invocation and response occur in the same step. 
%Define linearization points for operations in $H(\alpha)$ as follows:
\begin{itemize}
    \item Let $ll$ be an ${\LL}(\mathcal{O})$ operation in $H(\alpha)$. If operation $ll$ performs a successful {\CAS} in Line~\ref{lin:ll-cas}, then $lin(ll)$ is defined to be this successful {\CAS}.
    Otherwise, $lin(ll)$ is undefined.

    \item Let $vl$ be a ${\VL}(\mathcal{O})$ operation in $H(\alpha)$, $lin(vl)$ is defined to be the $\textsc{Read}(X)$ in Line~\ref{lin:vl-read} in $vl$.
    Otherwise, $lin(vl)$ is undefined.

     \item Let $rl$ be a ${\RL}(\mathcal{O})$ operation in $H(\alpha)$. If operation $rl$ performs a successful {\CAS} in Line~\ref{lin:rl-cas}, then $lin(rl)$ is defined to be this successful {\CAS}. If $rl$ performs a $\textsc{Read}(X)$ in Line~\ref{lin:rl-read1} or Line~\ref{lin:rl-read2} that returns value $v$ such that $v.context[\pi(rl)] = 0$, then $lin(rl)$ is defined to be this \textsc{Read}. 
     Otherwise, $lin(rl)$ is undefined.

     \item Let $sc$ be a ${\SC}(\mathcal{O},\wildcard)$ operation in $H(\alpha)$. If operation $sc$ performs a successful {\CAS} in Line~\ref{lin:sc-cas}, then $lin(sc)$ is defined to be this successful {\CAS}. 
     If $sc$ performs a $\textsc{Read}(X)$ in Line~\ref{lin:sc-read1} or Line~\ref{lin:sc-read2} that returns value $v$ such that $v.context[\pi(sc)] = 0$, then $lin(sc)$ is defined to be this \textsc{Read}. 
     Otherwise, $lin(sc)$ is undefined.

      \item Let $l$ be a ${\Ld}(\mathcal{O})$ operation in $H(\alpha)$, $lin(l)$ is defined to be the $\textsc{Read}(X)$ in Line~\ref{lin:rllscf-read} in $l$. 
      Otherwise, $lin(l)$ is undefined.
      
       \item Let $s$ be a ${\St}(\mathcal{O}, \wildcard)$ operation in $H(\alpha)$, $lin(s)$ is defined to be the ${\St}(X, \wildcard)$ in Line~\ref{lin:rllscf-write} in $s$.
       Otherwise, $lin(s)$ is undefined.
\end{itemize}

Define the linearization function $h_{llsc}$ according to these linearization points as follows; define $h_{llsc}(\alpha)$ to be the sequential history that consists of the operations $op$ in $H(\alpha)$ such that $lin(op)$ is defined, and order the operations such that $op_1$ precedes $op_2$ in $h_{llsc}(\alpha)$ if and only if $lin(op_1)$ precedes $lin(op_2)$ in $\alpha$.

Clearly, the {\VL}, {\Ld} and {\St} are wait-free and {\LL}, {\SC} and {\RL} are lock-free. 
We have the next theorem, proved in Appendix~\ref{app:R-LLSC}:

\begin{theorem}
\label{thm:RLLSCF-from-CAS}
    Algorithm~\ref{alg:hi-llsc-from-cas} is a lock-free linearizable perfect HI 
    implementation of a R-LLSC object from atomic {\CAS}.
\end{theorem}

We say that {\LL}, {\RL}, {\SC} and {\St} are \emph{context changing}, 
while {\SC} and {\St} are \emph{context resetting}.
We have the following stronger progress property for the lock-free operations 
which follows from the code.

\begin{lemma}
\label{lem:rllscf-lock-free}
    %The {\LL}, {\SC} and {\RL} operations in Algorithm~\ref{alg:hi-llsc-from-cas} are lock-free. Moreover, 
    % If a {\LL}, {\SC}, or {\RL} operation takes infinitely many steps 
    % without returning, then infinitely many context-changing operations return.
    \ns{
    Let $op$ be a {\LL}, {\SC}, or {\RL} operation that is pending in execution $\alpha$. In any extension of $\alpha$ where the process that invoked $op$ takes infinitely many steps without $op$ returning, infinitely many context-changing operations return.}
\end{lemma}

% In Appendix~\ref{app:rllsc-from-cas} we prove that Algorithm~\ref{alg:hi-llsc-from-cas} 
% is linearizable and satisfies perfect HI.
% However, 

%may prevent each other from returning.
%A waiting {\RL} or {\SC} operation may return, 
%without changing the state of $\mathcal{O}$, 
%if a concurrent $\SC$ or $\St$ operation removes the process from $\mathcal{O}.context$,
%as stated in the next lemma (proved in Appendix~\ref{app:rllsc-from-cas}).

When progress is concerned,
we cannot rely on the progress of Algorithm~\ref{alg:hi-llsc-from-cas} as a black box, 
because the $\LL$, $\RL$ and $\SC$ operations are not by themselves wait-free.
Still, we can rely on the \emph{interactions} among the R-LLSC operations 
to ensure that the way they are used in Algorithm~\ref{alg:universal-construct-R-LLSC} is wait-free.
The $\SC$ and $\St$ operations ``help'' $\RL$ and $\SC$ operations,
in the sense that a successful $\SC$ or $\St$ operation clears the $context$,
causing all pending $\RL$ and $\SC$ operations to complete: 
$\RL$ operations return because the process is indeed no longer in the $context$, 
and $\SC$ operations return because they have failed. 
This is formalized in the next lemma, proved in Appendix~\ref{app:R-LLSC}:

\begin{restatable}{lemma}{scrlfinish}
\label{lem:sc-rl-finish}
    Let $op$ be an {\RL} or {\SC} operation
    that is pending in execution $\alpha$,
    and suppose that in $\alpha$,
    a context-resetting
    operation is invoked after $op$,
    and returns \textit{true} before
    $op$ returns.
    Then in any extension of $\alpha$,
    $op$ returns within a finite number of steps by the process that invoked it. 
%    If an $\SC$ or $\St$ operation starts after $op$ starts and returns \textit{true} before $op$ returns, then $op$ returns after a finite number of steps by $op$.
\end{restatable}

\subsection{Wait-Free State-Quiescent HI Universal Implementation from Atomic CAS}
\label{sec:universal-impl-wait-free}

We now combine Algorithm~\ref{alg:universal-construct-R-LLSC} with an R-LLSC implementation
described in Section~\ref{sec:R-LLSC},
to get a wait-free state-quiescent HI universal implementation, despite that the R-LLSC implementation is only lock-free.

Since $\LL$, $\RL$ and $\SC$ are lock-free, 
if a process tries to modify $head$,
some process will eventually succeed in modifying $head$.
By Lemma~\ref{lem:sc-rl-finish}, 
this allows other pending $\SC$ and $\RL$ operations to complete.
This property does not hold for a {\LL} operation, which may never return, 
but this is handled by the invoking algorithm, which interleaves steps where 
it checks if the operation was performed by a different process, as explained above.
The priority-based helping mechanism ensures that every pending 
operation is eventually applied.
Thus, the wait conditions in Lines~\refl{lin:ll-head1},~\refl{lin:ll-announce} and~\refl{lin:ll-head2} 
eventually become false, releasing operations that might be stuck in an $\LL$ operation.
%, that only provides lock-freedom, and therefore, might not return even if the processes takes infinitely many steps.

Consider an {infinite} execution $\alpha$ of Algorithm~\ref{alg:universal-construct-R-LLSC} with the implemented R-LLSC objects, and fix the linearization of the operations on the R-LLSC objects described in Section~\ref{sec:R-LLSC}.
The next lemma shows that if processes invoke infinitely many state-changing operations, then infinitely many operations are linearized.

\begin{lemma}
\label{lem:uc-global-progress}
    For any $i\geq 0$, the algorithm transitions from mode $A_{i}$ to mode $B_{i+1}$ or from mode $B_{i+1}$ to mode $A_{i+1}$ in a finite number of steps by processes taking steps in state-changing operations.
\end{lemma}

\begin{proof}
    Assume the algorithm is in mode $A_i$ or mode $B_{i+1}$ and never transitions to mode $B_{i+1}$ or mode $A_{i+1}$, respectively, 
    despite processes taking an infinite number of steps in state-changing operations.
    By Lemma~\ref{lem:return-helper}, no new operations are cleared in mode $A_i$ and exactly one pending operation $op^*$ is cleared in mode $B_{i+1}$,
    namely, the operation linearized by the transition to mode $B_{i+1}$.
    A cleared operation eventually reaches Line~\ref{lin:read-rsp}, as the waiting conditions in Line~\ref{lin:start-loop}, Line~\ref{lin:wait-til1} and Line~\ref{lin:wait-til2} do not hold. 
    If the algorithm is in mode $A_i$, $head = \tup{\wildcard, \bot}$, and the waiting condition in Line~\ref{lin:wait-til3} does not hold for any cleared operation. If the algorithm is in state $B_{i+1}$, the waiting condition in Line~\ref{lin:wait-til3} does not hold for any cleared operation except $op^*$.
    Therefore, eventually,
    all processes taking steps in state-changing operation are inside the loop in Lines~\ref{lin:start-loop}--\ref{lin:end-loop}, except maybe one process that is stuck in Line~\ref{lin:ll-head2}. Let $\alpha'$ be this infinite suffix of $\alpha$.
    
    % Consider an operation that performs an ${\LL}(head)$ that returns a state that precedes $A_i$ or $B_{i+1}$ and reaches an ${\SC}(head,\wildcard)$ in state $A_i$ or $B_{i+1}$, respectively, by Lemma~\ref{lem:fail-sc-rl-finish}, this {\SC} returns \textit{false} in a finite number of steps by the process. 
    % Thus, eventually, an ${\LL}(head)$ in Line~\refl{lin:ll-head1} by a process that takes infinitely many steps can only return state $A_i$ or $B_{i+1}$, respectively.
    % Note that an ${\LL}(head)$ in Line~\refl{lin:ll-head2} by $op^*$ can only return state $B_{i+1}$.
    % Let $\alpha'$ be the infinite suffix of $\alpha$, where all processes taking steps take infinitely many steps and, except maybe one process, are inside the loop in Lines~\ref{lin:start-loop} --~\ref{lin:end-loop} in iterations that if the ${\LL}(head)$ in Line~\refl{lin:ll-head1} returns, it returns state $A_i$ or $B_{i+1}$, respectively.
    
    A transition to mode $B_{i+1}$ occurs if an operation performs an ${\LL}(head)$ in Line~\refl{lin:ll-head1} that returns in mode $A_i$ and afterwards a successful ${\SC}(head,\wildcard)$ in Line~\ref{lin:sc-headA1}.
    A transition to mode $A_{i+1}$ occurs if an operation performs an ${\LL}(head)$ in Line~\refl{lin:ll-head1} or Line~\refl{lin:ll-head2} that returns in mode $B_{i+}$ and afterwards a successful ${\SC}(head,\wildcard)$ in Line~\ref{lin:sc-headA1} or Line~\ref{lin:sc-headA2}, respectively. 
    Once an ${\LL}(head)$ returns in mode $A_i$ or $B_i$ in $\alpha'$ and the process performs an ${\SC}(head,\wildcard)$, any response of the {\SC} indicates a state transition. If it succeeds, then this {\SC} changes the value of $head$, and if it  
    fails, the value of $head$ is changed by a different {\SC}, between the {\LL} and the failed {\SC}.
    % From the point where all operations taking steps in the infinite suffix are in the loop in Lines~\ref{lin:start-loop} --~\ref{lin:end-loop},

    Since the algorithm mode is stuck in $A_i$ or $B_i$ in $\alpha'$, 
    any ${\LL}(head)$ that starts in mode $A_i$ or $B_i$, respectively, 
    returns this mode. 
    This implies that only a finite number of ${\LL}(head)$ and ${\SC}(head, \wildcard)$ operations 
    return in $\alpha'$. Since {\Ld} and {\VL} are not context-changing, 
    this contradicts Lemma~\ref{lem:rllscf-lock-free}, 
    which states that an infinite number of context-changing operations on $head$ return in $\alpha'$.
\end{proof}

% Once an operation is cleared, the wait condition parallel to the {\LL}'s in the loop in Lines~\ref{lin:start-loop}--\ref{lin:end-loop} no longer holds. 

%In addition, using Lemma~\ref{lem:uc-global-progress}, we can show that any {\RL} or {\SC} by $op$ returns after a finite number of steps by the operation. These two properties imply the next lemma, proved in Appendix~\ref{app:universal-impl-wait-free}:

% This lemma implies that an operation returns in a finite number of steps after it is cleared, 
% since the wait condition parallel to the {\LL}'s 
% in the loop in Lines~\ref{lin:start-loop}--\ref{lin:end-loop} no longer holds.

Finally, we discuss history independence.
At a state-quiescent configuration, the states of the R-LLSC objects are uniquely defined, 
%for any linearization,
according to the state reached by the sequence of operations applied during the execution.
Note that the state includes both the \textit{val} and \textit{context} part of the object.
Since the R-LLSC implementation is perfect HI,  by a simple composition, 
this state translates to a unique memory representation.
The implementation of the R-LLSC objects provides the strongest form of history independence, and for our need, a weaker state-quiescent HI implementation also suffices.
The next theorem
concludes this section by putting all the pieces together.
%summarizes the above discussion.
%shows that Algorithm~\ref{alg:universal-construct-R-LLSC}, 
%together with Algorithm~\ref{alg:hi-llsc-from-cas}, 
%is a wait-free state-quiescent HI universal implementation. 

\begin{restatable}{theorem}{universalconstructRLLSC}
\label{thm:universal-construct-RLLSC}
Algorithm~\ref{alg:universal-construct-R-LLSC}, 
with the R-LLSC objects implemented by Algorithm~\ref{alg:hi-llsc-from-cas}, 
is a linearizable wait-free state-quiescent HI universal implementation.
\end{restatable}

\begin{proof}
    Since Algorithm~\ref{alg:hi-llsc-from-cas} is linearizable (Theorem~\ref{thm:RLLSCF-from-CAS}), 
    Lemma~\ref{lem:uc-seq-spec} implies that Algorithm~\ref{alg:universal-construct-R-LLSC} 
    is a linearizable universal implementation.

    To show wait-freedom, first note that a read-only operation returns after a finite number of steps, 
    since {\Ld} is wait-free.
    
    Consider a state-changing operation $op$ and assume it never returns, despite taking infinitely many steps. 
    \ns{The next claim, proved in Appendix~\ref{app:universal-impl-wait-free}, implies that $op$ is never cleared:
    \begin{restatable}{claim}{ucclearfinitereturn}
    \label{lem:uc-clear-finite-return}
        Operation $op$ returns in a finite number of steps by the process that invoked $op$ after $op$ is cleared.
    \end{restatable}}
    %By Lemma~\ref{lem:uc-clear-finite-return}, $op$ is never cleared.
    After a successful transition from mode $A_i$ to mode $B_{i+1}$, $i\geq 0$, due to a ${\SC}(head, \wildcard)$ by operation $op'$, the value of $priority_{\pi(op')}$ is increased by 1 modulo $n$.
    For every process $p_j$, $1\leq j \leq n$, after at most $n$ consecutive mode transitions by this process since the write to $announce[\pi(op)]$ in Line~\ref{lin:announce-inv} by $op$, $priority_j = \pi(op)$. 
    Since $op$ is state changing, Lemma~\ref{lem:uc-global-progress} implies that 
    after a finite number of steps by $op$, all pending operations $op'$ have $priority_{\pi(op')} = \pi(op)$ 
    and the algorithm is in state $A_i$, $i\geq 0$.
    Since $op$ is never cleared, when $op'$ performs Line~\ref{lin:read-help-announce} 
    in state $A_i$ it must read the input operation of $op$. 
    Hence, the next ${\SC}(head,\wildcard)$ by $op'$ in Line~\ref{lin:sc-headB} tries 
    to linearize $op$. Since by Lemma~\ref{lem:return-helper}, 
    if $op$ is linearized it must also be cleared after a finite number of steps by $op$, the {\SC} fails.
    However, this implies the algorithm never transitions from state $A_i$ to $B_{i+1}$, 
    despite $op$ taking infinitely many steps, in contradiction to Lemma~\ref{lem:uc-global-progress}.

    To show that Algorithm~\ref{alg:universal-construct-R-LLSC} is state-quiescent HI, 
    consider execution $\alpha$ that ends in a state-quiescent configuration.
    By Lemma~\ref{lem:quiescent-announce-bot}, the value of $announce[i] = \bot$ for all $1\leq i \leq n$.
    By Lemma~\ref{lem:return-helper}, for any $i\geq 1$, the operation linearized by the transition from state $A_{i-1}$ to $B_i$ is still pending in state $B_i$. Hence, $\alpha$ must end in state $A$ and by Lemma~\ref{lem:uc-seq-spec}, $head = \tup{q, \bot}$, where $\seqstate{h_{uc}(\alpha)} = q$.
    By Lemma~\ref{lem:RLLSC-quiescent-state}, any R-LLSC variable $X$ has $X.context = \emptyset$ and only contains information about the $val$ part.
    By Theorem~\ref{thm:RLLSCF-from-CAS}, the R-LLSC implementation is perfect HI and the memory representation associated with each object is in a canonical state at the final configuration. This implies that the complete memory representation is also in a canonical state.
    Since this holds for any execution $\alpha$ that ends in a state-quiescent configuration such that $\seqstate{h_{uc}(\alpha)} = q$, this shows that Algorithm~\ref{alg:universal-construct-R-LLSC}, together with Algorithm~\ref{alg:hi-llsc-from-cas}, is a state-quiescent HI universal implementation.
\end{proof}
\section{Discussion}

This paper introduces the notion of history independence for concurrent data structures,
explores various ways to define it, 
and derives possibility and impossibility results.
We gave two main algorithmic results:
a wait-free multi-valued register, 
%\Rnote{There is more than one algorithmic result about the register. I think you're referring to a specific one here, which is the wait-free quiescent HI implementation.} 
and a universal implementation of arbitrary objects. 
Interestingly, both implementations follow a similar recipe:
starting with a history-independent lock-free implementation, 
helping is introduced to achieve wait-freedom.
However, helping tends to leak information about the history of the object, 
so we introduce mechanisms to clear it. 

Our results open up a range of research avenues, 
exploring history-independent object implementations and other notions of history independence.
\emph{Randomization} is of particular importance,
as it is a tool frequently used to achieve both algorithmic efficiency 
and history independence.
When randomization is introduced, 
the distinction between weak and strong history independence becomes meaningful.
We note that randomization will not help circumvent the impossibility result from Section~\ref{sec:gen-lb},
if we require \emph{strong} history independence:
by a result of~\cite{HartlineHoMo02,HartlineHoMo05},
in any strongly history-independent implementation of a reversible object, 
the canonical memory representation needs to be fixed up-front, 
and our impossibility proof then goes through.
However, this does not rule out \emph{weakly} history-independent implementations.
%and indeed, there is evidence that such implementations are not subject to our impossibility result \TODO{refer to appendix?} \NNote{I didn't include anything about WHI in the appendix}.
We remark that even coming up with a meaningful definition for 
history independence in randomized concurrent implementations is non-trivial, 
because randomization can affect the number of steps an operation takes, 
making it challenging to define a probability distribution over the memory 
states at the points where the observer is allowed to observe the memory.

%\begin{itemize}
%    \item Randomization won't help circumvent the impossibility result, if we stay within the confines of the strong HI sequential definition, which implies by Hartline et al that randomization needs to be fixed up front.
%    \item The wait-free register that is quiescent HI and the universal construction have a shared structure.
%    \item Reduce the number of possible memory states in the universal implementation (non-trivial).
%    \item Future work: any point definition, that is weaker than ``perfect'' HI.
%\end{itemize}

\begin{acks}
Hagit Attiya is partially supported by the Israel Science Foundation (grant number 22/1425).
Rotem Oshman is funded by NSF-BSF Grant No. 2022699.
Michael Bender and Martin Farach-Colton are funded by NSF Grants CCF-2106999, CCF-2118620, CNS-1938180, CCF-2118832, CCF-2106827, CNS-1938709, CCF-2247577.
\end{acks}

%\clearpage

%\bibliographystyle{ACM-Reference-Format}
\bibliographystyle{abbrv}
\bibliography{cite}

\appendix
\section{Proof of Algorithm~\ref{alg:lock-free-op-quiescent} (Lock-Free State-Quiescent HI Multi-Valued Register from Binary Registers)}
\label{app:swsr lock-free register proofs}

Consider a finite execution $\alpha$ of a SWSR register implementation, 
that ends in a configuration with no pending operation by the writer, 
and contains a sequence of $\textsc{Write}(v_1), \ldots, \textsc{Write}(v_k)$ operations, 
$k\geq 0$ (where by abuse of notation, $k=0$ stands for execution with no \textsc{Write} operation). 
Since there is a single writer, for any linearization of $\alpha$, 
this sequence of \textsc{Write} operations must be linearized in order,
as they do not overlap. Thus, let $\seqstate{\alpha} = v_k$, since for any linearization $H$ of $\alpha$, 
$\seqstate{H} = v_k$. 
Note that since we consider a general SWSR register implementation, $\alpha$ can be an execution of both Algorithm~\ref{alg:lock-free-op-quiescent} and Algorithm~\ref{alg:wait-free-quiescent}.

% The next lemma is easy to verify from the writer's code, and it shows the algorithm is state-quiescent HI.

% \begin{lemma}
% \label{lem:reg-op-quiescent}
%     Let $\alpha$ be a finite execution of Algorithm~\ref{alg:lock-free-op-quiescent} that ends in a state-quiescent configuration, and denote $v = \seqstate{\alpha}$. Then $A[v] = 1$ and for every other index $j \neq v$, $A[j] = 0$ in $\mem{\alpha}$.
% \end{lemma}

\lockfreelinearization*

\begin{proof}
     Let $\alpha$ be a finite execution of Algorithm~\ref{alg:lock-free-op-quiescent} that ends in a state-quiescent configuration, and denote $v = \seqstate{\alpha}$.
     It is easy to verify from the writer's code that  $A[v] = 1$ and for every other index $j \neq v$, $A[j] = 0$ in $\mem{\alpha}$. Since every value $1\leq v \leq K$ has a unique memory representation $\can{v}$ that the memory of a state-quiescent configuration is equal to, this implies that the algorithm is state-quiescent HI.
    
    By construction, $H$ is in the sequential specification of the register and includes all completed operations in $H(\alpha)$.
    It remains to show that the linearization respects the real-time order of non-overlapping operations.
    The order between two \textsc{Write} operations respects the real-time order by the construction of $H$, and a \textsc{Read} operation cannot be placed after a \textsc{Write} operation that follows it,
    again by definition of $H$.
    We need only rule out the following two cases.
    
    Write before read: assume a $\textsc{Write}(v)$ operation $W$ returns before a \textsc{Read} operation $R$ begins, but $R$ is placed before $W$ in the linearization. Then, by the construction of $H$, $R$ reads from a \textsc{Write} operation $W'$ that precedes $W$,
    and $R$ is linearized after $W'$ but before $W$.
    But this is impossible:
    in $W$, the writer writes $1$ to $A[v]$ and $0$ to every other index in $A$, and therefore, when $W$ returns, the write of the value $1$ by any operation that precedes $W$ is overwritten (either by 0 or by 1), so $R$ cannot read from $W'$.

    Read before read: assume a \textsc{Read} operation $R_1$ returns before a \textsc{Read} operation $R_2$ begins, but $R_2$ is placed before $R_1$ in the linearization.
    Since $\Read$
    operations are placed after the $\Write$
    operation from which they read,
    the operations $R_1$ and $R_2$
    read from $\Write$ operations $W_1 = \Write(v_1)$ and $W_2 = \Write(v_2)$, respectively,
    such that $W_2$ precedes $W_1$,
    and $R_2$ is placed between $W_2$
    and $W_1$, whereas $R_1$ is placed after $W_1$.
    There are three cases:
    
    {\bf (1) $v_1 = v_2$:}
            when $W_1$ writes 1 to $A[v_1]$ it overwrites $W_2$'s write of 1 to $A[v_2]$.
            Since $R_1$ reads from $W_1$, $W_1$'s write of 1 to $A[v_1]$ precedes $R_1$'s last read of $A[v_1]$.
    But $R_2$ begins after $R_1$ returns,
    and therefore it cannot read from $W_2$,
    as $A[v_1]$ has been overwritten by $W_1$
    before $R_2$ begins.
    
    {\bf (2) $v_1 > v_2$:}
    as above, since $R_1$ reads from $W_1$,
    $W_1$'s write of 1 to $A[v_1]$ precedes $R_1$'s last read of $A[v_1]$. 
    In addition, since $R_2$ reads from $W_2$, $W_2$'s write of 1 to $A[v_2]$ is not overwritten until $R_2$'s last read of $A[v_2]$,
    and in particular,
    until $R_1$ returns (because $R_1$
    returns before $R_2$ begins). However, since $W_2$'s write of 1 to $A[v_2]$ precedes $W_1$'s write of 1 to $A[v_1]$, this implies that $R_1$ reads 1 from $A[v_2]$ during its downward scan, 
    which would cause $R_1$
    to return a value no greater than $v_2$, contradicting our assumption that it returns $v_1$.
    
    {\bf (3) $v_1 < v_2$:} since $R_2$ returns $v_2 > v_1$, $R_2$'s first read of $A[v_1]$ must return 0.
    Recall, however, that
    $W_1$ writes $1$ to $A[v_1]$ before $R_1$ returns (as $R_1$
    reads that $1$),
    and $R_1$ returns before $R_2$ begins.
    Since $R_2$ does not read 1 from $A[v_1]$
    when it scans up for the last time,
    there must be an intervening write $W$
    that occurs after $W_1$ returns and before $R_2$ reads $A[v_1]$
    as it scans up, and $W$ overwrites $W_1$'s write of $1$ to $A[v_1]$. In particular, this implies that $W_1$ returns before 
    $R_2$'s first read of $A[v_2]$, and before returning, $W_1$ overwrites $W_2$'s write of 1 to $A[v_2]$ (as it overwrites all array locations), in contradiction to $R_2$ reading from $W_2$.
\end{proof}

\section{Proof of Algorithm~\ref{alg:wait-free-quiescent} (Wait-Free Quiescent HI Multi-Valued Register from Binary Registers)}
\label{app:swsr wait-free register proofs}

Fix an execution $\alpha$ of Algorithm~\ref{alg:wait-free-quiescent}. 
The \emph{reads from} relation for low-level reads and writes naturally extends also for the $B$ and $\flag$ array.

The next lemma shows that if a $\textsc{TryRead}$ returns $\bot$, then there is a concurrent \textsc{Write} operation that writes to $A$ for the first time \emph{after} the $\textsc{TryRead}$ begins and clears $A$ in an upward direction \emph{before} the $\textsc{TryRead}$ returns.
Intuitively, since there is always at least one index in $A$ that is equal to 1, \textsc{TryRead} returns $\bot$ if the 1-value in $A$ moves in the opposite direction to the direction \textsc{TryRead} reads $A$. This means that \textsc{TryRead} misses
a concurrent \textsc{Write}, which clears $A$ in an upwards direction.

\begin{lemma}
\label{lem:failed-try-read}
    If a $\textsc{TryRead}$ returns $\bot$, then there is a \textsc{Write} operation $W$ that performs Line~\ref{lin:1toA} after the $\textsc{TryRead}$ starts and reaches the loop in Line~\ref{lin:upward-clear} before the $\textsc{TryRead}$ returns.
\end{lemma}

\begin{proof}
There is a \textsc{Write} operation that overlaps $\textsc{TryRead}$,
since otherwise, there is always at least one index that is equal to 1 in $A$.
Since only \textsc{Write} operations write to $A$, 
$\textsc{TryRead}$ reads 1 from some index in $A$ and cannot return $\bot$.
    
    Assume all \textsc{Write} operations concurrent to $\textsc{TryRead}$ perform Line~\ref{lin:1toA} before $\textsc{TryRead}$ starts, 
    and let $W$ be the last such overlapping $\textsc{Write}(v)$ operation. 
    Since $W$ is the last write operation that overlaps $\textsc{TryRead}$, 
    the write of 1 to $A[v]$ in Line~\ref{lin:1toA} in $W$ is not overwritten until $\textsc{TryRead}$ returns. 
    Thus, $\textsc{TryRead}$ reads 1 from $A[v]$, contradicting the fact that it returns $\bot$.

    So there is a \textsc{Write} operation that performs Line~\ref{lin:1toA} after $\textsc{TryRead}$ starts, 
    let $W$ be the first such $\textsc{Write}(v_1)$ operation and assume $W$ does not reach the loop in
    Line~\ref{lin:upward-clear} before $\textsc{TryRead}$ returns.
    Let $v_2$ be the input of the $\textsc{Write}$ operation that precedes $W$, 
    or the initial value if there is no such operation. 
    Since $W$ is the first to perform Line~\ref{lin:1toA} after $\textsc{TryRead}$ starts, 
    the value of $A[v_2]$ is 1 when $\textsc{TryRead}$ starts. 
    Since $W$ does not reach the loop in Line~\ref{lin:upward-clear}, 
    it must write 0 to $A[v_2]$ in the loop in Line~\ref{lin:downward-clear} 
    before $\textsc{TryRead}$ first reads $A[v_2]$. 
    Then $v_2 < v_1$, however, this implies that $\textsc{TryRead}$ reads $A[v_1]$ 
    after $W$ performs Line~\ref{lin:1toA} and before $W$ returns, 
    thus it reads 1 from $A[v_1]$, contradicting that $\textsc{TryRead}$ returns $\bot$.
\end{proof}

\begin{figure}[!tb]
    \centering
    \begin{tikzpicture}
        % reader
        \draw[|-|,solid,thick] (0,0) -- (12,0) node [midway, below, text height = 0.4cm] {$\textsc{Read}$};
        \node at (-.5,0) {$r$};
        \draw[stealth-stealth,thick] (0.5,0) -- ++ (2,0) node [midway, above] {$\textsc{TryRead}$};
        \fill[gray, opacity=0.3] (0.5,-0.1) rectangle ++(2,0.2);
        \draw[thick] (2.5,0.1) -- ++ (0,-0.2);
        \node at (2.9,0.3) {\footnotesize{returns $\bot$}};
        \draw[stealth-stealth,thick] (3.5,0) -- ++ (6.5,0) node [midway, above] {$\textsc{TryRead}$};
        \fill[gray, opacity=0.3] (3.5,-0.1) rectangle ++(6.5,0.2);
        \draw[thick] (10,0.1) -- ++ (0,-0.2);
        \node at (10,-0.3) {\footnotesize{returns $\bot$}};
        \draw[thick] (11.5,0.1) -- ++ (0,-0.2);
        \node at (11.5,0.3) {\footnotesize{reads $B$}};

        % writer
        \draw[|-|,solid,thick] (0,2) -- ++ (3,0) node [midway, below] {$\textsc{Write}$};
        \draw[thick] (2.5,2.1) -- ++ (0,-0.2);
        \node at (2.2,2.5) {\footnotesize{\shortstack{clears $A$ in \\ upwards direction}}};
        \draw[|-|,solid,thick] (3.5,2) -- ++ (6,0) node [midway, below] {$\textsc{Write}$};
        \node at (-.5,2) {$w$};
        \draw[thick] (4.5,2.1) -- ++ (0,-0.2);
        \node at (4.5,2.3) {\footnotesize{$\exists j , B[j]=1$}};
        \draw[thick] (8,2.1) -- ++ (0,-0.2);
        \node at (8,2.3) {\footnotesize{writes 1 to $A$}};
        
        \draw[dotted, thick] (4.5,2) -- (11.5,0);
        \draw[<->,dashed,thick] (0.5,3) -- ++ (11,0) node [midway, above] {\footnotesize{$\flag[1] = 1$ and $\flag[2] = 0$}};
    \end{tikzpicture}
    \caption{Illustrating the proof of Lemma~\ref{lem:val-not-bot}}
    \label{fig:lem:val-not-bot}
\end{figure}
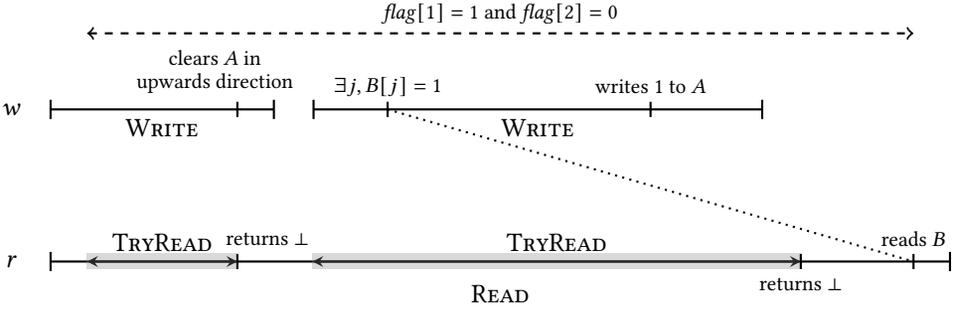

Lemma~\ref{lem:failed-try-read} implies that if two \textsc{TryRead} in Line~\ref{lin:try-read} return $\bot$ in a \textsc{Read} operation $R$, then there is a $\textsc{Write}$ operation, which overlaps the second \textsc{TryRead}, which, if $B$ has no index equal to 1, writes to $B$ before $R$ starts reading $B$. This is illustrated in Figure~\ref{fig:lem:val-not-bot}, and allows us to show that when a \textsc{Read} operation returns, $val \neq \bot$ and it returns a valid value.

\valnotbot*

\begin{proof}
    If one of the two $\textsc{TryRead}$ in Line~\ref{lin:try-read} returns a value different than $\bot$, the lemma holds.
    Otherwise, $R$ performs two $\textsc{TryRead}$ calls that both return $\bot$.
    By Lemma~\ref{lem:failed-try-read}, 
    there is a \textsc{Write} operation $W_1$ that performs Line~\ref{lin:1toA} after the first $\textsc{TryRead}$ starts and reaches the loop in Line~\ref{lin:upward-clear} before the first $\textsc{TryRead}$ returns. 
    Similarly, there is a \textsc{Write} operation $W_2$ that performs Line~\ref{lin:1toA} after the second $\textsc{TryRead}$ starts and reaches the loop in Line~\ref{lin:upward-clear} before the second $\textsc{TryRead}$ returns.
    By the real-time order, $W_1$ precedes $W_2$ and $W_2$ performs Line~\ref{lin:empty-B-cond} after the first $\textsc{TryRead}$ begins and Line~\ref{lin:1toA} before the second $\textsc{TryRead}$ returns.
    
    If the condition in Line~\ref{lin:empty-B-cond} holds, and $W_2$ does not find an index with value 1 in $B$, $W_2$ reads 1 from $\flag[1]$ in Line~\ref{lin:write-1-B-cond} and writes 1 to $B[\lastval]$ in Line~\ref{lin:write-1-B}. Then $W_2$ reads 0 from $\flag[2]$ and 1 from $\flag[1]$ in Line~\ref{lin:writer-clear-B-cond}, and $W_2$ does not write 0 to $B[\lastval]$.
    Since any \textsc{Write} operation that follows $W$ cannot write to $B$ until a write of 0 to $B[\lastval]$ occurs, $R$ reads 1 from $B[\lastval]$ in Line~\ref{lin:reader-read-B} and sets $val = \lastval$.

     If the condition in Line~\ref{lin:empty-B-cond} does not hold, and $W_2$ reads 1 from $B[j]$, then $W_2$ does not perform any write to the $B$ array. This is also true for any \textsc{Write} operation that follows $W$, until a write of 0 to $B[\lastval]$ occurs. 
     Hence, $R$ reads 1 from $B[j]$ in Line~\ref{lin:reader-read-B} and sets $val = j$.
\end{proof}

The next lemma shows that if the writer writes 1 to an index in $B$ in a \textsc{Write} operation, either this write is overwritten by the writer or by an overlapping \textsc{Read} operation.
This allows us to show two properties:
(a) if the reader returns a value read from $B$, it was written by an overlapping \textsc{Write} operation (Lemma~\ref{lem:reads-B-overlaps}), and 
(b) if a reachable configuration is quiescent, all indices in $B$ are equal to 0 in it (Lemma~\ref{lem:reg-quiescent-HI}). The first property helps to prove that the algorithm is linearizable, and the second property shows that the algorithm is quiescent HI.

\begin{restatable}{lemma}{overwriteB}
\label{lem:overwriteB}
    Consider a write of 1 to $B[j]$, $1 \leq j \leq K$, in Line~\ref{lin:write-1-B} by a \textsc{Write} operation W. Then, the write to $B[j]$ is overwritten with a write of $0$ to $B[j]$ by $W$ or by a \textsc{Read} operation $R$ that overlaps $W$, before $W$ and the last \textsc{Read} operation that overlaps $W$ returns.    
\end{restatable}

\begin{proof}
    % Before $W$ writes 1 to $B[j]$ in Line~\ref{lin:write-1-B}, $W$ reads 1 from $\flag[1]$ in Line~\ref{lin:write-1-B-cond}.
    % Let $R_1$ be the \textsc{Read} invocation that writes the 1 to $\flag[1]$ in Line~\ref{lin:raise-flag-1} the read by $W$ in Line~\ref{lin:write-1-B-cond} reads from.
    % $R_1$ must be pending when $W$ reads $\flag[1]$ in Line~\ref{lin:write-1-B}, as a \textsc{Read} invocation writes 0 to $\flag[1]$ in Line~\ref{lin:clear-flag1} before returning.
    If the condition in Line~\ref{lin:writer-clear-B-cond} holds, $W$ writes 0 to $B[j]$ in Line~\ref{lin:writer-clear-B} after the write of 1 to $B[j]$ in Line~\ref{lin:write-1-B} and the lemma holds. 
    Otherwise, $W$ reads 0 from $\flag[2]$ and then 1 from $\flag[1]$ in Line~\ref{lin:writer-clear-B-cond}. Let $R$ be the \textsc{Read} operation that writes the 1 to $\flag[1]$ in Line~\ref{lin:raise-flag-1} that the read by $W$ in Line~\ref{lin:writer-clear-B-cond} reads from (see Figure~\ref{fig:lem:overwriteB}).

    % If $R_1 = R_2$, then $R_1$ did not perform Line~\ref{lin:raise-flag-2} before $W$ reads $\flag[2]$ in Line~\ref{lin:writer-clear-B-cond}.
    % If $R_1 \neq R_2$, 
    Assume $R$ performs Line~\ref{lin:raise-flag-2} and writes 1 to $\flag[2]$ before $W$ reads $\flag[2]$ in Line~\ref{lin:writer-clear-B-cond}. Then, since $W$ reads 0 from $\flag[2]$, $R$ performs Line~\ref{lin:clear-flag2} before $W$
    performs Line~\ref{lin:writer-clear-B-cond}. This implies that $R$ performs Line~\ref{lin:clear-flag1} and write 0 to $\flag[1]$ before $W$ 
    reads $\flag[1]$ in Line~\ref{lin:writer-clear-B-cond}, however, this contradicts that the read in Line~\ref{lin:writer-clear-B-cond} reads the write of 1 in Line~\ref{lin:raise-flag-1} by $R$.
    So $R$ overlaps $W$ and performs Line~\ref{lin:raise-flag-2} after the read of $\flag[2]$ by $W$ in Line~\ref{lin:writer-clear-B-cond}, which happens after the write of 1 to $B[j]$ in Line~\ref{lin:write-1-B}.
    
    By Line~\ref{lin:empty-B-cond}, any \textsc{Write} operation that follows $W$, writes to an index in $B$ only after the write of $W$ to $B[j]$ is overwritten with the value 0.
    Therefore, when $R$ reaches the for loop in Line~\ref{lin:reader-clear-B-cond} before returning, it must overwrite $B[j]$ with the value 0. 
\end{proof}

\begin{figure}[!tb]
    \centering
    \begin{subfigure}{\textwidth}
    \begin{tikzpicture}
        % reader
        \draw[|-|,solid,thick] (0,0) -- (12,0) node [midway, below, text height = 0.5cm] {$\textsc{Read}$};
        \node at (-.5,0) {$r$};
        \draw[thick] (2,0.1) -- ++ (0,-0.2);
        \node at (2,0.3) {\footnotesize{writes 1 to $\flag[2]$}};
        \draw[thick] (4.5,0.1) -- ++ (0,-0.2);
        \node at (4.5,0.3) {\footnotesize{writes 0 to $\flag[1]$}};
        \draw[thick] (7,0.1) -- ++ (0,-0.2);
        \node at (7.5,-0.3) {\footnotesize{writes 0 to $\flag[2]$}};
        % writer
        \draw[|-|,solid,thick] (0,1.3) -- ++ (12,0) node [midway, below, text height = 0.4cm] {$\textsc{Write}$};
        \node at (-.5,1.3) {$w$};
        \draw[thick] (1,1.4) -- ++ (0,-0.2);
        \node at (1,1.6) {\footnotesize{writes 1 to $B[j]$}};
        \draw[thick] (8,1.4) -- ++ (0,-0.2);
        \node at (8,1.6) {\footnotesize{reads $\flag[2] = 0$}};
        \draw[thick] (10.5,1.4) -- ++ (0,-0.2);
        \node at (10.5,1.6) {\footnotesize{reads $\flag[1] = 1$}};
        \draw[dotted, thick] (8,1.3) -- (7.,0);
        \draw[dotted, thick, red] (10.5,1.3) -- (4.5,0);
    \end{tikzpicture}
    \caption{$W$ reads $\flag[2]$ before $R$ writes 0 to $\flag[2]$, otherwise, we get a contradiction.}
    \label{fig:lem:overwriteBa}
    \vspace{2mm}
    \end{subfigure}
    \begin{subfigure}{\textwidth}
    \begin{tikzpicture}
        % reader
        \draw[|-|,solid,thick] (0,0) -- (12,0) node [midway, below, text height = 0.5cm] {$\textsc{Read}$};
        \node at (-.5,0) {$r$};
        \draw[thick] (1,0.1) -- ++ (0,-0.2);
        \node at (1,0.3) {\footnotesize{writes 1 to $\flag[1]$}};
        \draw[thick] (6,0.1) -- ++ (0,-0.2);
        \node at (5.5,0.3) {\footnotesize{writes 1 to $\flag[1]$}};
        \draw[thick] (8,0.1) -- ++ (0,-0.2);
        \node at (8,-0.3) {\footnotesize{writes 1 to $\flag[2]$}};
        \draw[thick] (11,0.1) -- ++ (0,-0.2);
        \node at (11,0.3) {\footnotesize{writes 0 to $B[j]$}};
        % writer
        \draw[|-|,solid,thick] (0,1.3) -- ++ (12,0) node [midway, below, text height = 0.4cm] {$\textsc{Write}$};
        \node at (-.5,1.3) {$w$};
        \draw[thick] (2,1.4) -- ++ (0,-0.2);
        \node at (2,1.6) {\footnotesize{writes 1 to $B[j]$}};
        \draw[thick] (4.5,1.4) -- ++ (0,-0.2);
        \node at (4.5,1.6) {\footnotesize{reads $\flag[2] = 0$}};
        \draw[thick] (7,1.4) -- ++ (0,-0.2);
        \node at (7,1.6) {\footnotesize{reads $\flag[1] = 1$}};
        \draw[dotted, thick] (7,1.3) -- (6,0);
        \draw[dotted, thick] (2,1.3) -- (11,0);
    \end{tikzpicture}
    \caption{$R$ writes 0 to $B[j]$ if $W$ does no do so.}
    \label{fig:lem:overwriteBb}
    \end{subfigure}
    \caption{Illustrating the proof of Lemma~\ref{lem:overwriteB}}
    \label{fig:lem:overwriteB}
\end{figure}
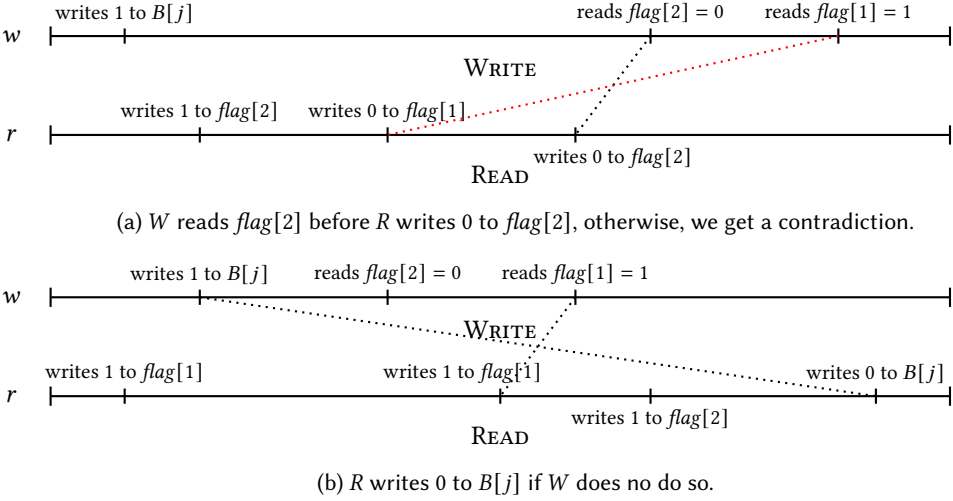

\readsboverlaps*

% \begin{lemma}
%     Consider a \textsc{Read} invocation $R$, if $B[j] = 1$, for some $0 \leq j \leq K-1$, between $R$ invocation and response, then the \textsc{Write} invocation $W$ that writes this value overlaps $R$.
% \end{lemma}

\begin{proof}
If $W$ does not overlap $R$, then $W$ precedes $R$. 
This implies that $W$ cannot overwrite its write of 1 to $B[j]$ with 0. 
Since there is a single reader, this contradicts Lemma~\ref{lem:overwriteB}, 
as the write of 1 to $B[j]$ by $W$ is not overwritten by $W$ 
or before the last \textsc{Read} operation that overlaps $W$ returns.
\end{proof}

\begin{lemma}
\label{lem:reg-quiescent-HI}
    If a configuration $C$ reachable by an execution of Algorithm~\ref{alg:wait-free-quiescent} is quiescent, then the value of $B[j]$ is 0 in $\mem{C}$ for all $1\leq j \leq K$.
\end{lemma}

\begin{proof}
    At initialization, $B[j] = 0$ for all $1\leq j \leq K$. 
    A write of 1 to $B[j]$ in the execution leading to $C$ happens in Line~\ref{lin:write-1-B}  by a \textsc{Write} operation $W$.
    Since $C$ is quiescent, $W$ and any \textsc{Read} operation that overlaps $W$ returns by the time the execution reaches $C$.
    By Lemma~\ref{lem:overwriteB}, the write of 1 to $B[j]$ is overwritten with value 0 by the time the execution reaches configuration $C$.
\end{proof}

\waitfreeqhireg*

\begin{proof}[Proof of history independence]
    Consider a finite execution $\alpha$ that ends with a quiescent configuration and let $v = \seqstate{\alpha}$.
    Since only the reader writes to $\flag[1]$ and $\flag[2]$, and before a \textsc{Read} operation returns it writes 0 to $\flag[1]$ in Line~\ref{lin:clear-flag1} and 0 $\flag[2]$ in Line~\ref{lin:clear-flag2}, $\flag[1]$ and $\flag[2]$ are both equal 0 in $\mem{\alpha}$.
    Also, $A[v] = 1$ and for any other index $j\neq v$, $A[v] = 0$ in $\mem{\alpha}$.

    At initialization, $B[j] = 0$ for all $1\leq j \leq K$. 
    A write of 1 to $B[j]$ in the execution happens in Line~\ref{lin:write-1-B}  by a \textsc{Write} operation $W$.
    Since the execution ends in a quiescent configuration, $W$ and any \textsc{Read} operation that overlaps $W$ return by the time the execution ends. Hence,
    by Lemma~\ref{lem:overwriteB}, the write of 1 to $B[j]$ is overwritten with value 0 by the time the execution ends. This implies that $B[j] = 0$ in $\mem{\alpha}$, for all $1\leq j \leq K$.

     By Lemma~\ref{lem:reg-quiescent-HI} ,$B[j] = 0$ in $\mem{\alpha}$, for all $1\leq j \leq K$.
     
    Since this holds for any execution $\alpha'$ that ends in a quiescent configuration such that $\seqstate{\alpha'} = v$, this concludes the theorem.
\end{proof}

\section{Impossibility of Wait-Free, State-Quiescent HI  Queue Implementation}
\label{app:queue-impossibility}

Consider a wait-free state-quiescent implementation of a queue 
from $m \geq 1$ base objects $\obj_1,\ldots, \obj_m$.
For each base object $\obj_i$,
let $Q_i$ be the state space of $\obj_i$;
we assume that $|Q_i| \leq t$.
This is the only assumption we make about the base objects.
Since each base object has at most $t$ states and there are $t+1$ representative states, for every object $\obj_\ell$, $1\leq \ell \leq m$, there are two representative states $q_i \neq q_j$ such that $\can{q_i}[\ell] = \can{q_j}[\ell]$.

We consider executions with two processes:
\begin{itemize}
    \item A ``reader'' process $r$, which executes a single $\peek$ operation, and
    \item A ``changer'' process $c$, which repeatedly invokes $\enq$ and $\deq$ operations. 
\end{itemize}

The executions that we construct have the following form:
\begin{equation*}
    \alpha_{{i_0},{i_1}\ldots,{i_k}} = S(i_0,i_1), r_1,
    S(i_1,i_2), r_2, \ldots, S(i_{k-1},i_k), r_k
\end{equation*}
where $i_0 = 0$ and $S(i_j, i_{j+1})$ is a sequence of at most two operations
executed by the changer process during which the reader process takes no steps,
and $r_i$ is a single step by the reader process.
The reader
executes a single $\peek$
operation that is invoked immediately
after the first sequence of
operations complete, and we will argue that the reader never returns.

In any linearization of $\alpha_{{i_0},{i_1}\ldots,{i_k}}$,
the operations' sequence $S(i_0,i_1),\ldots,S(i_{k-1},i_k)$
must be linearized in order,
as they do not overlap.
Furthermore, $\peek$
operation carried out by the reader process is not state-changing.
Thus, since each sequence $S(i_j, i_{j+1})$ takes the object from state $q_{i_j}$ to state $q_{i_{j+1}}$,
the linearization of $\alpha_{{i_0},{i_1}\ldots,{i_k}}$ ends
with the object in state $q_{i_k}$,
and we abuse the terminology by saying that the execution ``ends at state $q_{i_k}$''.

We say that execution $\alpha_{{i_0},{i_1}\ldots,{i_k}}$
\emph{avoids} index $i$, $0\leq i \leq t$,
if $i_j \neq i$ for every $1\leq j \leq k$.

\begin{lemma}
    If an execution $\alpha_{{i_0},{i_1}\ldots,{i_k}}$ avoids index $i$, $0\leq i \leq t$,
    then the $\peek$ operation cannot return response
    value $r_{i}$ at any point in $\alpha_{{i_0},{i_1}\ldots,{i_k}}$.
    \label{lemma:exclusive_returns_app}
\end{lemma}
\begin{proof}
    Fix an execution $\alpha_{{i_0},{i_1}\ldots,{i_k}}$
    that avoids index $i$,
    and recall that in any linearization,
    the operations' sequence
    $S(i_0,i_1),\ldots,S(i_{k-1},i_k)$
    must be linearized in-order,
    as they are non-overlapping operations by the same process.
    The $\peek$ operation cannot be linearized before
    the operations sequence $S(i_0,i_1)$,
    because it is only invoked after the last operation in the sequence completes.
    Thus, the $\peek$ operation 
    either does not return in $\alpha_{{i_0},{i_1}\ldots,{i_k}}$,
    or it is linearized after some operation in sequence
    $S(i_j, i_{j+1})$ where $j \geq 0$.
    In the latter case, we show that the value returned by the $\peek$ is either $r_{i_j}$ or $r_{i_{j+1}}$ and as $\alpha_{{i_0},{i_1}\ldots,{i_k}}$ avoids $i$, $r_{i_j},r_{i_{j+1}} \neq r_{i}$.
    If $S(i_j, i_{j+1})$ contains a single operation, 
    which takes the object from state $q_{i_j}$ to state $q_{i_{j+1}}$, the $\peek$ operation returns response $r_{i_{j+1}}$.
    If $S(i_j, i_{j+1})$ contains two operations, since  the $\peek$ operation returns response $r_{i_j}$ from the state resulting from applying $\enq(i_{j+1})$ from state $q_{i_j}$, if the $\peek$ operation is linearized after the {\enq} operation, it returns response $r_{i_j}$. Otherwise, if it is linearized after the $\deq$ operation, it returns response $r_{i_{j+1}}$.
    Therefore in $\alpha_{{i_0},{i_1}\ldots,{i_k}}$ the
    $\peek$ operation either does not return,
    or returns a value different than $r_{i}$.
\end{proof}

Using the fact that each base object has at most $t$
possible states, we can construct 
$t+1$ arbitrarily long executions
that the reader cannot distinguish from one another,
such that each response $r_i$ is avoided by one of the $t$ executions.
The construction is inductive, with the step extending the executions
captured by the following lemma:
\begin{lemma}
    Fix $k \geq 0$,
    and suppose we are given $t$
    executions of the form
    $\alpha_i = \alpha_{{i_0},{i_1}\ldots,{i_k}}$
    for $i = 0,\ldots,t$,
    such that $\alpha_1 \localind{r} \ldots \localind{r} \alpha_t$,
    and each $\alpha_i$ avoids index $i$.
    Then we can extend
    each $\alpha_i$
    into an execution 
    $\alpha_i' = \alpha_{{i_0},{i_1}\ldots,{i_{k+1}}}$
    that also avoids $i$,
    such that $\alpha_1' \localind{r} \ldots \localind{r} \alpha_t'$.
    \label{lemma:exclusive_executions_app}
\end{lemma}

\begin{proof}
    Let $\alpha_i = \alpha_{{i_0},{i_1}\ldots,{i_k}}$
    for $i = 0,\ldots,t$ be executions
    satisfying the conditions of the lemma,
    and let us construct extensions $\alpha_i' = \alpha_{{i_0},{i_1}\ldots,{i_{k+1}}}$ for each $i = 0,\ldots,t$.
    By assumption, the reader is in the same local state
    at the end of all executions $\alpha_i$ for $0 \leq i \leq t$,
    and so its next step is the same in all of them.
    Our goal is to 
    choose a next index $i_{k+1}$
    for each $i = 0,\ldots,t$,
    and extend each $\alpha_i = \alpha_{{i_0},{i_1}\ldots,{i_k}}$
    into $\alpha_i' = \alpha_{{i_0},{i_1}\ldots,{i_{k+1}}}$
    by appending the operations sequence $S(i_k, i_{k+1})$,
    followed by a single step of the reader.
    We must do so in a way that both continues to avoid
    $i$, and maintains indistinguishability to the reader.

    Let $\obj_{\ell}$
    be the base object
    accessed by the reader in its next step
    in all $t$ executions.
    Because $\obj_{\ell}$
    has only $t$
    possible memory states
    and there are $t+1$ states $q_0, \ldots, q_{t}$,
    there must exist two distinct states $q_j, q_{j'}$, $j \neq j'$,
    such that
    $\can{q_j}[\ell] = \can{q_{j'}}[\ell]$.
    For every $0 \leq i \leq t$,
    there is an index $i_{k+1} \in \set{j,j'}$
    such that $i_{k+1} \neq i$:
    if $i \notin \set{j, j'}$
    then we choose between $j$ and $j'$
    arbitrarily, and if $i = j$ or $i = j'$
    then we choose $i_{k+1} = j'$
    or $i_{k+1} = j$, respectively.
    
    We extend each $\alpha_i = \alpha_{{i_0},{i_1}\ldots,{i_k}}$
    into $\alpha_i' = \alpha_{{i_0},{i_1}\ldots,{i_{k+1}}}$
    by appending
    a complete $S(i_k, i_{k+1})$
    operation sequence, followed by a single
    step of the reader.
    The resulting execution $\alpha_i'$
    still avoids $i$,
    as we had for every $1\leq j \leq k$, $i_j \neq i$,
    and the new index also satisfies $i_{k+1} \neq i$.
    Moreover,
    when the reader takes its step,
    it observes the same state for the base
    object $\obj_{\ell}$ that it accesses in all executions,
    as all of them end in either state $q_j$
    or state $q_{j'}$,
    and $\can{q_j}[\ell] = \can{q_{j'}}[\ell]$.
    Therefore, the reader cannot distinguish
    the new executions from one another.
\end{proof}

We can force the $\peek$ operation to never return, to obtain the main result of this section:

\queue*

\begin{proof}
    We construct $t+1$
    arbitrarily long executions, in each of which a $\peek$
    operation takes infinitely many steps but never returns.
    The construction uses Lemma~\ref{lemma:exclusive_executions_app}
    inductively:
    we begin with empty executions,
    $\alpha_0^0 = \ldots = \alpha_{t}^0 = \alpha_{i_0}$.
    These executions trivially satisfy the conditions of Lemma~\ref{lemma:exclusive_executions},
    as each $\alpha_i^0$ avoids $i$ (technically,
    it avoids \emph{all} indices $j$),
    and furthermore,
    since the reader has yet to take a single step
    in any of them,
    so it is in the same local state in all executions.
    We repeatedly apply Lemma~\ref{lemma:exclusive_executions}
    to extend these executions,
    obtaining for each $k \geq 0$
    a collection of
    $t+1$ executions
    $\alpha_i^k = \alpha_{i_0,i_1,\ldots,i_k}$, $0\leq i \leq t$,
    such that each $\alpha_i^k$ avoids $i$,
    and the reader cannot distinguish the executions from one another.

    Suppose for the sake of contradiction
    that the reader returns a value
    $r$ at some point in $\alpha_i^k$.
    Then it returns the same value $r$
    at some point in each execution $\alpha_i^k$
    for each $i = 0,\ldots,t$,
    as it cannot distinguish these executions,
    and its local state encodes
    all the steps it has taken,
    including whether it has 
    returned a value,
    and if so, what value.
    By Lemma~\ref{lemma:exclusive_returns_app},
    for each $i = 0,\ldots,t$,
    since execution $\alpha_i^k$
    avoids $i$,
    we must have $r \neq r_i$.
    Since the response space $R = \set{r_0,\ldots, r_{t}}$, this means that there is no value
    that the reader can return, a contradiction.

    Continuing on in this way,
    we can construct arbitrarily long executions,
    with the reader taking more and more steps
    (since in an execution $\alpha_{i_0,i_1,\ldots,i_k}$
    the reader takes $k$ steps)
    but never returning.
    This contradicts the wait-freedom of the implementation.    
\end{proof}

\section{A History-Independent Universal Implementation}

\subsection{Properties of Algorithm~\ref{alg:universal-construct-R-LLSC}}
\label{app:universal-impl-correct}

\clearedinv*

\begin{proof}
    The value of $announce[\pi(op)]$ can change in Line~\ref{lin:announce-inv}, Line~\ref{lin:replace-inv-rsp} or Line~\ref{lin:clear-my-announce}. The first step of operation $op$ is to write the input operation descriptor in $announce[\pi(op)]$. Right before $op$ returns, it writes $\bot$ to $announc[i]$ in Line~\ref{lin:clear-my-announce}. 
    By the condition in Line~\ref{lin:start-loop}, Line~\ref{lin:wait-til1} and Line~\ref{lin:wait-til2}, $op$ exits the while loop in Lines~\ref{lin:start-loop}--\ref{lin:end-loop} only if it reads a response value from $R$ from
    $announce[\pi(op)]$. Since a read in Line~\ref{lin:start-loop}, Line~\ref{lin:wait-til1} or Line~\ref{lin:wait-til2} follows Line~\ref{lin:announce-inv}, the value $op$ writes to $announce[\pi(op)]$ in Line~\ref{lin:announce-inv} is overwritten by another operation.
    This can only happen in Line~\ref{lin:replace-inv-rsp}, where a process writes a value from $R$ to $announce[\pi(op)]$.
    % and $O\cap R = \emptyset$, the operation value $op$ writes in Line~\ref{lin:announce-inv} must be overwritten with a value from $R$. 
\end{proof}

\universalvaluepairs*

\begin{proof}
    The value of $head$ can be changed only due to a successful {\SC} in Line~\ref{lin:sc-headA1}, Line~\ref{lin:sc-headB}, or Line~\ref{lin:sc-headA2}.
    Consider a successful ${\SC}(head, \tup{q_2,r_2})$, performed in one of these lines by operation $op$, which replaces the value of $head$ from $\tup{q_1,r_1}$ to $\tup{q_2,r_2}$. That is, the previous write to $head$ before this {\SC} writes the value $\tup{q_1,r_1}$.
    Before a successful {\SC} by process $p_i$, there must be a previous {\LL} by the same process that returns the current value replaced in the {\SC}.
    If the successful {\SC} happens in Line~\ref{lin:sc-headA1} or Line~\ref{lin:sc-headB}, then the preceding {\LL} happens in Line~\refl{lin:ll-head1} and it returns the value $\tup{s_1,r_1}$. 
    If the {\SC} happens in Line~\ref{lin:sc-headB}, then the condition in Line~\ref{lin:head-state-check} holds, implying that $r_1 = \bot$ and by the code, $r_2 \neq \bot$.
    If the {\SC} happens in Line~\ref{lin:sc-headA1}, then the condition in Line~\ref{lin:head-state-check} does not hold, implying that $r_1 \neq \bot$.
    By the code, $s_1 = s_2$ and $r_2 = \bot$.

     If the {\SC} happens in Line~\ref{lin:sc-headA2}, the preceding {\LL} happens in Line~\refl{lin:ll-head2}, and it returns the value $\tup{s_1,r_1}$. By the code, the condition in Line~\ref{lin:sc-headA2} holds and $r_1 \neq \bot$. In addition,
     $s_1 = s_2$ and $r_2 = \bot$.
\end{proof}

\begin{lemma}
\label{lem:cleared-linearized-inv}
    Consider transition from mode $A_{i-1}$ to mode $B_i$, $i\geq 1$, and let $op^*$ be the operation applied by this transition, then when the algorithm transitions to mode $A_i$ operation $op^*$ is cleared.
\end{lemma}

\begin{proof}
    The transition from mode $B_i$ to mode $A_i$ is due to a successful {\SC} to $head$ by operation $op$, either in Line~\ref{lin:sc-headA1} or in Line~\ref{lin:sc-headA2}.
    Before the successful {\SC} by process $p_i$ that transitions the algorithm to mode $B_i$, $p_i$ reads in Line~\ref{lin:read-help-announce} or Line~\ref{lin:read-my-Announce} the value $op^*$ writes in Line~\ref{lin:announce-inv}, thus, this write to $announce[\pi(op^*)]$ happens before the transition from mode $A_{i-1}$ to mode $B_i$, which happens before the transition from mode $B_i$ to mode $A_i$.
    
    Assume the successful {\SC} that transitions the algorithm from mode $B_i$ to mode $A_i$ happens in Line~\ref{lin:sc-headA1}, and let $a$ be the value the ${\LL}(announce[j])$ in Line~\refl{lin:ll-announce} returns. 
    Since the {\SC} is successful, the ${\LL}(head)$ in Line~\refl{lin:ll-head1} returns in mode $B_i$, and the algorithm stays in mode $B_i$ up until the successful {\SC}. Since $op$ performs ${\LL}(announce[j])$ after the ${\LL}(head)$ and before the {\SC}, the ${\LL}(announce[j])$ returns in mode $B_i$ and either $a$ is the value $op^*$ writes in Line~\ref{lin:announce-inv}, or a newer value written after this write.
    If $a\notin O$, then $op^*$ is cleared before the {\LL} returns.
    Otherwise, $op$ performs an ${\SC}$ to $announce[j]$ in Line~\ref{lin:replace-inv-rsp}.
    Both outcomes of the {\SC} imply that the value $a$ is overwritten with another value.
    If the {\SC} is successful then it overwrites the value $a$, otherwise, if it fails, another {\SC} or {\St} that follows the {\LL} in Line~\refl{lin:ll-announce} and precedes the {\SC} overwrites the value $a$.
    % If it is successful, it clears an operation by $\pi(op^*)$.
    % If it does not clear $op^*$, then $op^*$ must have returned, and a new operation by $\pi(op^*)$ was invoked. By Lemma~\ref{lem:cleared-inv}, $op^*$ must be cleared before returning.
    % Otherwise, if the ${\SC}$ fails, $op^*$ is cleared before the ${\SC}$ returns.
    Hence, when $op$ performs Line~\ref{lin:sc-headA1}, $op^*$ must be cleared.

    If the successful {\SC} happens in Line~\ref{lin:sc-headA2}, then by the condition in Line~\ref{lin:sc-headA2}, $op = op^*$.
    Since $op^*$ exits the loop in Lines~\ref{lin:start-loop}--\ref{lin:end-loop}, by Invariant~\ref{lem:cleared-inv}, $op^*$ is cleared before the successful {\SC} is performed.
\end{proof}

\appliedlinearized*

\begin{proof}
    The transition from mode $A_{i-1}$ to mode $B_i$ is due to a successful {\SC} to $head$ by operation $op$ in Line~\ref{lin:sc-headB}. The preceding ${\LL}(head)$ by $op$ happens in line~\refl{lin:ll-head1} and it returns $\tup{q,\bot}$. 
    This {\LL} returns in mode $A_{i-1}$ and
    the algorithm stays in mode $A_{i-1}$ up until the successful {\SC}, thus, $state(i-1) = q$.
    $op$ writes to $head$ the tuple $\tup{q_i, \tup{{rsp}_i, \wildcard}}$ in the successful {\SC} in Line~\ref{lin:sc-headB}, where, by Line~\ref{lin:apply}, $\Delta(state(i-1),I(op^*)) = (q_i, {rsp}_i)$.

    If $op^*$ was applied by an earlier transition from mode $A_{j-1}$ to mode $B_j$ for $j < i$,
    then, by Lemma~\ref{lem:cleared-linearized-inv}, $op^*$ is cleared before the transition to mode $A_j$. This contradicts that $op$ reads in Line~\ref{lin:read-help-announce} or Line~\ref{lin:read-my-Announce} the value $op^*$ writes in Line~\ref{lin:announce-inv}, since this read happens in mode $A_{i-1}$.
    % Therefore, $op^*$ must be applied for the first time.
    \end{proof}

    \returnhelper*

\begin{proof}
    The proof is by induction on $i$.
    To prove the two properties in the lemma statement, we add an additional property to the induction hypothesis:
    \begin{enumerate}
    \setcounter{enumi}{2}
        \item If ${\LL}(head)$ in Line~\refl{lin:ll-head1}, by operation $op$, returns in mode $B_i$ and $op$ writes to $announce[j]$ in Line~\ref{lin:replace-inv-rsp} at the same iteration, then this write must replace the value $op^*$ writes in Line~\ref{lin:announce-inv}.
    \end{enumerate}

    {\bf Base case:} In mode $A_0$, no operation is linearized yet.
    No operation is cleared by the transition from mode $A_0$ to more $B_1$, 
    since for a process to perform Line~\ref{lin:replace-inv-rsp},
    it must first perform $\LL(head)$ in Line~\ref{lin:ll-head1} that returns in a B mode. 
    By Invariant~\ref{lem:cleared-inv}, this also implies no operation can return yet.
    
    {\bf Induction step:}
    Consider the transition from mode $A_{i-1}$ to mode $B_i$, $i\geq 1$, and assume the induction hypothesis holds for any $j < i$.
    If operation $op^*$ returns, by Invariant~\ref{lem:cleared-inv}, it must be cleared beforehand in Line~\ref{lin:replace-inv-rsp} by some operation $op$. Operation $op$ reads $head$ in Line~\refl{lin:ll-head1} and then writes to $announce[j]$ in Line~\ref{lin:replace-inv-rsp}. When the algorithm is in mode $A_{i-1}$, the induction hypothesis implies that a successful {\SC} in Line~\ref{lin:replace-inv-rsp} does not replace the value $op^*$ writes in Line~\ref{lin:announce-inv}.
    Hence, $op^*$ is not cleared when the algorithm transitions to mode $B_i$.
    By Lemma~\ref{lem:cleared-linearized-inv}, $op^*$ is cleared before the transition to mode $A_i$. Hence, $op^*$ is cleared in mode $B_i$.

    Before returning, $op^*$ performs Line~\ref{lin:ll-head2}.
    If $op^*$ finishes performing Line~\ref{lin:ll-head2} on the left-hand side, then the  ${\LL}(head)$ in Line~\refl{lin:ll-head2} returns.
    If the algorithm is in mode $B_i$ when the ${\LL}(head)$ returns, 
    it returns a tuple that contains $\pi(op^*)$. Then, in Line~\ref{lin:sc-headA2}, $op^*$ attempts a mode transition. If the {\SC} to $head$ succeeds, then necessarily a mode transition occurs. Otherwise, a different {\SC} to $head$ by a different operation must succeed between the ${\LL}(head)$ and failed ${\SC}$ operation, and this must transition the algorithm mode.
    If $op^*$ finishes performing Line~\ref{lin:ll-head2} on the right-hand side, $op^*$ performs a ${\Ld}(head)$
    in Line~\ref{lin:wait-til3} that returns a tuple that doesn't contain $\pi(op^*)$. Thus, the algorithm transitions from mode $B_i$ before the {\Ld} returns.

    Consider an operation $op$ that performs an ${\LL}(head)$ in Line~\refl{lin:ll-head1} that returns in mode $B_i$ and writes to $announce[j]$ in Line~\ref{lin:replace-inv-rsp} at the same iteration.
    Before the successful {\SC} to $announce[j]$ in Line~\ref{lin:replace-inv-rsp}, $op$ performs an ${\LL}(head)$ in Line~\refl{lin:ll-head1} and a successful ${\VL}(head)$ in Line~\ref{lin:vl-head}. 
    This guarantees that there was no mode change between the {\LL} in Line~\refl{lin:ll-head1} and the {\VL} in Line~\ref{lin:vl-head} and the algorithm is in mode $B_i$ between the two operations.
    We showed that in mode $B_i$ $op^*$ is not cleared yet and by Invariant~\ref{lem:cleared-inv}, is still pending.
    Between the ${\LL}(head)$ and ${\VL}(head)$, $op$ performs an ${\LL}(announce[j])$ in Line~\refl{lin:ll-announce} that precedes the successful {\SC} to $announce[j]$, which returns a value from $O$. Since there can only be one pending operation by each process, the {\LL} in Line~\refl{lin:ll-announce} returns the value $op^*$ writes in Line~\ref{lin:announce-inv}. This implies that if the {\SC} in Line~\ref{lin:replace-inv-rsp} succeeds, it replaces this value.

    By the condition in Line~\ref{lin:head-state-check}, the $\LL(head)$ in Line~\refl{lin:ll-head1} that precedes a successful {\SC} to $announce[j]$ in Line~\ref{lin:replace-inv-rsp}, returns in a $B$ mode. 
    Since $op^*$ is cleared in mode $B_i$, we showed that if the $\LL(head)$ returns in mode $B_i$, then the successful {\SC} to $announce[j]$ also happens in mode $B_i$. By the induction hypothesis, if the  $\LL(head)$ returns in mode $B_j$, $j<i$, the {\SC} to $announce[j]$ happens in mode $B_j$. 
    If the $\LL(head)$ returns in mode $B_j$, $j>i$,  the {\SC} to $announce[j]$ cannot happen in a previous mode, and thus, cannot happen in mode $A_i$.
    So, a  successful {\SC} to $announce[j]$ cannot happen in mode $A_i$ and no operation is cleared in mode $A_i$.
\end{proof}

\quiescentannouncebot*

\begin{proof}
    At initialization, the value of $announce[i]$ is $\bot$, and the last step before a state-changing operation by $p_i$ returns is to write $\bot$ to $announce[i]$ in Line~\ref{lin:clear-my-announce}.
    
    The only place where a process different from $p_i$ can change the value of $announce[i]$, is in Line~\ref{lin:replace-inv-rsp}. Since this happens in an ${\SC}$ to $announce[i]$ and the preceding ${\LL}(announce[i])$ in Line~\refl{lin:ll-announce} returns a value from $O$, this write cannot overwrite the value $\bot$.
    Thus, if $p_i$ has no pending state-changing operation, the value $\bot$, there from initialization or written by the last state-changing operation by $p_i$, cannot be overwritten.
\end{proof}

\subsection{Lock-Free Perfect-HI R-LLSC Object from Atomic {\CAS}}
\label{app:R-LLSC}

We translate each value $x\in V\times \braces{0,1}^n$, stored in the $\CAS$ object $X$, to a state of the R-LLSC object, 
$\seqstate{x} = \parens*{x.val, \set{p_i\,|\,x.context[i] = 1}}$.
The proof of Theorem~\ref{thm:RLLSCF-from-CAS} follows immediatly from the next lemma:

\begin{restatable}{lemma}{rllscfseqspecification}
\label{lem:rllscf-seq-specification}
    \sloppy $h_{llsc}(\alpha)$ is a linearization of $\alpha$  
    and $\seqstate{h_{llsc}(\alpha)} = \seqstate{\mem{\alpha}}$.
\end{restatable}

\begin{proof}
If an operation $op$ in $H(\alpha)$ has linearization point,
then $lin(op)$ is between $op$'s invocation and response, since it maps $op$ to a step it performs. 
This implies that $h_{llsc}(\alpha)$ contains all completed operations in $H(\alpha)$.
It is remains to show that $h_{llsc}(\alpha)$ is in the sequential specification of the R-LLSC object.
The proof is by induction on the linearization points in $\alpha$.
    
{\bf Base case:} $h_{llsc}(\alpha)$ maps the empty execution $\alpha = C_0$ to an empty history. 
Trivially, $h_{llsc}(\alpha)$ is a linearization of $\alpha$, 
and in the algorithm initialization we have that $\seqstate{h_{llsc}(C_0)} = \seqstate{\mem{C_0}}$.

{\bf Induction step:} assume $\alpha$ includes exactly $i+1$ linearization points, 
$i\geq 0$, and the induction hypothesis holds for the prefix $\alpha'$ of $\alpha$ 
that contains exactly the first $i$ linearization points.
    Let $op$ be the operation mapped by $lin(op)$ to the $i+1$ linearization point. 
    Note that at each value change of $X$ a linearization point is defined.
    Hence, the value of $X$ does not change between two consecutive linearization points.
    Consider the different cases according to the invoked operation $op$:
    \begin{description}
        \item[${\LL}(\mathcal{O})$:]
        If $op$ performs a successful ${\CAS}(x, cur, new)$ in Line~\ref{lin:ll-cas}, then $\mem{\alpha'} = cur$. 
        Since $op$ returns $cur.val$ and $\seqstate{h_{llsc}(\alpha')} = \seqstate{cur}$, $h_{llsc}(\alpha)$ is in the sequential specification. 
        The only possible difference between $cur$ and the newly written value $new$ is that $new.context[\pi(op)] = 1$. Following the sequential specification, $\seqstate{h_{llsc}(\alpha)} = \seqstate{new} = \seqstate{\mem{\alpha}}$.
     
        \item[${\VL}(\mathcal{O})$:]
        Let $cur$ be the value the $\textsc{Read}(X)$ by $op$ returns in Line~\ref{lin:vl-read}, then $\mem{\alpha'} = cur$. Since $op$ returns $cur.context[\pi(op)]$ and $\seqstate{h_{llsc}(\alpha')} = \seqstate{cur}$, $h_{llsc}(\alpha)$ is in the sequential specification. 
        Since the object's state and memory did not change, $\seqstate{h_{llsc}(\alpha)} = \seqstate{h_{llsc}(\alpha')} = \seqstate{\mem{\alpha'}} = \seqstate{\mem{\alpha}}$.

        \item[${\RL}(\mathcal{O})$:]
        Since $op$ always returns \textit{true}, $h_{llsc}(\alpha)$ is in the sequential specification.
        If $op$ performs a successful ${\CAS}(x, cur, new)$ in Line~~\ref{lin:rl-cas}, then $\mem{\alpha'} = cur$. 
        The only possible difference between $cur$ and the newly written value $new$ is that $new.context[\pi(op)] = 0$. Following the sequential specification, $\seqstate{h_{llsc}(\alpha)} = \seqstate{new} = \seqstate{\mem{\alpha}}$.
        
        If a $\textsc{Read}(X)$ by $op$ returns $cur$, either in Line~\ref{lin:rl-read1} or Line~\ref{lin:rl-read2}, such that $cur.context[\pi(op)] = 0$. Then $\mem{\alpha'} = cur$ and $\pi(op) \notin \seqstate{cur}.context$.
        Hence, the object's state does not change and so does the memory, and 
        $\seqstate{h_{llsc}(\alpha)} = \seqstate{h_{llsc}(\alpha')} = \seqstate{\mem{\alpha'}} = \seqstate{\mem{\alpha}}$.

        \item[${\SC}(\mathcal{O}, v)$:]
        If $op$ performs a successful ${\CAS}(x, cur, new)$ in Line~\ref{lin:sc-cas}, then $\mem{\alpha'} = cur$. 
        In this case, $op$ returns \textit{true}, and by the condition in Line~\ref{lin:sc-valid-cond}, 
        $cur.context[\pi(op)] = 1$.
        Hence, $h_{llsc}(\alpha)$ is in the sequential specification.
        The newly written value contains the value $v$ with an empty $context$ field and
        by the sequential specification, 
        $\seqstate{h_{llsc}(\alpha)} = \seqstate{new} = \seqstate{\mem{\alpha}}$.
        
        If a $\textsc{Read}(X)$ by $op$ returns $cur$, either in Line~\ref{lin:sc-read1} or in Line~\ref{lin:sc-read2}, 
        such that $cur.context[\pi(op)] = 0$. Then $\mem{\alpha'} = cur$ and $\pi(op) \notin \seqstate{cur}.context$.
        Since $op$ returns \textit{false}, $h_{llsc}(\alpha)$ is in the sequential specification.
        In addition, the object's state does not change and so does the memory, and 
        $\seqstate{h_{llsc}(\alpha)} = \seqstate{h_{llsc}(\alpha')} = \seqstate{\mem{\alpha'}} = \seqstate{\mem{\alpha}}$
        
        \item[${\Ld}(\mathcal{O})$:]
        Let $cur$ be the value the $\textsc{Read}(X)$ by $op$ returns in Line~\ref{lin:rllscf-read}, then $\mem{\alpha'} = cur$. 
        Since $op$ returns $cur.val$ and $\seqstate{h_{llsc}(\alpha')} = \seqstate{cur}$, $h_{llsc}(\alpha)$ is in the sequential specification. 
        Since the object's state and memory did not change, $\seqstate{h_{llsc}(\alpha)} = \seqstate{h_{llsc}(\alpha')} = \seqstate{\mem{\alpha'}} = \seqstate{\mem{\alpha}}$.
        
        \item[${\St}(\mathcal{O}, v)$:]
        Since $op$ always returns $true$, $h_{llsc}(\alpha)$ is in the sequential specification.
        Regardless of $\seqstate{h_{llsc}(\alpha')}$ and $\mem{\alpha'}$, 
        the \textsc{Write} in Line~\ref{lin:rllscf-write} writes the value $v$ with an empty $context$ to $X$. 
        Hence, by the sequential specification, $\seqstate{h_{llsc}(\alpha)} = \seqstate{\mem{\alpha}}$.
    \end{description}
\end{proof}

\scrlfinish*

\begin{proof}
    An {\RL} or {\SC} operation can either return after a successful {\CAS} operation in Line~\ref{lin:rl-cas} or Line~\ref{lin:sc-cas}, respectively, or after a $\textsc{Read}(X)$ that returns $cur$ such that $cur.context[\pi(op)] = 0$. Assume there is a linearization point of a successful context-resetting operation $op_{cr}$ after $op$ starts and before $op$ returns. Let $v$ be the value of $X$ right after the linearization point of $op_{cr}$.
    By Lemma~\ref{lem:rllscf-seq-specification}, for every $1\leq i \leq n$, $v.context[i] = 0$. The value of $context[\pi(op)]$ in $X$ can be changed to 1 only in an {\LL} operation by process $\pi(op)$, which can be invoked only after $op$ returns. 
    Hence, in a finite number of steps by $op$ after $lin(op_{cr})$, either $op$ performs a successful {\CAS} or a \textsc{Read} that identifies that $context[\pi(op)] = 0$ in $X$. In both cases, $op$ must return.
\end{proof}

\subsection{Wait-Free State-Quiescent HI Universal Implementation from Atomic CAS}
\label{app:universal-impl-wait-free}

% Next, we show the individual progress of each operation. 
% The next lemma shows that any {\RL} or {\SC} by $op$ returns after a finite number of steps by the operation. 

\begin{lemma}
\label{lem:uc-individual-progress}
    An ${\SC}(X)$ or ${\RL}(X)$, for any $X\in \braces{head, announce[1],\dots, announce[n]}$, performed by operation $op$, returns after a finite number of steps by $op$. 
\end{lemma}

\begin{proof}
    Assume ${\SC}(head, \wildcard)$ in Line~\ref{lin:sc-headA1}, Line~\ref{lin:sc-headB} or Line~\ref{lin:sc-headA2}, or ${\RL}(head)$ in Line~\ref{lin:rl-head} by $op$ never returns, despite $op$ taking an infinite number of steps.
    Since $op$ is state changing, 
    Lemma~\ref{lem:uc-global-progress} implies that there is an infinite number of mode transitions since $op$ starts the {\SC} or {\RL}. 
    Each mode transition is due to a successful ${\SC}(head, \wildcard)$. 
    Eventually, any successful {\SC} starts after the {\SC} or {\RL} by $op$ starts. 
    By Lemma~\ref{lem:sc-rl-finish}, the {\SC} or {\RL} by $op$ must return in a finite number of steps by $op$.
    
    Assume an ${\SC}(announce[j], \wildcard)$ in Line~\ref{lin:replace-inv-rsp} or ${\RL}(announce[j])$ in Line~\ref{lin:rl-announce}, $1\leq j \leq n$, by operation $op$ never returns, despite $op$ taking an infinite number of steps. By Lemma~\ref{lem:sc-rl-finish}, no context-resetting operation on $announce[j]$ that starts after $op$ starts the {\SC} or {\RL}, returns. This implies that process $p_j$ does not perform a ${\St}(announce[j],\wildcard)$ in Line~\ref{lin:announce-inv} that returns in a new operation after the {\SC} or {\RL} by $op$ starts.
    Thus, the only operations that can prevent $op$ to return are 
    ${\LL}(announce[j])$ in Line~\refl{lin:ll-announce} or ${\RL}(announce[j], \wildcard)$ in Line~\ref{lin:rl-announce} 
    that complete.
    
    By the code, operation $op'$ performs one of these lines if $op'$ performs a successful ${\LL}(head)$ in Line~\ref{lin:ll-head1} in state $B_i$, $i\geq 1$, and an operation by $p_j$ was linearized by the transition to state $B_i$.
    By Lemma~\ref{lem:uc-global-progress}, since $op$ is state changing,
    a finite number of steps after the {\SC} or {\RL} start, 
    the algorithm transitions from mode $B_i$. 
    Since no new operation starts by process $p_j$, 
    there is no future transition to mode $B$ that linearizes an operation by $p_j$. Hence, eventually, in a finite number of steps by $op$, no new context-changing operations are performed on $announce[j]$, and by Lemma~\ref{lem:rllscf-lock-free}, the {\SC} or {\RL} by $op$ returns.
\end{proof}

Once an operation is cleared, the wait condition parallel to the {\LL}'s in the loop in Lines~\ref{lin:start-loop}--\ref{lin:end-loop} no longer holds. Hence, along with the previous lemma, we have the next claim.

\ucclearfinitereturn*

\begin{proof}
     By Invariant~\ref{lem:cleared-inv}, $announce[\pi(op)]$ is cleared with a response from $R$ and only $op$ can write $\bot$ to $announce[\pi(op)]$ after exiting the loop in Lines~\ref{lin:start-loop}--\ref{lin:end-loop}.
     Thus, the waiting condition in  Line~\ref{lin:start-loop}, Line~\ref{lin:wait-til1} and Line~\ref{lin:wait-til2} does not hold after $op$ is cleared.
     By Lemma~\ref{lem:uc-individual-progress}, $op$ reaches Line~\ref{lin:start-loop}, Line~\ref{lin:wait-til1} or Line~\ref{lin:wait-til2} in a finite number of steps by $op$, and exits the loop.
     Until $op$ returns, no new operation by $\pi(op)$ starts, and by Lemma~\ref{lem:uc-global-progress}, in a finite number of steps by $op$ the algorithm mode transitions and $\pi(op)$ is not written in $head$ until $op$ returns. Thus, the waiting condition in Line~\ref{lin:wait-til3} no longer holds after a finite number of steps by $op$, and by Lemma~\ref{lem:uc-individual-progress}, $op$ returns.
\end{proof}

\end{document}